\documentclass[12pt,a4paper]{article}
\usepackage[unicode=true,linktoc=page]{hyperref}
\usepackage[utf8]{inputenc}
\usepackage[a4paper]{geometry}
\usepackage{amsmath,amsthm,mathbbol,eufrak,latexsym,MnSymbol} 
\usepackage{graphicx,color}
\usepackage[all,arc,line]{xy}
\providecommand\email[1]{\href{mailto:#1}{#1}}
\newcommand\abox[1]{\begin{minipage}{0.4\textwidth}\small\noindent\ignorespaces 
#1\end{minipage}}
\providecommand\addressi[1]{\def\theAddressi{\abox{\footnotemark[1]#1}}}
\def\theAddressi{}
\providecommand\addressii[1]{\def\theAddressii{\abox{\footnotemark[2]#1}}}
\def\theAddressii{}

\title{{\large{General Yang-Mills type gauge theories for $p$-form gauge fields:\\ 
{From physics-based ideas to a mathematical framework} \\ \emph{or} From Bianchi identities to twisted Courant algebroids}}}
\author{Melchior Gr\"utzmann\footnotemark[1]{}$\ $ and Thomas Strobl\footnotemark[2]}
\addressi{University of Jena \\
  Germany \\[1ex]
  e-mail: \email{melchiorG@gMail.com}
}
\addressii{Université Lyon 1 \\ 
 France  \\[1ex]
 e-mail: \email{Strobl@math.univ-lyon1.fr}
}
\date{July 17, 2014}

\newcommand\void[1]{}

\numberwithin{equation}{section}

\newtheorem{vdef}{Definition}[section]
\newtheorem{prop}[vdef]{Proposition}
\newtheorem{lemma}[vdef]{Lemma}

\newtheorem{thm}{Theorem}[section]
\newcommand\qedSymbol{\ensuremath{\Box}}
\renewenvironment{proof}[1][]{\par\noindent\textbf{Proof#1: }\newline\noindent\ignorespaces}{\hfill\qedSymbol\par\medskip}
\theoremstyle{definition}
\newtheorem{rem}[vdef]{Remark}

\newtheorem{example}[vdef]{Example}

\newcommand\defbb[2]{\newcommand#1{{\mathbb{#2}}}}
\newcommand\deffrak[2]{\newcommand#1{{\mathfrak{#2}}}}
\newcommand\defcal[2]{\newcommand#1{{\mathcal{#2}}}}
\newcommand\redefcal[2]{\renewcommand#1{{\mathcal{#2}}}}
\newcommand\defrm[2]{\newcommand#1{{\mathrm{#2}}}}

\def\[{\begin{equation}}
\def\]{\end{equation}}
\def\beq{\begin{equation}}  
\def\eeq{\end{equation}}  
\def\ba{\begin{eqnarray}}
\def\ea{\end{eqnarray}}
\def\bal{\begin{align}}  
\def\eal{\end{align}}
\def\be{\begin{equation}}
\def\ee{\end{equation}}
\def\<{(}
\def\>{)}
\def\1{\mathbb{1}}
\def\2{\frac12}
\def\:{\colon}
\def\A{\mathcal{A}}

\DeclareMathOperator\ad{ad}
\def\alt{{\text{alt.}}}
\def\b{\beta}
\def\c{\gamma}
\redefcal\C{C}
\DeclareMathOperator\CDO{CDO}
\def\conn_#1#2{\nabla_{\!#1\,}#2}
\def\smooth{\mathcal{C}^\infty}
\def\:{\colon}
\def\cycl{{\text{cycl.}}}
\defcal\CD{D}
\defrm\ud{d}
\def\udCE{{\ud_{\mathrm{CE}}}}  
\defrm\uD{D}
\def\D{\uD}
\def\ED{{{}^E\!{}\mathrm{D}}}
\DeclareMathOperator\Der{Der}
\let\Euler=\epsilon
\def\e{\varepsilon}
\def\ue{\epsilon}

\def\embed{\hookrightarrow}
\DeclareMathOperator\End{End}
\DeclareMathOperator\ev{ev}
\defcal\CF{F}

\def\lie{\g}
\def\Eco{\Gamma}
\def\Econn_#1{\Ena_{\!#1\,}}
\deffrak\g{g}
\deffrak\h{h}
\DeclareMathOperator\Hom{Hom}
\defcal\I{I}

\DeclareMathOperator\id{Id}

\def\Koszul{\Euler}  
\DeclareMathOperator\Lie{Lie}
\def\L_#1{\mathcal{L}_{#1\,}}
\def\laction{\circlearrowright}

\def\m{\mu}
\defcal\CM{M}  
\def\Mo{{\CM_1}}
\def\Mt{{\CM_2}}
\def\M{\CM}
\defbb\N{N}
\def\Ena{{{}^E{}\nabla}}
\def\Econn_#1{\Ena_{\!#1\,}}
\def\Wconn_#1{\Wna_{\!#1\,}}
\def\Wna{{{}^W\nabla}}
\def\EO{{{}^E{}\Omega}}
\DeclareMathOperator\Part{Part} 

\renewcommand\pmatrix[2]{\left(\begin{array}{#1}#2\end{array}\right)}
\def\pt{\mathrm{pt}}
\def\Qo{{Q_1}} 
\def\Qt{{Q_2}}
\def\R{\Real}
\defbb\Real{R}
\def\ER{{{}^E\!{}R}}
\DeclareMathOperator\rk{rk}
\providecommand\eqref[1]{(\ref{#1})}
\def\s{\sigma}
\def\x{\sigma}  
\def\S{\Sigma}
\def\smooth{{C^\infty}}
\def\susu{\twoheadrightarrow}
\def\lto{\longrightarrow}
\def\xto{\xrightarrow}

\DeclareMathOperator\Unsh{Un}
\defbb\V{V}  
\defbb\VV{V}  

\def\W{W}
\deffrak\X{X}
\def\XM{{\X_\bullet(\CM)}}
\def\Xv{{\X_\bullet^{\mathrm{vert}}(\CM)}}
\def\Xz{\X_0^{\mathrm{vert}}(\CM)}
\def\Xmo{\X_{-1}^{\mathrm{vert}}(\CM)}
\defbb\Z{Z}
\newcommand\zero{\item[0.]}

\newcommand\newpassage{4mm}

\begin{document}
\maketitle
\vskip-2ex
\centerline{\def\thefootnote{\fnsymbol{footnote}}%
  \theAddressi \quad \theAddressii \\
}
\medskip
\begin{abstract}
 Starting with minimal requirements from the physical experience with higher gauge theories, i.e.~gauge theories for a tower of differential forms of different form degrees, we discover that all the structural identities governing such theories can be concisely recombined into a so-called Q-structure or, equivalently, an $L_\infty$-algebroid. This has many technical and conceptual advantages: Complicated higher bundles become just bundles in the category of Q-manifolds in this approach (the many structural identities being encoded in the one operator $Q$ squaring to zero), gauge transformations are generated by internal vertical automorphisms in these bundles and even for a relatively intricate field content the gauge algebra can be determined in some lines only and is given by the so-called derived bracket construction. 

This article aims equally at mathematicians and theoretical physicists; each more physical section is followed by a purely mathematical one. While the considerations are valid for arbitrary highest form-degree $p$, we pay particular attention to $p=2$, i.e.~1-~and 2-form gauge fields coupled non-linearly to scalar fields (0-form fields). The structural identities of the coupled system correspond to a Lie 2-algebroid in this case and we provide different axiomatic descriptions of those, inspired by the application, including e.g.~one as a particular kind of a vector-bundle twisted Courant algebroid. 
\end{abstract}
\newpage
\tableofcontents
\section{Introduction and Results}
\subsection{Motivation}
These are basic and general considerations on higher gauge theories, i.e.~gauge theories where the standard connection 1-forms of Yang-Mills (YM) gauge theories are replaced by a whole tower of such gauge fields of different form degrees. Standard Yang-Mills theories are governed by some (finite-dimensional, quadratic) Lie algebra $\lie = \Lie(G)$, where $G$ is the  structure group of the underlying principal bundle. On the other hand, when one attempts to construct an interacting theory of 1-form gauge fields, various consistency conditions would lead one to a constraint that has to be satisfied by the interaction constants $C^a_{bc}$, namely\footnote{There is a second condition that one obtains: The constants $\eta_{ab}$ entering the ``free'' kinetic term in a functional is also constrained; it has to correspond to an ad-invariant scalar product.}
\begin{equation} \label{Jacobi}
C^a_{eb} C^e_{cd} + \cycl(bcd) = 0 \, .
\end{equation}
This constraint is easily recognized as the Jacobi identity of a Lie algebra. Thus, also in the other direction, one is led by itself to a Lie algebra $\lie$ which, in a second step, then can be related to a principal bundle as mentioned before. In this paper we want to follow this strategy for the much more involved context of gauge fields of form degrees up to some $p \in \N$, including 0-forms for completeness. 

At the place of \eqref{Jacobi} one then obtains a much more involved set of equations for the unknown parameters introduced in an ansatz for the multiple interactions. Due to the fact that we permit also the 0-forms, these are even differential equations. It is one of the goals of this paper to identify this coupled system of structural equations with an underlying mathematical structure which it corresponds to, in generalization of the incomparably simpler relation of \eqref{Jacobi} with a structural Lie algebra $\lie$, needed for the construction of an ordinary Yang-Mills theory. 

Another goal is, however, to provide one mathematical formulation of what ``consistency conditions'', often imposed in more physically oriented constructions, can mean. In particular, we aim at isolating a basic or minimal set of such requirements, which we believe that the physics community would usually want to have and which we found to be fulfilled in all the examples we know or are aware of. 

Two, three remarks are in place here in this context: First, there exists a very general and elegant approach to the \emph{deformation} of free gauge theories to interacting ones, given by a BRST-BV-formalism, cf.~\cite{BBH00} for a review. There one starts with just the kinetic terms of the fields that one wants to consider and the related, simple-to-find free gauge transformations. The result of the procedure provides the most general possible action that contains the original terms to lowest order and that preserves the ``number'' of the gauge symmetries, although not their form; e.g.~starting with $r$ 1-form fields $A^a$ and the free (abelian) action functional 
\begin{equation}\label{Skin}
S[A] = \int_\Sigma \eta_{ab} \, \ud A^a \wedge *\ud A^b \, ,
\end{equation}
one is led to the most general Yang-Mills functional for a quadratic Lie algebra $\lie$ of dimension $r$ (cf.~the above discussion including the previous footnote), the abelian curvature or field strength $\ud A^a$ being effectively replaced by its non-abelian generalization $F^a = \ud A^a + \frac{1}{2}C^a_{bc}A^b A^c$ in the above functional. There are several advantages of this procedure: except for possible problems of global nature or of convergence (one works in the setting of formal power series) the approach is completely systematic and extensive, providing the most general (formal) deformation moreover already up to trivial ones, that can be induced by changes of coordinates on the space of fields and that is captured by BRST-exact terms vanishing in the cohomological treatment; furthermore, the resulting theory already comes in its BV-form, thus ready for a subsequent quantization. There is, however, one decisive drawback of this formalism: The calculations are generically very technical and lengthy and the resulting constraints on the parameters of the deformation do not come with a mathematical interpretation. This does not pose a problem if they are of a comparatively simple form like in Equation \eqref{Jacobi} (accompanied by a second condition expressing the ad-invariance of $\eta$ in this case), where one retrieves a more mathematical interpretation rather easily. But it does so in a context like the present one, where the structural equations that one finds can easily fill a page (cf., e.g., results like in \cite{Bizdadea09,Bizdadea09b,Bizdadea10}).

Second, what is also transparent from the above YM example already, in the case of blowing up the number of equations by adding another permitted form degree, it may be useful and will prove so to separate equations that replace the defining condition of the Lie algebra, Eq.~\eqref{Jacobi}, from objects defined on it, like the ad-invariant tensor $\eta$, that is needed for the functional  (cf.~the discussion above) and thus the dynamics of the physical fields, but not yet for the underlying geometry of a principal bundle. Note that a metric $h$ on $\Sigma$ entered the functional (in terms of the Hodge duality operation $*$) and in more complicated theories the gauge transformations leaving invariant a functional can and will depend on $h$ also (cf.~\cite{DSM} for an example), while some important part of the gauge transformations, to be specified below, will prove to be independent of this additional structure. Both these observations lead us to focus on generalizations of the curvature or field strength in a first step, where neither $\eta$ nor $h$ enters. 

This leads us also to postpone questions of the global bundle structure to a separate investigation: one may want to find the generalization of the structural Lie algebra, that governs the theory already on a local level, in a first instance. 
This restricted focus permits one to work in a local chart on the spacetime or base manifold $\Sigma$, to replace a connection by a set of 1-forms (in the ordinary YM case, but generalized to include higher form degrees here) and to deal with the question of how to represent an (adequately generalized) connection on a possibly non-trivial (higher) bundle in a second step only.\footnote{In fact, this investigation appeared in the mean time already, at least on the arXiv, cf.~\cite{KS07}. The present article thus can be seen, not solely but also, as a possible introduction or motivation from the perspective of physics to \cite{KS07},
which has an essentially mathematical focus. We will describe the resulting global picture in subsection \ref{sec:phys} below, cf., e.g., Fig.~\ref{fig1}.} 

\subsection{Structure of the article}
The paper is organized as follows.  In the two subsections \ref{sec:phys} and \ref{sec:math} we summarize the main results of the present paper. 

In Section~\ref{s:Bian} we expand in detail on the first part of what is explained in subsections \ref{sec:phys} below. In particular, we draw the necessary mathematical conclusions from the physical requirements, which we believe to be minimal or mandatory for any higher gauge theory. The main result of these considerations is summarized in Theorem~\ref{thm:phys1} below. In Section~\ref{s:examQ} we unravel the geometric structure of Q2-manifolds, as they appear as what replaces the structural Lie algebra of an ordinary YM gauge theory when considering a higher gauge theory with 0-form, 1-form, and 2-form gauge fields. They are in equivalence with Lie 2-algebroids or 2-term $L_\infty$-algebroids. In Section~\ref{s:examQ} we provide, in particular, a description generalizing the one of a Lie 2-algebra as a crossed module, in the Appendix~\ref{s:Linf} we complement this by a descriptions in terms of multiple brackets as usual for $L_\infty$-algebras (cf., also, subsection \ref{sec:math} below).

In Section~\ref{s:gauge} we take another step in the construction of a ``physically acceptable'' higher gauge theory, focusing on the gauge symmetries in this context. The main result of this investigation is summarized in Theorems~\ref{thm:phys2} and \ref{theo:Lie}) below. Having been led to regard derived brackets in this context, we take another view on Q2-manifolds or Lie 2-algebras from this perspective in the subsequent Section~\ref{sec:derived}. This motivates the definition of what we call a V-twisted Courant algebroid. We derive its elementary properties in that section and provide sufficient conditions under which it is equivalent to a Lie 2-algebroid again. 

Appendix~\ref{s:Linf} reviews briefly the relation between Q-structures and $L_\infty$-algebras/ algebroids in general.  

\subsection{Chronology of the results}\label{s:hist}

Sections \ref{s:Bian}, \ref{s:examQ}, and \ref{s:gauge} were written already in 2005, Section  \ref{sec:derived} between 2005 and 2008 (as visible also by several talks of TS and a partial distribution of preliminary versions). The Introduction and the Appendix, as well as some fine-tuning of the remainder, is from the end of 2013/beginning of 2014, where we finally finished the article. For TS the results of this work as of 2005 were among the main motivations for starting his work with Alexei Kotov on characteristic classes and Q-bundles, the first results of which appeared as a preprint in 2007 on the arXiv and is now published right behind this article in the present journal \cite{KS07}. We attempted to leave the Sections  \ref{s:Bian}, \ref{s:examQ}, \ref{s:gauge}, and \ref{sec:derived} as much as possible in their original form from the time of 2005 to 2008, including an updated perspective only in the two subsections to follow (to which we then also moved important formulas and statements from those older sections, like Equation \eqref{ideal} and Theorem \ref{thm:phys1}, respectively).

For what concerns our work summarized in section  \ref{sec:derived}, we found out later that there were parallel developments arriving at a similar notion of a vector-bundle twisted Courant algebroid as the one given in the present paper. It is interesting to see that quite different motivations and starting points led to the same or similar objects of interest at about the same time.  In the following, we comment on the relation in more detail: Chen, Liu and Sheng have introduced $E$-Courant algebroids in \cite{CLS08}.  An $E$-Courant algebroid as defined by them is a $V$-twisted Courant algebroid as defined in this paper after identifying their vector bundle $E$ with our $V$ iff in addition their inner product is surjective to this vector bundle.  Conversely a $V$-twisted Courant algebroid is a $V$-Courant algebroid in their sense iff the anchor $\rho$ factors through $\mathfrak{D}V$, the derived module of $V$.  Similarly, Li-Bland introduced  an $AV$-Courant algebroid in \cite{LiB11}. This is a  $V$-twisted Courant algebroid in our sense and, conversely, every $V$-twisted Courant algebroid is an $AV$-Courant algebroid iff the anchor map factors through a Lie algebroid $A$ and the action of the twisted Courant algebroid on $V$ induces an action of $A$ on $V$.
  

\subsection{Acknowledgments}
We thank the TPI in Jena and in particular A.~Wipf for the support of our work in 2005 and 2013, the ESI in Vienna for support in 2007, as well as the ICJ in Lyon for likewise support in 2008 and 2013. 
We also thank the anonymous referee for helpful suggestions, such as a change of the subtitle. T.S.~thanks A.~Kotov for many valuable discussions during our stimulating parallel work. 

\subsection{Mathematical framework for physical theories} \label{sec:phys}
Let us now describe our approach in some detail and summarize the main results of the paper. As mentioned above, we consider a tower of differential forms. Let us denote them collectively by $A^\alpha$ and call them the ``gauge fields'' of the theory under investigation. So $A^\alpha$ collects differential forms of various form degrees, in general at each level of a different number: $A^\alpha = (X^i, A^a,B^D, \ldots)$, where $(X^i)_{i=1}^n$ are $n$ 0-forms (or scalar fields), $(A^a)_{a=1}^r$ are $r$ 1-forms (called, in a somewhat misleading way, ``vector fields'' in the physics literature), $(B^C)_{C=1}^s$ are $s$ 2-forms etc., up to some highest form degree $p$. At this stage of the consideration we do not want to attach any geometrical interpretation to those differential forms yet, treating them in a most elementary way possible in the first step and believing in the power of ``physically oriented'' argumentation, supposed to lead us to a geometrically interesting picture more or less by itself in the end. 

Next we make a most general ansatz for the ``field strengths'' $F^\alpha$ compatible with the form degrees and not using any additional structures (like a metric $h$ on $\Sigma$), since we believe that these objects, which in the end should represent in one way or another some generalized curvature(s), should not depend on anything else but the data specified up to here. Most physicists would agree, moreover, that $F^\alpha$ should start with the exterior derivative acting on the gauge fields, $F^\alpha = \ud A^\alpha + \ldots$, and that this should be complemented by terms, represented by  the dots, that are polynomial in the gauge fields of form degree 1 and higher, but, to stay as general as possible, with a priory unrestricted coefficient functions of the scalar fields $X^i$. This ansatz can be considered as still too restrictive, however. First, one may possibly want to replace the first term in $F^\alpha$ by a coefficient matrix $M^\alpha_\beta$---depending on the $X^i$s---times $\ud A^\beta$; second, one may want to permit also $\A^\beta \wedge \ud A^\gamma$-contributions to the remaining  terms indicated by the dots above. However, as we will argue in the main text, both of these apparent generalizations of our ansatz can be mapped in a first transformation to a situation of the simpler form above, provided only that the matrix $M^\alpha_\beta$ is invertible.\footnote{Correspondingly, also theories that have field strengths containing Chern-Simons type of contributions can be tackled by the present formalism and, e.g., the Theorem \ref{thm:phys1} below is applicable also in such situation. This is exemplified for instance in \cite{SLS14}.} Thus, anticipating this argument, without loss of generality we can thus assume the field strengths to have the form
\begin{equation} \label{curv}
F^\alpha = \ud A^\alpha + t^\alpha_\beta(X) A^\beta + \tfrac{1}{2} C^\alpha_{\beta \gamma}(X) A^\beta \wedge A^\gamma +\tfrac{1}{6} H^\alpha_{\beta \gamma\delta}(X) A^\beta \wedge A^\gamma \wedge A^\delta + \ldots ,
\end{equation}
where the coefficients or coefficient functions  $t^\alpha_\beta$,  $C^\alpha_{\beta \gamma}$, $H^\alpha_{\beta \gamma\delta}$, $\ldots$ are assumed to be fixed for defining the theory (on a local level and at this stage of constructing the theory). Their choice is now assumed to be constrained by some kind of ``physical consistency condition'', like one that would lead to \eqref{Jacobi} in the case of only 1-form gauge fields $A^a$. It is thus now our task and declared goal to specify one or several possible consistency requirements, motivated by a more or less physical argument, to display the resulting constraints on the coefficients in some explicit form, and, in a final step, to give them some algebraic-geometrical meaning so as to unravel what replaces the structural Lie algebra $\lie$ of a standard YM theory in this much more general setting. 

Before turning to these issues two more remarks on the nomenclature and the setting, however. We chose to include scalar fields (i.e.~0-forms) into that what we call gauge fields. In some sense this is unconventional, since they normally describe matter, while the conventional gauge fields (the 1-forms) represent the interaction forces. There are two reasons for this: First, it is suggested by the study of the Poisson sigma model (PSM) \cite{Ikeda94, Str94}: As argued in  \cite{Str04b} and \cite{KS07}, the PSM can be seen as a ``Chern Simons (CS) theory''  of the Lie algebroid $T^*M$ associated to a Poisson manifold $M$; like the CS-theory corresponds to the Pontryagin class, also the PSM, and more generally all AKSZ sigma models \cite{AKSZ}, correspond to characteristic classes (cf.~also \cite{FRS13}). In this context, all the tower of differential forms correspond to a single map or section of an adequate higher bundle, and it in particular also includes the 0-forms. 

The second reason for including the 0-forms is that even if one wants to consider the 0-forms $X^i$ as matter fields, it has advantages to include them right away.\footnote{In this case the ``1-form field strength'' $F^i \equiv \ud  X^i - \rho^i_a(X) A^a$ should be rather interpreted as a (generalized) covariant derivative. All the considerations concerning the ideal $\I$ below will apply also to those.} Suppose one considers an ordinary YM-connection with scalar fields being a section of an associated (vector) bundle induced by a representation $\rho$ of $\lie$
on a vector space or, more generally, by a $\lie$-action on a manifold $M$ that serves as a fiber of the associated bundle. In a local chart, the connection and the section correspond to a collection of 0-forms and 1-forms, $X^i$ and  $A^a$, respectively. On the other hand, $E:=\lie \times M\to M$ becomes a Lie algebroid in a canonical way and the data of the fields $X^i$ and $A^a$ are the same as a vector bundle morphism from $T\Sigma \to E$; moreover, the covariant derivative of the scalar fields and the curvature of the connection combine into a field strength map, cf.~\cite{Str04b,Str09,KS07}. There are many more Lie algebroids
(cf., e.g., \cite{Mack87,CaWe99})
 than those coming from a Lie algebra action $\rho$; if one intends to treat the 0-forms in a second step only, one misses all these additional possibilities. From a more physical perspective,  scalar fields are sometimes seen to deform the structural identities underlying the gauge fields of degrees at least one in a highly nontrivial way (cf, e.g., \cite{Str04b} or some supergravity theories). Thus also from this perspective it seems reasonable to not exclude a geometrical explanation of such structures (in terms of some Lie algebroids for example) from the outset by the restriction to a two-step procedure. 

The second remark concerns the nomenclature. We have chosen to call the $A^\alpha$-fields \emph{gauge fields} and the combinations $F^\alpha$ as in \eqref{curv} \emph{field strengths}, following a ``physics oriented'' wording. We in particular avoid to call them generalized connections and curvatures, as it may seem more adequate from a (superficial) mathematical perspective. The reason for this is the following. As already obvious from the previously given example of the PSM, that what replaces the structural Lie algebra $\lie$ (as well as its integration $G$) of an ordinary principal bundle can and in general will be a bundle itself, in fact, even some kind of higher bundle. So, the fibers of what replace principal and associated bundles will be bundles themselves (and it may be useful  to choose auxiliary connections for technical reasons, that must be distinguished strictly from the ``dynamical'' gauge fields $A^\alpha$). But even worse, a Poisson manifold for example does not always give rise to an integrating Lie groupoid, there can be obstructions in the integration process \cite{CrF01}, while still the above mentioned characteristic classes as well as the PSM do exist as meaningful objects for any Poisson manifold $M$. To cover all these examples in a common mathematical framework it thus is indicated to work with bundles (over the base $\Sigma$ or its tangent space $T\Sigma$) where one does not need the analogue of the structural Lie group $G$, but it is sufficient to work with its Lie algebra $\lie$. In the ordinary setting this is the so-called Atiyah algebroid $A\to T\Sigma$ associated to a principal bundle $P\to \Sigma$, where $A$ can be identified with $TP/G$ and, as a bundle over $T\Sigma$, its fiber is isomorphic to $\lie$ only (cf.~\cite{KS07} for further details). Now a connection in $P$, i.e.~an ordinary gauge field in physics language, corresponds to a \emph{section} in $A$, to be distinguished from connections in $A$, that exist as well as on any bundle. (The generalization of $A$ will be what was called a $Q$-bundle in \cite{KS07}.) These potential sources of a possible confusion with other natural connections lead us to avoid calling $A^\alpha$ a generalized connection.

We now come back to our principal goal, finding conditions for the restriction of the coefficients appearing in \eqref{curv}. The approach nearest at hand may seem to look at gauge transformations and to require some generalized form of invariance of the field strengths. This, however, has a decisive disadvantage: There is a lot of freedom in defining gauge transformations, and, as examples show (cf., e.g., the Dirac sigma model (DSM) \cite{DSM}), they may depend on several additional structures not needed for the definition of the above field strengths.

The main idea for the departure is, therefore, to use a generalization of the Bianchi identities. Normally, they would contain a covariant derivative, that we have not yet defined as well, not yet knowing the adequate algebraic-geometric setting to use. Not wanting to use any additional new parameters/coefficients, we thus will want to formulate a minimalistic version of Bianchi that at least most people in the field would agree with. Regarding the standard Bianchi identity for an ordinary YM theory, $\ud F^a + C^a_{bc}A^b\wedge F^c = 0$, we note that the application of the exterior derivative on the field strength becomes ``proportional'' to the field strength itself, albeit with a field dependent ``prefactor''. 

Mathematically we express this as follows: We require that the exterior derivative $\ud$ applied to any monomial containing a field strength $F^\alpha$, i.e.~terms of the form $\ldots \wedge  A^\alpha \wedge F^\beta \wedge A^\gamma \wedge \ldots$, yields expressions that again contain at least one field strength. In other words, we regard the ideal $\I$ generated by expressions of the form \eqref{curv} within the (graded commutative, associative) algebra $\A$ generated by abstract elements $A^\alpha$ and $\ud A^\alpha$. The condition we require is then simply
\beq \boxed{\ud \I \subset \I } \, ,\label{ideal} \eeq
i.e.~$\I$ has to be what is called a differential ideal. As innocent and weak as this condition may seem, it turns out to restrict the coefficient functions of \eqref{curv} already very drastically, at least if the dimension $d$ of spacetime $\Sigma$ exceeds the highest permitted form degree $p$ of $A^\alpha$ by two, i.e.~$d \geq p+2$. In that case, we will find that \eqref{ideal} holds true iff the following vector field squares to zero:
\begin{equation} \label{Q1}
Q=  \left(t^\alpha_\beta(x) q^\beta + \frac{1}{2} C^\alpha_{\beta \gamma}(x) q^\beta q^\gamma +\frac{1}{6} H^\alpha_{\beta \gamma\delta}(x) q^\beta  q^\gamma q^\delta + \ldots\right) \dfrac{\partial}{\partial q^\alpha} .
\end{equation}
Here we introduced for each gauge field $A^\alpha$ an abstract coordinate $q^\alpha$, carrying a degree that equals the form degree of  $A^\alpha$ (for the degree zero variables we also use $q^i \equiv x^i$ to distinguish them from the others, since they can appear non-polynomially).
These coordinates commute (anticommute) iff the respective differential forms commute (anticommute). Also, since by its definition, $F^\alpha = \ud A^\alpha + \ldots$ has a form degree that exceeds the one of $A^\alpha$ by precisely one, the vector field $Q$ is an odd vector field of degree $+1$. In general, the condition
\beq\label{Qsquare}
Q^2 \equiv \tfrac{1}{2} [Q,Q] = 0
\eeq
subsummarizes a whole list of equations and in particular reduces precisely to \eqref{Jacobi} in the case of coordinates of only degree one. 

This observation turns out to be rather easy to prove, but so general that we consider it worth being called a fundamental theorem in the context of higher gauge theories:
\begin{thm} \label{thm:phys1} If $\dim \Sigma \geq p+2$ the generalized Bianchi identities \eqref{ideal} hold true if and only if the vector field \eqref{Q1} squares to zero, $Q^2=0$.
\end{thm}
Thus any higher gauge theory is associated to what is called a differential graded manifold or a Q-manifold  
$\CM$. 
(In the case of an ordinary YM theory this corresponds to a reformulation of the Lie algebra $\lie$ in terms of its Chevalley-Eilenberg complex; in this case, simply $ (C^\infty(\CM), Q) \cong (\Lambda^\cdot \lie^*, \udCE)$). The most important part of the above theorem is not that a Q-manifold, as a generalization of a Lie algebra, gives rise to some generalized Bianchi identities, but the reverse direction: Even if Bianchi identities are formulated in as weak a version as (\ref{bian}), one \emph{necessarily}
deals with a Q-manifold on the target if one observes it or not. This is the strength of the above observation. 

Together with \cite{BKS,KS07} (and the considerations on the gauge symmetries to follow) this paves the way for an elegant, general, and simple picture of higher gauge theories. Their field content (and symmetries) is compactly described by a section $a$ in a Q-bundle (cf.~Fig.~\ref{fig1} below):
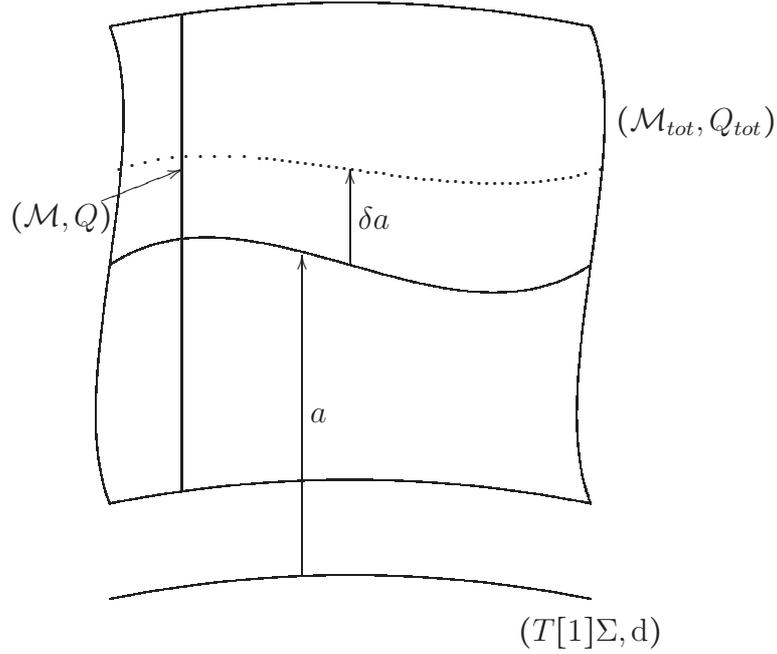
\begin{figure}
$$\begin{xy}/r0.15pc/:
  (-50,-70)*{}; (50,-70)*{} **\crv{(0,-60)}+(0,-7)*{
  (T[1]\Sigma,\ud)};
  (-50,0)*{}; (50,0)*{} **\crv{(-20,20) & (20,-20)};  {\ar_{\textstyle a}(-10,-65)*{};(-10,2)*{}};
  (-50,-50)*{}="A"; (50,-50)*{} **\crv{(0,-40)};    (72,30)*{(\M_{tot},Q_{tot})};
  (50,-50)*{}; (50,50)*{} **\crv{(40,-25)&(60,25)};
  (50,50)*{}; (-50,50)*{} **\crv{(0,60)};
  "A"; (-50,50)*{} **\crv{(-60,-25)&(-40,25)};
  (-35,-47.5)*{};(-35,52.5)*{} **\dir{-};
  {\ar(-60,10)*{(\M,Q)};(-35,20)*{}};
  (-48,20)*{}; (52,20)*{} **\crv{~*=<4pt>{.} (-20,30) & (20,10)};  {\ar_{\textstyle \delta a}(0,0)*{};(0,20)*{}};
\end{xy}
$$
\caption{\label{fig1}  The resulting geometry governing all higher gauge theories: The generalization of a connection, called a \emph{gauge field} in this work, corresponds to a section $a$ of a graded bundle $\CM_{tot}$. The bundle is a Q-bundle and (the principal part of) gauge transformations correspond to vertical inner automorphisms. The section is a Q-morphism iff the generalization of curvature, called the \emph{field strength} here, vanishes. In physical theories, however, it is almost exclusively those gauge fields which have a non-vanishing field strength that, up to gauge invariance, are of interest.}
\end{figure}
The tangent bundle $T\Sigma$ of spacetime $\Sigma$ can be viewed as a graded manifold when fiber-linear coordinates are regarded as carrying degree 1, functions on this graded manifold $T[1]\Sigma$ correspond to differential forms, and the nilpotent, degree 1 de Rham differential $\ud$ equips it with a Q-structure. The Q-manifold $(\CM,Q)$ that we introduced above in the context of (\ref{Q1}) can be viewed as a target if we interpret the tower of gauge fields $A^\alpha$ as corresponding to a degree-preserving map from $T[1]\Sigma$  to 
$\CM$.\footnote{We will explain all this in more detail in the main text, but we want to give an overview of the resulting picture already as this point as a guide for what we are going to obtain or derive as a final framework.} As a straightforward generalization of the considerations in \cite{BKS} this can be reinterpreted as a section $a$ in a trivial Q-bundle, $(\CM_{tot},Q_{tot}) := (\CM \times T[1]\Sigma, Q+\ud)$,  an observation that proves essential in the context of describing the gauge symmetries (in \cite{BKS} for the PSM, here for general higher gauge theories). One of the observations fundamental to \cite{KS07} is, on the other hand, that this picture generalizes also to the twisted case, i.e.~that a connection in an ordinary principal bundle corresponds just to a section of a (possibly non-trivial) Q-bundle. 

We add a remark in parenthesis: The section $a \colon T[1]\Sigma \to \CM_{tot}$, corresponding to our gauge fields, is a section of the category of ($\N_0$-)graded manifolds, and, not, in the sense of the category augmented by a Q-structure. It becomes an ``honest'' section of the category of Q-bundles, i.e.~a Q-morphism, iff all the field strengths $F^\alpha$ vanish. As nice as this may seem from a purely mathematical perspective, as inadequate it is from the one of physics: It would exclude all the particles (like photons, gluons, etc) for which gauge fields were introduced in physics in the first place. In fact, even in topological models, such as the PSM, the DSM, or the AKSZ sigma models, the gauge fields become Q-morphisms only ``on-shell'', i.e.~on the level of the Euler-Lagrange equations of the corresponding action functionals, cf.~\cite{BKS,DSM,KS08} for the proofs, respectively.

Let us now address the gauge symmetries. Similarly to the general ansatz (\ref{curv}), we make one for the gauge symmetries
\begin{equation}\label{gaugedelta0}
\delta^{(0)}_{\epsilon} A^\alpha = \ud \epsilon^\alpha +  \left(s^\alpha_\beta(X)  + D^\alpha_{\gamma \beta}(X) A^\beta  +\tfrac{1}{2} G^\alpha_{ \gamma\delta\beta}(X) A^\gamma \wedge A^\delta + \ldots \right)\wedge \epsilon^\beta,
\end{equation}
for some collection $\epsilon^\alpha$ of differential forms of degree up to  $p-1$, where we may allow the total gauge transformations to also 
contain a part proportional to $F^\alpha$ again:
\begin{equation} \label{gaugedelta}
\delta_{\epsilon} A^\alpha = \delta^{(0)}_\epsilon A^\alpha + O(F^\alpha) \, . 
\end{equation} 
It is important to remark that the terms $O(F^\beta)$ have to be in the ideal $\I$ for the following to be still true. This still includes higher gauge theories as they appear as a part of gauged supergravity for example, cf.~\cite{SSW11,SLS14,SKS14}, and where they play an important role, but excludes terms containing e.g.~$*F^\alpha$ (as they appear in the DSM \cite{DSM}) or terms containing ``contractions'' of $F$ (with the inverse of the part of $A$ that may describe a vielbein) as they typically appear in the gravitational part of supergravity theories. In analogy to (\ref{ideal}), we now require that also our gauge transformations do not leave the ideal 
$\I$ generated by the field strengths,
\beq \boxed{\delta_\varepsilon \I \subset \I } \, .\label{deltaI} \eeq
The result we obtain, illustrated and proven in detail in the main text, is the following (cf.~also Fig.~\ref{fig1})
\begin{thm} \label{thm:phys2}
Provided $\dim \Sigma \geq p+1$ and $Q^2=0$, infinitesimal gauge transformations of the form (\ref{gaugedelta}) leave the field strength ideal $\I$ invariant, Eq.~\eqref{deltaI}, iff their principal part $\delta^{(0)}_\varepsilon$ corresponds to vertical inner automorphisms of the Q-bundle $(\CM_{tot},Q_{tot})$. Moreover, if the infinitesimal gauge transformations are of the more restricted form (\ref{gaugedelta0}), $\delta_\varepsilon \equiv \delta^{(0)}_\varepsilon$, one does not need to require $Q^2=0$, but one can conclude it together with the generalized Bianchi identities (\ref{bian}). 
\end{thm}
Let us briefly explain what inner automorphisms of a Q-manifold should be thought of to our mind. An automorphism should be degree preserving diffeomporphism that leaves invariant the vector field $Q$. Infinitesimally this corresponds to degree zero vector field annihilated by $\ad_Q$. It is an immediate consequence of (\ref{Qsquare}) that $\ad_Q$ squares to zero as well. We call an inner automorphism of a Q-manifold one that is generated by the image of $\ad_Q$, i.e.~by vector fields of the form  $\ad_Q(\epsilon)$, where $\epsilon$ is a degree -1 on the Q-manifold. If the Q-manifold is an ordinary Lie algebra (shifted in degree by one) $\lie[1]$, this corresponds to the usual definition of inner automorphisms, i.e.~to  the adjoint action of the integrating group $G$ on $\lie=\Lie(G)$ (generated infinitesimally by the left action on itself). For $(T[1]\Sigma, \ud)$ this definition of inner automorphism gives infinitesimally precisely the Lie derivatives of vector fields on $\Sigma$, so they correspond to the canonical lift of a diffeomorphism of $\Sigma$ to its tangent bundle. 

Gauge transformations of an ordinary YM-theory are generated by vertical automorphisms of the underlying principal bundle $P$. They correspond precisely to inner vertical automorphisms in the above sense for the associated Q-bundle $T[1]P/G$. The advantage of the present formulation is that like this it generalizes in a straightforward way to \emph{all} higher gauge theories. 

Some further explanations are in place: First, additional structures entering a particular theory can certainly constrain the symmetries permitted in the respective case. For example, in the Poisson sigma model or, more generally, the AKSZ sigma models, the inner automorphisms also have to preserve the symplectic form defined on the fibers and that are compatible with the Q-structure on the fibers. For the PSM, for example, this implies that its gauge symmetries are generated infinitesimally by the Lie subalgebra $\Omega^1_{closed}(M)$ of $\Gamma(T^*M)$ ($T^*M$ viewed upon as the Lie algebroid of the Poisson manifold $M$ in this case). Second, there are cases where the gauge transformations including the additional terms of the form $O(F)$ in (\ref{gaugedelta}) can be given a similarly good geometrical interpretation. For example, as will be shown in \cite{SLS14}, for the higher gauge theory sector of the (1,0) superconformal model in six spacetime dimensions \cite{SSW11}, the $O(F)$-part is essential for an off-shell closed algebra of gauge transformations on the one hand. On the other hand, in this case the gauge transformations (\ref{gaugedelta}) correspond to nothing but inner automorphisms of another Q-bundle associated to the one in Fig.~\ref{fig1}, where the typical fiber $(\CM,Q)$ is replaced by $(T[1]\CM, \ud_{\CM} 
+ \L_Q)$ (cf.~also \cite{KS07,SaS13}).\footnote{Any Q-bundle $(\CM_{tot},Q_{tot}) \to (\CM_1,Q_1)$ can be lifted to a tangent Q-bundle $(T[1]\CM_{tot}, \ud_{\CM_{tot}} 
+ \L_{Q_{tot}}) \to (T[1]\CM_1, \ud_{\CM_1} + \L_{Q_1})$. Using the canonical 0-section $\sigma \colon \CM_1\to T[1]\CM_1$, we can pull back $T[1]\CM_{tot}$ to become a bundle over $\CM_1$, which turns out to be compatible with the respective Q-structures
.
}


Third, it is important to be aware that $Q_{tot}$ projects to $\ud$ on the base of the bundle, and it is like this that $\ud \varepsilon^\alpha$ is generated in (\ref{gaugedelta0}), (\ref{gaugedelta}); this becomes possible only by going to the picture of a bundle, even in the case of a local treatment, where the bundle (at least when defined as in \cite{KS07}) is trivial. 

In the present description, the vector fields generating the gauge transformation always form a Lie algebra. They are vertical vector fields of degree zero of the form $\ad_{Q_{tot}}(\epsilon)$ and they close with respect to the usual commutator of vector fields. They are parameterized by vertical vector fields $\epsilon$ of degree minus one \emph{up to $\ad_{Q_{tot}}$-closed} contributions. On the vector fields $\epsilon$ one discovers that one has an induced bracket of the form 
\begin{equation}\label{derQtot}
[\epsilon,\epsilon']_{Q_{tot}} \equiv [\ad_{Q_{tot}}(\epsilon), \epsilon'] \, ,
\end{equation} called the ``derived bracket'' (following Cartan and \cite{YKS03en}). For example, if one considers the Q-manifold $(\lie[1], \udCE)$, the derived bracket just reproduces the Lie algebra bracket. While in general, however, this bracket is a Leibniz--Loday algebra only and not necessarily antisymmetric, it becomes so for the coexact elements, i.e.~for those $\epsilon$ where one considers the quotient by closed contributions (which do not change the gauge transformation, as argued above). In this way one obtains the Theorem  \ref{theo:Lie} in the main text.\footnote{To simplify the notation in that section, we called the total Q-structure simply $Q$ and renamed the Q-structure on the fiber by $Q_2$.}

In a conventional treatment of gauge symmetries (i.e.~along the lines as described e.g.~in \cite{HT}), they form a so-called \emph{open} algebra for a generic higher gauge theory. This is not in contradiction to the present, to our mind more elegant formulation of its gauge symmetries. To resolve the apparent paradox, it is useful to consider the PSM---as a simpler toy-model, showing already all the essential features of this behavior. In \cite{BKS} it is shown, in particular, that one can obtain the usual, only on-shell closed algebra of gauge symmetries of the PSM from a incomparably shorter calculation in the above-mentioned framework. This generalizes to the context of any higher gauge theory. For instance for a higher gauge theory containing 0-forms $X^i$, 1-forms $A^a$, and 2-forms $B^C$ one obtains in this way the result of Proposition \ref{prop:onshellclosed} below. We also will comment further on this issue in \cite{SLS14,GSS14}. 
 
While one can still discuss the usefulness or necessity  of using the exterior derivative $\ud$ in vector analysis or $ \udCE$ for Lie algebras, there will be no doubt in the case of $p \ge 2$; already for $p=2$ there is, in general, a whole list of coupled differential equations satisfied by the structural quantities, cf.\ Eqs.~\eqref{B1}--\eqref{B7} in the main text below, which are to be used in almost every calculation related to such a gauge theory and which can be replaced and subsummarized in a single operator $Q$ squaring to zero. The technical usefulness becomes particularly obvious in the above-mentioned context of the calculation of the commutator of gauge symmetries, but also already for the concrete form of the Bianchi identities (cf., e.g., the derivation of Eqs.~\eqref{BFi}--\eqref{BFB} in the main text below).

Let us briefly come back to Fig.~\ref{fig1} and the corresponding two theorems, Thm.~\ref{thm:phys1} and \ref{thm:phys2}. Not only the base of the Q-bundle is in fact a bundle in ordinary differential-geometric terms, but even worse so for the fibers: the Q-manifold $(\CM,Q)$ isomorphic to the  typical fiber is generically a rather intricate (higher) bundle structure itself. We will come back to this issue in more detail in subsection \ref{sec:math} below. What we consider important in this context is that the infinite-dimensional space of gauge fields is described as sections in \emph{finite-dimensional} (graded) bundles. This gives additional mathematical control. The situation can be compared with the interest in Lie algebroids (for a definition cf., e.g., \ref{def:Loid}): While for infinite-dimensional Lie algebras topological questions usually play an important role and are intricate to deal with, the situation simplifies greatly when this infinite-dimensional Lie algebra comes from the sections of a Lie algebroid. 

Likewise, the above operator $Q$ (or also $Q_{tot}$) must not be confused also with a BV- or BRST-operator. Although the BV-BRST operator $Q_{BV}$  also squares to zero, $Q_{BV}^2=0$,  and is of degree +1, it is defined over the \emph{infinite-dimensional} graded manifold of (gauge) fields, ghosts, and their antifields, and, more importantly, it, in general, needs many more structures to be defined, namely those needed for an action functional that we have not yet addressed much up to here (and for which there are many options still open).\footnote{Correspondingly, also the discovery of the role of the derived bracket in the context of the gauge symmetries, cf.~Thm.~\ref{theo:Lie}, is not or at most vaguely related of the role of the derived bracket in the BV-context as observed in \cite{YKS97}.}

There are very particular models of topological nature, like the Poisson sigma model or, more generally, the AKSZ-sigma models \cite{AKSZ}, on the other hand, where there \emph{is} a close relation of the two Q-structures $Q$ and $Q_{BV}$: $Q$ corresponds to a defining Q-structure on the target of the sigma model, which is also equipped with a compatible symplectic form $\omega$, $\ud$ is a Q-structure on its source $T[1]\Sigma$, all very much like in Fig.~\ref{fig1}. Then $a \colon T[1]\Sigma \to \CM$ is no more degree preserving (so as to capture ghosts and antifields), but a general supermap and $\Omega_{BV}$ is, roughly speaking, the ``difference'' of $Q$ and $\ud$ induced on this space of supermaps $\CM_{BV}$. Using $\omega$ and a canonical measure on $T[1]\Sigma$, one then induces also a (weakly) symplectic form $\omega_{BV}$ on $\CM_{BV}$ and shows that $Q_{BV}$ is Hamiltonian, i.e.~$Q_{BV}=\{ S_{BV}, \cdot \}_{BV}$, where $S_{BV}$ is the BV-action, extending the classical action of the PSM and AKSZ-sigma model, respectively (cf., e.g., 
\cite{AKSZ,Catt01,Royt06}).

This changes quite drastically, however, in the case of non-topological, physically relevant theories; already for an ordinary YM gauge theory the BV operator is not constructed in such a simple way (unfortunately). Even for \emph{relatively} simple generalizations like (for the PSM) the topological DSM \cite{DSM} or (for YM) the non-topological Lie algebroid Yang-Mills theories \cite{Str04b}, the BV operator is even not known at the present time; however, in both cases the target Q-structures are very well known, corresponding to Lie algebroids in both cases. And, again, the field content for any such a theory, and, more generally, for any higher gauge theory, its described compactly by sections in the finite-dimensional Q-bundle as sketched in Fig.~\ref{fig1}. 

Although originally intended to do so, we will not address the construction of action functionals within this paper, leaving it for eventual later work (but cf.~also \cite{Str04b,KS08,SLS14}).

\subsection{Further mathematical results} \label{sec:math}
Here we briefly summarize the ideas and results contained in the more mathematical sections \ref{s:examQ} and \ref{sec:derived}. It mainly consists in reformulating supergeometrically defined objects appearing in the Q-bundles as potential fibers in 
more conventional algebraic-geometrical terms. 

Our graded or supermanifolds, on which the homological vector fields $Q$ are defined, always will be $\N_0$ graded manifolds, being simultaneously supermanifolds by means of the induced $\Z/2\Z$-grading. These are conventionally called simply
 N-manifolds and are defined as follows (for further details see also \cite{Royt02,Gru09}): An N-manifold is a ringed second countable Hausdorff topological space $\M=(M,\smooth)$ where the structure sheaf, which by abuse of notation we also denote as $\smooth$, is locally of the form smooth functions in tensor product with the exterior algebra in the odd coordinates tensor product with the polynomial algebra in the even coordinates not of degree 0. This implies that it 
 consists of a tower of affine fibrations.  In degree 0 we just have a smooth manifold $M$.  In degree 1, also after truncation of a higher degree N-manifold, we have an odd vector bundle $E[1]$, i.e.\ the fiber-linear coordinates of $E$ are declared as functions of degree 1.  Their parity is odd and thus the function algebra are the $E$-forms.  The topological structure of such an N-manifold comes from the smooth base only.  

Morphisms of N-manifolds are morphisms of ringed spaces that preserve the $\N_0$-grading of the structure sheaf.  Coordinate changes are particular local isomorphisms. Like this one finds e.g.~that the coordinates of degree 1 transform as the coordinates of a vector bundle so that any N-manifold $\CM$  with degrees up to 1, i.e.~an N1-manifold,\footnote{In general, we call an N$k$-manifold an N-manifold, where the highest coordinate degree is $k$. \label{fn:Nk}} can be canonically identified with vector bundles $E\to M$, $\CM \cong E[1]$.  

In degree 2, coordinate changes contain, however, also affine terms beside the terms one would expect for a vector bundle.  Globally an N2-manifold $\M$ is an affine bundle over the base $E[1]$ modeled after the pullback $\pi_1^*V[2]$ of a vector bundle $V\to M$ where $\pi_1\colon E[1]\to M$.  Since the fibers $V_x$ for $x\in M$ of this bundle are contractible, there exists a (non-canonical) global section (splitting) that identifies the affine bundle $\M$ with the graded vector bundle $E[1]\oplus V[2]$ over the base $M$. In all of our theorems on Q2-manifolds, i.e.~N2-manifolds equipped with a degree +1 homological vector field, we will need such a splitting permitting us to identify the underlying graded manifold with the direct sum of two vector bundles $E\to M$ and $V\to M$, $\CM \cong E[1]\oplus V[2]$. 

As mentioned already in the previous subsection, the definition of an ordinary Lie algebra $\lie$ is equivalent to a Q1-manifold over a point, i.e.~with coordinates of degree 1 only. The corresponding N1-manifold gives the vector space (vector bundle $E$ over a point), $\CM \cong \lie[1]$ (as graded vector spaces), the functions $C^\infty(\CM)$ are identified with elements of $\Lambda^\cdot \lie^*$, and the degree +1 vector field $Q$ corresponds precisely to the Chevalley-Eilenberg differential, $Q\cong \udCE$, from which one certainly can also reconstruct the Lie bracket on the vector space. 

A general Q1-manifold can be seen to be equivalent to a Lie algebroid, as was first observed in \cite{Vai97}. For completeness, we briefly provide a definition:
\begin{vdef} \label{def:Loid} A \emph{Lie algebroid} consists of a vector bundle $E\to M$ together with a bundle map $\rho \colon E \to TM$ (over the identity) and a Lie algebra $(\Gamma(E),[\cdot,\cdot])$ on the sections of $E$ such that for all $\psi_1$, $\psi_2 \in \Gamma(E)$ and $f \in C^\infty(M)$ one has the following Leibniz rule
\begin{equation} \label{Leib}
[\psi_1, f \psi_2] = f [\psi_1,  \psi_2] + \rho(\psi_1)f \,\cdot  \psi_2 \;.
\end{equation}
\end{vdef}
These data can be recovered from a homological vector field of degree 1 on an N1-manifold. In local coordinates $x^i$ and $\xi^a$ of degree 0 and 1, respectively, it always is of the form 
\beq \label{Q0} Q = \rho^i_a(x) \xi^a \frac{\partial}{\partial x^i}
- \2 C^a_{bc}(x) \xi^b \xi^c \frac{\partial}{\partial \xi^a} \, . \eeq
The axioms can be reconstructed
by using $Q^2=0$ \emph{and} the behavior of the coefficient functions $ \rho^i_a$ and 
$C^a_{bc}$ with respect to (degree preserving) coordinate changes on $\CM \cong E[1]$ (so as to recover also \eqref{Leib}). The details of this process can be found in the main text. It is also true, moreover, that the Q-morphisms (morphisms of N-manifolds such that their pullback intertwine the two Q-structures) correspond precisely to Lie algebroid morphisms in the sense of \cite{MaH90} (cf.~\cite{BKS} for an explicit proof). 
We summarize this in 
\begin{thm}[Vaintrob 97] \label{Vain}
A Q1-manifold is in 1:1-correspondence with a Lie algebroid. This is an equivalence of categories.
\end{thm}
In the present paper we treat the analogue question for one degree higher, i.e.~we address the description of a Q2-manifold. Instead of two coefficient functions $\rho$ and $C$ satisfying two equations resulting from $Q^2=0$, one ends up with five coefficient functions satisfying seven coupled differential equations, cf.~Eqs.~\eqref{Q} and \eqref{B1}--\eqref{B7}, respectively, below. For an orientation we first collect related results, corresponding to special cases.

If the Q2 manifold is one over a point, it corresponds to a (semistrict) Lie 2-algebra or, equivalently, to a 2-term $L_\infty$-algebra. A special case of this are strict Lie 2-algebras, which, according to \cite{Baez03vi}, are equivalent to Lie algebra crossed modules. We thus here define a Lie 2-algebra by means of this description:
\begin{vdef} A \emph{strict Lie 2-algebra} \label{strictL2}
is a pair $(\lie,\h)$ of Lie algebras 
together with a homomorphism $t \colon \h\to \lie$ and a  representation
$\alpha \colon \lie\to\Der(\h)$ such that  $\, t(\alpha (x)v) =[x,t(v)]$   
and $ \alpha(t(v))w =[v,w]$  hold true for all $x \in \lie$, $v,w \in \h$.  
\end{vdef}
We remark in parenthesis that this definition can be shortened:
\begin{lemma} A strict Lie 2-algebra $(\lie,\h)$  is the same as an intertwiner, $t \colon \h \to \g$, from a $\g$-representation $\h$ to the adjoint representation on $\g$ such that $t(v)\cdot v = 0$ for all $v\in \h$.
\end{lemma}
The data of Definition \ref{strictL2} can be retrieved easily with all their properties; e.g.~the Lie algebra structure on $\h$ can then be \emph{defined} by means of $[v,w]_\h := t(v) \cdot w$. 

The definition of a semistrict Lie 2-algebra or simply \emph{a Lie 2-algebra} is analogous, 
but contains an additional map $H \colon \Lambda^3 \lie \to \h$, satisfying a ``coherence-type'' consistency relation. $H$ turns out to control the Jacobiators of brackets on $\g$ and $\h$ and enters a modification of the definition above in several places. For $H=0$ one recovers a strict Lie 2-algebra. (In fact, one may extract its precise definition from specializing the Definition \ref{def:Lie2d} below to the case of $M$ being a point.)

Alternatively, one can describe a Lie 2-algebra also by means of a 2-term $L_\infty$-algebra. We recapitulate its definition in detail in the Appendix \ref{s:Linf1} (cf.~also \cite{SHLA,SHLA2} for the original definitions). Essentially this consists of a 2-term complex $\h[-1] \xto{t} \g$,\footnote{Note that in the context of N-manifolds, $\h[-1]$ means that the linear coordinates are of degree -1, which implies that the \emph{elements} of this vector space are of degree +1. The shifting is slightly changed with respect to the previous occurrence of $\g$ and $\h$, but this is due to different but equivalent conventions, cf.~also Appendix \ref{s:Linf1}.} a unary bracket $[\cdot]_1=t$ of degree -1, a binary bracket $[\cdot,\cdot]_2$ of degree 0 and a ternary bracket $[\cdot,\cdot,\cdot]_3$ of degree 1 
that have to fulfill several quadratic relations which can be understood as generalized Jacobi identities. One of those is, e.g., for $x,y,z\in\g$
\begin{equation}
 [x,[y,z]_2]_2+[z,[x,y]_2]_2 +[y,[z,x]_2]_2 = [\,[x,y,z]_3]_1 \, ,
\end{equation}
showing that the 2-bracket forms a Lie algebra iff the image of the 3-bracket lie in the kernel of $t$. It is thus not surprising that this 3-bracket is the above mentioned map $H$.  
It is thus again a non-vanishing $H$ that impedes the appearance of ordinary Lie algebras and that controls their deviation from being so on the other hand. In particular, $H$ itself has to fulfill a Jacobi-type quadratic relation. Further details on this perspective of a Q2-manifold over a point and $L_\infty$-algebras in general are deferred to Appendix \ref{s:Linf1}.

Finally, concerning Q2-manifolds with an honest base manifold $M$, there is the following result (see also \cite{Royt02}). 
\begin{thm}[Roytenberg]$\quad$\\ \label{thm:Roy}
A QP2-manifold is in 1:1-correspondence with a Courant algebroid.
\end{thm}
Some comments are in order: A QP-manifold is a Q-manifold (in the above sense) equipped with a compatible symplectic form $\omega$. QP2 signifies that the underlying N-manifold is an N2-manifold in our sense (cf.~footnote \ref{fn:Nk}) or, equivalently, that the 2-form $\omega$ has total degree $4=2+2$. 
A Courant algebroid can be defined as follows:
\begin{vdef}
\label{def:Cour} A \emph{Courant algebroid} consists of a vector bundle $E\to M$ together with a bundle map $\rho \colon E \to TM$ (over the identity), a Leibniz--Loday algebra $(\Gamma(E),[\cdot,\cdot])$ on the sections of $E$, and fiber metric   $\<.,.\>$ on $E$
such that for all $\psi_1$, $\psi_2 \in \Gamma(E)$ one has the following rules
\begin{align} \label{Coura}
  \< \psi_1,[\psi_2,\psi_2]\>&=  \tfrac{1}{2}\rho(\psi_1)\<\psi_2,\psi_2\>  = \<[\psi_1,\psi_2],\psi_2\> \, .
\end{align}
\end{vdef}
Here some further explanations may be in place: First, we call Leibniz--Loday algebra an algebra, where the left-multiplication is a derivation, i.e.
\begin{equation}\label{JaCou}
[\psi_1,[\psi_2,\psi_3]] = [[\psi_1,\psi_2],\psi_3] +[\psi_2,[\psi_1,\psi_3]]
\end{equation}
 holds true $\forall \psi_k \in \Gamma(E)$. Moreover, one can deduce the Leibniz rule \eqref{Leib} from the axioms. So, \emph{if} the bracket $[\cdot , \cdot]$ were antisymmetric, this would be a Lie algebroid (with an additional fiber metric). However, the first equation of \eqref{Coura} prohibits this: The inner product being non-degenerate, the symmetric part of the bracket can vanish only in the relatively uninteresting case of an identically vanishing map $\rho$ 
(in which case one obtains a bundle of quadratic Lie algebras over $M$---ad-invariance following from the remainder of \eqref{Coura} then). So, it is precisely the inner product (together with $\rho$) that governs the symmetric part of the bracket. The second equation of \eqref{Coura} expresses a (for non-vanishing $\rho$ not just pointwise) ad-invariance of the fiber-metric.

This completes the explanations of the terms appearing in the above theorem. As an aside we mention that one may also rephrase the axioms of a Courant algebroid in terms of the anti-symmetrization of the above bracket, 
\begin{equation}\label{anti}
[\phi,\psi]_2=\tfrac12\left([\phi,\psi] -[\psi,\phi]\right) \, .
\end{equation}
In that case, however, (\ref{JaCou}) does not yield a Jacobi identity for this new bracket. It thus does not come as a surprise that this can be rephrased in terms of an (infinite-dimensional) 2-term $L_\infty$-algebra $[.]_1\:V_1\to V_0$ \cite{Royt98}: Here the vector spaces are $V_0=\Gamma(E)$ and $V_1=\smooth(M)$, respectively, the map $[.]_1$ is provided by means of the de Rham differential $\ud$ followed by the transpose of the anchor map and a subsequent use of the inner product (to identify $E^*$ with $E$ and likewise so for the sections). Between two elements of $V_0$ the degree 0 binary bracket is given by \eqref{anti}, while for $f\in V_1$ and $\phi \in V_0$ one has $[\phi,f]_2:=\rho(\phi)f$.
The degree 1 ternary bracket, finally, can be non-vanishing only on $[.,.,.]_3:\Lambda^3V_0\to V_1$, where it is given by $[\psi_1,\psi_2,\psi_3]_3=\tfrac13\<[\psi_1,\psi_2]_2,\psi_3\>+\cycl$  All higher brackets vanish for degree reasons.  

We now turn to describing our main result of section \ref{s:examQ}. As in the previously mentioned results on Q1-manifolds, Lie 2-algebras, and QP2-manifolds, an essential ingredient of it is a definition (specifying to what precisely the Q2-manifold is equivalent). To obtain this definition, we proceeded as sketched above for Q1-manifolds, i.e.~analyzed the behavior of coefficient functions with respect to coordinate changes on the N2-manifold, which we identified with the sum of two vector bundles (in a non-canonical way, but different choices being related by an isomorphism). We first provide the generalization of Definition \ref{strictL2} to the case of a general base $M$ of $\CM$, trying to stay as close as possible to the known one for $M=\pt$: 
\begin{vdef}\label{d:sLie2}
A \emph{strict Lie 2-algebroid} or a \emph{crossed module of Lie algebroids} is a pair of 
 $(E,V)$ of Lie algebroids over the same base $M$, where $\rho_V \colon V \to TM$ vanishes,
together with a homomorphism $t \colon V\to E$ and a representation $\Ena$ of $E$ on $V$
such that  $t\left(\Econn_\psi v\right) = [\psi,t(v)]$
 and  $\Econn_{t(v)}w = [v,w]$   hold true for all $\phi,\psi\in\Gamma(E)$, $v,w\in\Gamma(V)$.
\end{vdef}
The generalization of a representation of a Lie algebra $\g$ on a vector space to a representation of a Lie algebroid $E$ on a vector bundle $V$ (over the same base) uses the notion of an $E$-covariant derivative: This is a map $\Ena\:\Gamma(E)\times\Gamma(V)\to\Gamma(V)$ such that for all $\psi\in\Gamma(E)$, $v\in\Gamma(V)$, and $f\in\smooth(M)$
\[\label{EConn}
  \Ena_{f\psi} v = f \Ena_\psi v, \quad\quad 
  \Ena_\psi (f\cdot v) = \rho(\psi)[f]\cdot v +f\Ena_\psi v.
\]
This becomes a representation iff its $E$-curvature, defined by means of  $\ER(\phi,\psi)v:=[\Econn_\phi,\Econn_\psi]v -\Econn_{[\phi,\psi]}v$, vanishes,$\ER=0$.  This flatness condition implies also that $\Econn_\psi$ is a derivation of the (bundle of) Lie algebra(s) $(\Gamma(V),[.,.])$.

\begin{vdef}\label{def:Lie2d}  A (semi-strict) \emph{Lie 2-algebroid} is defined as in Definition \ref{d:sLie2} except that the (still antisymmetric) brackets on $\Gamma(E)$ and $\Gamma(V)$ are no more required to satisfy the Jacobi identity and 
 $\Ena$ is no more flat in general. Instead, all these deviations are governed by a single $E$-3-form $H\in\Omega^3(E,V)$ with values in $V$:  For all $\psi_k\in\Gamma(E)$, $v,w\in\Gamma(V)$ and $f\in\smooth(M)$ one has
\begin{align}
  [\psi_1,[\psi_2,\psi_3]] +\cycl &= t\left(H(\psi_1,\psi_2,\psi_3)\right), \label{HJacobi} \\
   \ER(\psi_1,\psi_2)v &= H(\psi_1,\psi_2,t(v)), \label{HRep}\\
  \ED H &= 0 \, . \label{DH}
\end{align}
\end{vdef}
Here $\ED$ is the exterior $E$-covariant derivative associated to $\Ena$. We did not need to specify the Jacobiator of the bracket on $\Gamma(V)$, since it follows at once from \eqref{HRep} and $\Econn_{t(v)}w = [v,w]$: 
\begin{equation}
[u,[v,w]] + \cycl = H(t(u),t(v),t(w)) \qquad \forall u,v,w \in \Gamma(V) \, .
\end{equation} 
It is noteworthy that in this formulation, the coherence condition for $H$ takes the simple form \eqref{DH}. 
In this paper we will then prove the following:
\begin{thm}\label{thm1.5} $\quad$  \\ A Q2-manifold is in 1:1-correspondence with a (semi-strict) Lie 2-algebroid \\ up to isomorphism.
\end{thm}
Similarly to Theorem \ref{Vain}, this statement can be turned also into an equivalence of categories. 


In this context it is important to remark that the notion of a morphism of Lie 2-algebroids (as defined above in Definition \ref{def:Lie2d}) is less restrictive than the one of  strict Lie 2-algebroids (Definition \ref{d:sLie2}). This can be seen as follows: A definition of a morphism of a Lie 2-algebroid can be tailored from the corresponding Q-morphisms and the identification given in the above theorem  (such as one can do it for ordinary Lie algebroids and Q1-manifolds). On the one hand, the change of a splitting of the Q-manifold corresponds to a Q-morphism. 
On the other hand, it corresponds to an element $B\in\Omega^2(E,V)$ and $H$ changes under this $B$ (in a slightly tricky way, cf.~Equation \eqref{split3} below) and in particular this can be used to generate a non-zero $H$ from a vanishing one. 

A remark to the nomenclature seems important. One of the goals here is to give a definition of a mathematical structure equivalent to a Q2-manifold on the one hand and generalizing the definition of a Lie 2-algebra as given in \cite{Baez03vi} on the other hand. In other words, we strive at a definition equivalent to the \emph{super}-geometrical notion of a Q2-manifold in terms of notions used only in \emph{ordinary} differential geometry. This can be viewed as the same philosophy underlying the Theorems \ref{Vain} and \ref{thm:Roy}---with the difference that in these two cases the definitions \ref{def:Loid} and \ref{def:Cour} were already established as a geometrical notion before the equivalence  (expressed in the two theorems mentioned above) to the (simpler) super- or graded-geometrical structure was observed. Some also call Q$k$-manifolds simply Lie $k$-algebroids \emph{by definition}. To emphasize the non-tautological nature of Theorem \ref{thm1.5}, it may thus be useful to add the specification ``semi-strict'' in the equivalence, with \emph{semi-strict} Lie 2-algebroids being then \emph{introduced} (to our knowledge for the first time) in 
Definition \ref{def:Lie2d} above. 

Note that as a simple corollary of the above theorem, we find that any Q2-manifold gives rise to an infinite-dimensional Lie 2-algebra or also 2-term $L_\infty$-algebra, $\Gamma(V) \xto{t} \Gamma(E)$. In fact, a Lie 2-algebroid provides the additional control over such an infinite-dimensional Lie 2-algebra in the same way as a Lie algebroid provides one for its infinite-dimensional Lie algebra of sections. 

Certainly, a QP2-manifold is also a Q2-manifold. The infinite-dimensional Lie 2-algebra found in the above way is, however, \emph{not} the same one as the one constructed in \cite{Royt98}: There the  2-term complex is $\Omega^0(M) \xto{t} \Gamma(E)$ and, containing  the de Rham differential  $\ud$, neither $t$ nor the Jacobiator are  $C^\infty(M)$-linear, 
whereas in our construction one has $\Omega^1(M) \xto{t} \Gamma(E)$ rendering  the map $t$ and the Jacobiator $C^\infty(M)$-linear.

Further perspectives on Q$2$-manifolds in terms of particular  $L_\infty$-algebras or algebroids are deferred to Appendix \ref{s:Linf}).

There is also another way of obtaining the axioms for a Lie algebroid from studying a Q1-manifold than the one we sketched above (by a study of components of $Q$ with respect to changes of coordinates). This is provided by regarding the Q-derived bracket on the space of degree $-1$ vector fields on the Q-manifold $\CM$: Let $\psi$, $\psi'$ be two such vector fields, then
\begin{equation}
[\psi,\psi']_Q := [[\psi,Q],\psi']] \label{QderLie}
\end{equation}
is another vector field of degree -1 on $\CM$. While in the case of a Q1-manifold one easily identifies vector fields of degree -1 on $\CM = E[1]$ with sections in the Lie algebroid $E$ and the bracket \eqref{QderLie} with the Lie bracket between them, so that at the end one rediscovers also by this procedure the usual axioms of a Lie algebroid, for a Q2-manifold one obtains a \emph{different} picture. First of all, the derived bracket does not need to be antisymmetric anymore, since $[\psi,\psi]$ can be non-zero now as, in contrast to a Q1-manifold, a Q2-manifold carries non-trivial degree -2 vector fields (for more details and explanations we refer to section \ref{sec:derived}). The derived bracket, however, always gives a Leibniz--Loday algebra, satisfying a condition of the form \eqref{JaCou}. These observations, which are motivated moreover by our study of the gauge transformations (cf., in particular, Equation \eqref{derQtot})\footnote{Note that $\ad_Q(\psi) \equiv [Q,\psi] = [\psi,Q]$ since both vector fields are odd.}, lead to an attempt of finding another description of a Q2-manifold $\CM$. In fact, since $[\cdot,\cdot]$, restricted to degree -1 (and thus odd) vector fields, yields a $C^\infty(M)$-linear symmetric map from the degree -1 vector fields on $\CM$ to the degree -2 vector fields, and  since vector fields of negative degrees on an N-manifold can always be identified with sections of vector bundles, one swiftly extracts from the above an axiomatic formulation very close to the one of a Courant algebroid, with the difference that now the inner product can take values in another bundle. 

\begin{vdef}
\label{def:VCour} A \emph{$V$-twisted Courant algebroid} consists of a vector bundle $W\to M$, a bundle map $\rho \colon W \to TM$, a Leibniz--Loday algebra $(\Gamma(W),[\cdot,\cdot])$, a non-degenerate surjective symmetric product $\<.,.\>$ taking values in another vector bundle $V \to M$, and  a $W$-connection  $\Wna$ on $V$
such that for all $\psi_1$, $\psi_2 \in \Gamma(W)$ one has the following axioms
\begin{align} \label{VCoura}
  \< \psi_1,[\psi_2,\psi_2]\>&=  \tfrac{1}{2}\Wna_{\psi_1}\<\psi_2,\psi_2\>  = \<[\psi_1,\psi_2],\psi_2\> \, .
\end{align}
\end{vdef}
It is easy to see that this reduces to an ordinary Courant algebroid if $V=M \times \R$ and $\Wna(1)=0$ (in this case of the trivialized bundle we can identify its sections with functions and if the unit section is parallel, $\Wna$ reduces to $\rho$ on these functions). More generally, we can say that a $\rk V=1$, $V$-twisted Courant algebroid admitting a non-vanishing, $W$-constant section $v\in\Gamma(V)$ (i.e.~$v(x)\neq 0 \; \forall x\in M$ and $\Wconn_\psi v=0 \; \forall \psi \in \Gamma(W)$) is in a canonical correspondence with a Courant algebroid. However, line-bundle twisted Courant algebroids are a strictly more general notion than the one of a Courant algebroid. 

A nontrivial example of a $V$-twisted Courant algebroid can be provided as follows: Let $A$ be a Lie algebroid together with a flat $A$-connection ${}^A\nabla$ on a vector bundle $V$. Then the couple $W$, $V$, becomes a $V$-twisted Courant algebroid if one takes $W:=A\oplus \Hom(A,V) \cong A\oplus A^*\otimes V$ and equips it with the bracket
\begin{equation}\label{ACourant}
[X\oplus\alpha,Y\oplus\beta] = [X,Y] \oplus [\imath_X,\uD]\beta -\imath_Y\uD\alpha +H(X,Y,\bullet)
\end{equation}
where $\uD\:\Omega^k(A,V)\to\Omega^{k+1}(A,V)$ is the exterior covariant derivative constructed from ${}^A\nabla$, $[.,.]$ the Lie algebroid bracket on $A$, and $H\in\Omega^3(A,V)$ is any $A$-3-form with values in $V$ satisfying $\uD H=0$.  The other maps are obvious, for example $\Wconn_{X\oplus\alpha} v={}^A\conn_X v$.  Certainly, this reduces to the standard Courant algebroid $W\cong TM \oplus T^*M$ twisted by a closed 3-form $H$ if $A$ is chosen be the standard Lie algebroid $TM$ (and $V$ and $\Wna$ as mentioned above), in which case \eqref{ACourant} reduces to the known, $H$-twisted Courant-Dorfman bracket 
\begin{equation} \label{standardtwistCour}
[X\oplus\alpha,Y\oplus\beta] = [X,Y] \oplus \L_X \beta -\imath_Y\ud\alpha + \imath_Y \imath_X H \, .
\end{equation}
While an ordinary Courant algebroid reduces to a quadratic Lie algebra when $M$ is a point, the ``decoupling'' of $\rho$ and $\Wna$ has the advantage that now we can also host (non-Lie) Leibniz--Loday algebras in this framework. For $M=\pt$, the definition reduces to a Leibniz--Loday algebra together with an invariant inner product with values in a vector space that has a representation.  Also, conversely, one can promote \emph{any} Leibniz--Loday algebra $W$ to a $V$-twisted Courant algebra $W$ in a \emph{canonical} way, cf.~Proposition \ref{canonical} below.


The main result of section \ref{sec:derived} is then the following (for further details, cf.~Proposition~\ref{p:FWder} and Theorem~\ref{thm:der})
\begin{thm}$\,$ \\ Up to isomorphism and  provided that $\rk V>1$, there is a 1:1 correspondence between Lie 2-algebroids $V \xto{t} E$ and $V$-twisted Courant algebroids $W$ fitting into the following exact sequence \begin{equation} \label{VEW}
0\lto\Hom(E,V) \lto W\xto{\;\;\pi\; \;} E\lto 0
\end{equation} with $\rho_W = \rho_E \circ \pi$. In particular, if $\rk V >1$ and after the choice of a splitting of \eqref{VEW}, 
the bracket on a $V$-twisted Courant algebroid $W$ as above has the form
\begin{equation}\label{VCourant2}
[X\oplus a,Y\oplus b]_W = [X,Y]_E - t(a(Y)) \oplus [\imath_X,\ED]b -\imath_Y\ED a +\imath_Y \imath_X H+ b \circ t \circ a -  a\circ t \circ b  ,
\end{equation}
where all the operations on the r.h.s.~are the induced ones on the corresponding Lie 2-algebroid.\label{thm1.6}
\end{thm}
In the above formula, $X$, $Y \in \Gamma(E)$ and $a$, $b$ are either viewed as elements of $\Omega^1(E,V)$ (in the first two terms following the $\oplus$) or as homomorphisms from $E$ to $V$ so that they can be composed with the fixed map $t \in \Hom(V,E)$ of the Lie 2-algebroid. 

The similarity of formula \eqref{VCourant2} with Equations \eqref{standardtwistCour} and \eqref{ACourant} is striking, in particular in view of the fact that, in contrast to those examples, now the (anti-symmetric) Lie 2-algebroid bracket $[ \cdot , \cdot ]_E$ does no more satisfy a Jacobi condition in general, while still \eqref{VCourant2} obeys the left-derivation property \eqref{JaCou}. Similarly, $\ED$ does no more square to zero in general (such as $\uD$ and $\ud$ in the analogous formulas \eqref{ACourant} and \eqref{standardtwistCour}, respectively),
but rather gives the $E$-curvature \eqref{HRep}. It is here where the simple $t$-contributions to the new bracket become essential. 

Theorems \ref{thm1.5} and \ref{thm1.6} certainly also induce a 1:1 correspondence between the above mentioned ``exact V-twisted Courant algebroids'' (with $\rk V >1$) and  Q2-manifolds $\CM$ (with at least two independent generators of $C^\infty(\CM)$ of degree 2). Exact $V$-twisted Courant algebroids thus provide an alternative, ordinary differential geometric description of the graded-geometrical notion of a Q2-manifold---with respective restrictions provided in the parenthesis of the previous sentence.

The condition of $\rk V \neq 1$ seems necessary, since there exist examples of line bundle twisted Courant algebroids not corresponding to Q2-manifolds (in the above way). In fact, one to our mind interesting problem left open in the analysis of the present paper is a classification of line bundle twisted Courant algebroids, i.e.~$V$-twisted Courant algebroids $W$ with $\rk V = 1$. 

On the technical level this restriction on the rank can be traced back to the following simple fact of linear algebra (proven in Lemma \ref{p:Jh=0}): Let $V_1$ and $V_2$ be two vector spaces. Then any linear operator $\Delta \colon \Hom(V_1,V_2) \to V_1$ satisfying $a(\Delta(a))=0$ for any $a \in \Hom(V_1,V_2)$ has to necessarily vanish \emph{iff} $\dim V_2 >1$, while, e.g., for $V_2 = \R$ one may identify $\Delta$ with an element of $V^1 \otimes V^1$ and then any $\Delta \in \Lambda^2 V_1$ obviously satisfies the condition for symmetry reasons. 

In view of the different descriptions of QP2 manifolds as infinite dimensional Lie 2-algebras sketched before, it is at this point no more astonishing that if the QP2 manifold is regarded as a Q2 manifold (by the forgetful functor), one does not obtain the usual Courant algebroid in the above fashion. Instead, as before, $M \times \R$ is replaced by $T^*M$, so one obtains a particular $T^*M$-twisted Courant algebroid in the above identification (when applied to a QP2 manifold viewed upon as a Q2 manifold). 

The bracket \eqref{standardtwistCour} is also well-defined for sections of $TM \oplus \Lambda^k T^*M$ if $H \in \Omega^{k+2}_{closed}(M)$, (without the $H$-twist, this bracket goes back to Vinogradov, cf.~\cite{Vino90}). One obtains similar properties to those of the standard Courant algebroid, but, for example, the symmetrized contraction of its elements provides an inner product which now takes values in $V=\Lambda^{k-1} T^*M$. In section \ref{sec:derived} we propose an axiomatization of these properties:  Very much as \cite{Xu97} does it for the case of $k=1$, leading  (together with its later reformulations/simplifications) to the Definition \ref{def:Cour} above, we propose the definition of a \emph{Vinogradov algebroid} (cf.~Definition \ref{def:Vino} below). We also show that $V$-twisted Courant algebroids are particular Vinogradov algebroids. 

Section \ref{sec:derived} is a relatively long section and contains further potentially interesting subjects (like the axiomatization of the bracket \eqref{standardtwistCour} when applied to higher form degrees into a Vinogradov algebroid and a simple non-standard Q-description of $H$-twisted standard Courant and Vinogradov algebroids). We do not want to summarize all this here, so as to keep some tension for the interested reader.

\vskip\newpassage

Finally, in the Appendix \ref{s:Linf} we provide a further reformulation of Q$k$-manifolds in terms of $k$-term $L_\infty$-algebroids, which we also define there so as to generalize the more standard relation between Q-manifolds over a point and $L_\infty$-algebras \cite{SHLA}.

\newpage
\section{From Bianchi identities to Q-structures} \label{s:Bian}
We want to construct a gauge theory for some $p$-form gauge field. As examples
show, at least if we want to go beyond a trivial abelian version of such a theory with $p\ge 1$, 
this necessarily involves a tower of gauge fields, where the highest form
degree is $p$. For definiteness we take $p=2$, but the generalization to
arbitrary $p$ is immediate. We thus consider a collection of 2-forms $B^D$,
together with 1-form fields $A^a$, and in general also 0-form fields
$X^i$. The scalar fields are permitted to take values in some internal
manifold $M$, which is not necessarily a vector space---so in general the
theory is a sigma model. Denoting spacetime by $\Sigma$, one has
 $X \colon \Sigma \to M$, and the particular collection of scalar fields
 $X^i=X^i(\x)$ ($\x \in \Sigma$) results after a (local) choice of
 coordinates $x^i$ 
on $M$, $X^i = X^*(x^i)$.\footnote{It is
tempting to exchange the notation of $M$ and
$\S$; then $M$ would be spacetime, and if $\S$ is nontrivial, one has a
sigma model. However, here we conform to previous publications.}  In
many examples there are no scalar fields and the tower of gauge fields starts
with 1-forms only. This just corresponds to the choice $M$ a point, in which
case one may drop the map $X$ as it contains no information. However, even in this case, scalar fields would then enter the theory in a second step representing matter; intending to stay as general as possible, we permit their presence thus right from the beginning. 

One of the main ingredients in a gauge theory are the field strengths. We
also know how they should start, namely with the exterior derivative acting
on the respective gauge field, possibly then modified by adding wedge
products of gauge fields to this expression. We thus make the most
general ansatz of this kind:
\ba  F^i &:=& \ud X^i - \rho^i_a A^a \, , \label{Fi}\\
        F^a&:=& \ud A^a + \2 C^a_{bc} A^b A^c - t^a_D B^D\, , \label{Fa}\\
        F^D&:=& \ud B^D + \Eco^D_{aC} A^a B^C - \frac{1}{6} 
        H^D_{abc} A^a A^b A^c \, ,
\label{FB}
\ea
where wedge products between differential forms are understood 
and
all the yet unrestricted coefficients can be functions of  $X$.
For some purposes it will prove convenient to use a more condensed
notation: We then denote the field strengths 
collectively by $F^\alpha$ and likewise 
the gauge fields by $A^\alpha$ (so $A^i \equiv X^i$, $A^D \equiv B^D$).


We now want to require a  very mild form of  Bianchi identities. In
ordinary YM-gauge theories Bianchi identities tell us that $\ud
F^a\equiv -C^a_{bc} A^b F^c$. We want to relax/generalize this
condition on the complete set of field strengths $F^\alpha$ by requiring
that the exterior derivative of them produces terms that all 
\emph{contain} field strengths again, $\ud F^\alpha = \lambda^\alpha_\b F^\b$,
where $\lambda^\alpha_\b$ may be  arbitrary, field dependent factors.
Interpreted in this sense, we require that one has:
\beq \boxed{\ud F^\alpha |_{F^\beta=0} \equiv 0}\, . \label{bian}\ee
This equation has to be understood in such a way that for \emph{every} choice of the gauge fields $A^\alpha$ satisfying $F^\alpha =0$ one has $\ud F^\alpha =0$.
We remark right away that this is satisfied also for generalizations
such as for 
higher YM theories as discussed in \cite{Baez02jn}
or the Algebroid YM theories of \cite{Str04b}. Moreover, there are
occasions where
the $X$-fields above, considered as part of the gauge field here, can
be identified with some Higgs type scalar fields  and $F^i$ then
corresponds to the covariant
derivative of this scalar field $X^i$. In standard
YM-theory e.g.~we may use the fact that the square of the
exterior covariant derivative
is proportional to the curvature to show
that also in this enlarged setting 
\eqref{bian} is satisfied.

Let us rephrase the 
requirement \eqref{bian}
in more mathematical terms: Within
the differential graded commutative algebra $(\A,\ud)$
generated by abstract elements $A^\alpha$ and $\ud A^\alpha$ there is an ideal $\I$
generated by $F^\alpha$ as 
defined in Eqs.~\eqref{Fi}--\eqref{FB}, i.e.~by the choice of
structural functions $\rho^i_a(X)$, $\ldots$, $H^D_{abc}(X)$. The
generalized Bianchi identities hold iff (by definition) this ideal is
an invariant subspace with respect to the exterior derivative, i.e.~iff Eq.~(\ref{ideal}) holds true.

There are some further possible reformulations of the above
conditions. 
One may rephrase them as saying e.g.~that the
structural functions must be such that the quotient
algebra $\A/\I$ forms a so-called free differential algebra
(cf.~\cite{Sull77b} or, for a physical
example in low dimensions, e.g.~\cite{Izq99}). Since this quotient
corresponds to looking at vanishing field strengths only, which will
\emph{not} be our field equations when interested in constructing
a (non-topological) generalization of Yang-Mills theories, we avoid
this perspective, such as that we avoid the notation $\ud F^\alpha
\approx 0$ instead of \eqref{bian} since for non-topological
theories the ``on-shell'' equality would not coincide with the on-shell
notion dictated by field equations. Finally, Eq.~\eqref{bian} may
also be viewed as an integrability condition in the following sense:
For all fields $A^\alpha$ which satisfy $F^\alpha(\s)=0$ for some point $\s
\in \S$, it follows automatically (by the choice of structural
functions) that also $\left(\ud F^\alpha\right)(\s)=0$.


The main purpose of this section is twofold: First, we want to show
that the simple ``physical'' and rather general requirement 
\eqref{bian}
automatically and naturally leads to
Q-structures.\footnote{By definition (but cf.~also the more detailed
 explanation below), Q-structures, introduced in the 
 context of topological field theories in \cite{AKSZ}\label{ct:Scha} 
  are homological vector fields of degree one on a
supermanifold.}
Secondly, to obtain the explicit form
of the generalized Bianchi identities (cf., e.g.,
Eqs.~\eqref{BFi}--\eqref{BFB}
below), the use of this formalism shortens calculations extremely.

For this purpose we first unify the gauge fields into a single map.
The scalar fields $X ^i(\x)$ describe some map  $a_0 \colon \Sigma \to
M$ ($a_0 \equiv X$ in previous notation); as already mentioned, they
are the pullback
of coordinate functions 
on $M$ to functions on $\Sigma$  by precisely this map (possibly
restricted to local charts). To imitate
this for the other components $A^a, 
B^D, \ldots$ of the gauge 
field, we introduce auxiliary spaces spanned by coordinates 
$\xi^a, b^D, \ldots$, taking these spaces to be linear and thus 
defined by the respective linear coordinates: $\W = \langle \xi_a   \rangle$,
$\V = \langle b_I   \rangle$, where $\xi_a$ is a basis 
of  vectors 
dual to $\xi^a$ (linear coordinates on a vector space are a basis in the
dual vector space) etc. 

Since every 1-form on $\S$ may be viewed as a (fiberwisely linear)
function on $T\S$, we may thus
regard $A^a$ as the pullback of the coordinate function $\xi^a$ by a map from
$T\S$ to $\W$. A local coordinate system on $T\S$ can be chosen in the form
$\x^\mu, \theta^\mu$, where $\theta^\mu=\ud \x^\mu$. Since the natural
multiplication of 1-forms is the antisymmetric wedge product, it is near at
hand to declare the fiber coordinates $\theta^\mu$ of $T\S$ as anti-commuting;
this is taken into account by the notation $T[1]\S$: the coordinates $\x^\mu$
of the base are of degree zero (and thus even and commuting), while
$\theta^\mu$ is of degree one (and in particular odd). In this way we may interpret $A^a$ as
the pullback of $\xi^a$ to functions on $T[1]\S$ (i.e.~differential forms) by
a map $a_1 \colon T[1]\S \to \W[1]$, $A^a = (a_1)^* \xi^a$; here we declared
the coordinates $\xi^a$ to be of degree one, too, so that the map $a_1$ becomes
degree preserving, and $(a_1)^*$ a morphism of algebras.  Likewise,
 $B^D = (a_2)^* b^D$, where $a_2 \: T[1]\S \to \V[2]$, and $b^D$ are declared to
be of degree two, since $B^D$ is a 2-form (that is a function on $T[1]\S$ of
degree two) and we want $a_2$ to be degree preserving again.\footnote{If one considers supersymmetric or supergravity theories, the form degree does not determine if a homogeneous element is commuting or anti-commuting. If one has bosonic and fermionic 0-forms, for example, then already $M$ will be a corresponding supermanifold. Likewise, there may be anticommuting and commuting 1-forms in a supergravity theory; to describe such objects one would introduce a bidegree on the target, so that the first degree corresponds to the form degree and the other one to the commutativity property. We will not go into further details about this possible generalization here; as a simple prototype of such a theory one may look at supergravity theories reformulated as Poisson sigma models with a super-Poisson structure on the target as given in \cite{Strobl:1999zz}.}

If $X^i$ denote local coordinates on $U \subset M$, we can assemble the above
maps into a single map $a$ from (open domains of) $T[1]\S$ to $U \times \W[1] \times
\V[2]$. It is tempting to require this direct product form of the target
manifold only on local charts $U$. Globally we may permit a twist (nontrivial
transition functions) in changing the vector spaces when changing charts in
$M$.  In fact, what we permit is a super-manifold $\CM_2$ of degree 2: Locally
it is described by degree 0 coordinates $x^i$, degree 1 coordinates $\xi^a$,
and degree 2 coordinates $b^D$.\footnote{\label{fn:pos} 
In this and the following section we will 
consider only non-negatively graded supermanifolds. (We will often call an integer graded manifold with its induced superalgebra of functions determined by the parity of the degrees simply a supermanifold, understanding implicitly from the context that we consider the full $\Z$-grading (and not just modulo 2)). Speaking of the degree of it we will 
then always mean the maximal degree of ``local coordinate functions'' (local generators). 
Later on also supermanifolds with (coordinate) functions of negative degrees will play 
some role, cf in particular Section~\ref{sec:derived}.} 
The transition functions need to be degree
preserving (by definition). This implies that forgetting about the degree two
coordinates one has a vector bundle $E \to M$ with typical fiber
$\W$. Similarly, $x^i, b^D$ define a vector bundle $V \to M$ with typical fiber
$\V$.  (Since the transition functions of $b^D$s may also contain a part
quadratic in the degree one coordinates $\xi^a$, the total space
$\CM_2$ does not correspond to the direct sum vector bundle $E \oplus
V$ canonically, but does so only after an additional choice of
section; we will come back to this subtlety in the subsequent section).

Thus in sum we obtain a degree preserving map $a \colon \CM_1 \to \CM_2$,
where $\CM_1 = T[1]\S$ is a degree one super-manifold.\footnote{More generally,
we may also permit a nontrivial fibration over $\CM_1$, such that $\CM_2$ is
only a typical fiber, leading to the notion of ``$Q$-bundles''. This
is dealt with in \cite{KS07}, while in this paper we restrict to
trivial bundles (except for the target bundle, encoded in the
supermanifold $\CM_2$). Still, the present 
formulas remain correct also in the more general setting if they are
reinterpreted as local ones.}
Choosing local coordinates $q^\alpha : = (x^i, \xi^a, b^D)$
on $\CM_2$ this gives rise to 
functions on $\CM_1$ by means of  the pullback $a^*$, which then 
just yields the set of  local 0-forms, 1-forms, and 2-forms,
respectively, with which we started this section, $A^\alpha = a^*(q^\alpha)$. 

The de Rham differential $\ud$ on $\S$ corresponds to a degree one
vector field 
on $T[1]\S \equiv \CM_1$, namely $\Qo = \theta^\mu 
\partial/\partial \x^\mu$. It squares to zero, $(\Qo)^2 \equiv 0$, and thus defines 
what is called a $Q$-structure on $\CM_1$. We now define a likewise vector field on $\CM_2$ 
using the above data. It is sufficient to know its action on coordinates $q^\alpha$: 
\ba \Qt x^i &:=& \rho^i_a \xi^a \, , \quad \label{Qi}\\
     \Qt \xi^a& := &- \2 C^a_{bc} \xi^b \xi^c +t^a_D b^D\, , 
     \label{Qa}\\
        \Qt b^D&:= & -\Eco^D_{aC} \xi^a b^C + \frac{1}{6} 
        H^D_{abc} \xi^a \xi^b \xi^c \, .\label{QB}
\ea
This vector field is obviously also of degree one, i.e.~it raises the degree
of the respective coordinate by precisely one unit. Furthermore we
introduce (cf.~also \cite{BKS})
\beq \CF := \Qo \circ a^* - a^* \circ \Qt \qquad \Rightarrow \qquad \CF (q^\alpha) 
\equiv F^\alpha \, . \label{CF}\ee

Up to here we only rewrote the initial data contained in
Eqs.~\eqref{Fi} --  \eqref{FB}. We now come to analyze
Condition~\eqref{bian}. Denoting by $\ev_\x$ evaluation of a
function on $\CM_1=T[1]\S$ at a point $\x \in \S$, in the present
language the generalized Bianchi identities (in their formulation as
integrability conditions for $F^\alpha=0$ equations) can be translated into the 
requirement that for every $\x \in \S$
the composed operator $\ev_\x \circ \Qo \circ
\CF$ vanishes whenever 
$a\colon\CM_1 \to \CM_2$ is such that the operator $\ev_\x
\circ \CF$ vanishes.

Now, using  
$(\Qo)^2 \equiv 0$, we find the identity
\be
\ev_\x \circ \Qo \circ \CF =- \ev_\x \circ \left( \CF \circ \Qt
  +a^* \circ (\Qt)^2\right). \label{preBian}
\ee
Thus, for \eqref{bian} to hold true
it is obviously sufficient that the vector field $\Qt$ 
squares to zero. But also the converse is true, provided only that the
dimension of spacetime $\S$ is at least $p+2$ ($p$ is the degree of
the supermanifold $\CM_2$, in our discussion above $p=2$): Assume that
$(\Qt)^2 \neq 0$, that $a\colon \CM_1 \to \CM_2$ is such that
$\ev_\x \circ \CF \equiv 0$, and that the integrability
condition holds true. Then there exists at least one
basic coordinate function $q^{\alpha_0} \in C^\infty(\CM_2)$ for which
$(\Qt)^2 q^{\alpha_0}=: \varphi$ is a non-vanishing function on $\CM_2$.
Note that this function has at most degree $p+2$, so $a^* \varphi$ is a
differential form on $\S$ of at most degree $p+2$. 
Applying \eqref{preBian} to $q^{\alpha_0}$ under the above assumptions,
we obtain $\ev_\x a^* \varphi =0$, valid for every $\s \in \S$. 
If we are able to show that
there always exists some $a$ fulfilling $(F^\alpha)_\x=0$ for which this
equation is violated, we obtained a contradiction, proving that $\Qt$
must square to zero. For this to be the case, we need the condition
$\dim \S \ge p+2$, because, otherwise, \emph{any} violation
$\varphi$ of $(\Qt)^2=0$ in a degree higher than $\dim \S$ will not contribute
to $\ev_\x a^* \varphi$. On the other hand, it is precisely the
integrability of 
 the equations $F^\alpha =0$ that  ensure that for \emph{every} choice of
 $\ev_\x \circ a$ (i.e.~for every choice of $(A^\alpha)_\x$)
we can 
find a map $a$, defined at least in a neighborhood of $\x \in \S$, such
that also $(F^\alpha)_\x=0$. This concludes the proof.

We thus managed to translate the relatively weak form of Bianchi identities 
\eqref{bian}, $\ud F^\alpha \sim \sum F^\b$ (where the sum contains
arbitrary field dependent coefficients), 
into a $Q$-structure on a target super manifold $\CM_2$. This gives the first part of Theorem \ref{thm:phys1} mentioned in the introduction, or, in the
particular case $p=2$, it means in detail:
\begin{thm} \label{theo:nilpot}
For a spacetime dimension of at least 4, the nilpotency of the
super vector field $\Qt$
defined in Eqs.~\eqref{Qi} -- \eqref{QB} is necessary and
sufficient for the validity of \eqref{bian}. 
\end{thm}

The relation $(\Qt)^2=0$ contains \emph{all} the conditions we
want to place on the 
structural functions in \eqref{Fi} -- \eqref{FB}. They contain the 
generalization of the Jacobi identity of the structure constants in ordinary 
Yang-Mills gauge theories. In the subsequent section we will display all 
the conditions (in the example $p=2$), cf.~Eqs.~\eqref{B1} --
\eqref{B7} below, discuss
their general structure and present examples of it
such as Lie 2-algebras or Lie algebroids (an alternative perspective, containing Courant 
algebroids and their generalizations will be discussed separately in Section~\ref{sec:derived}).

In the present section we still want to display the explicit form of
the Bianchi identities that we obtain when \emph{using} the nilpotency 
of $\Qt$. One finds
\begin{align}  \ud F^i &\equiv \rho^i_{a,j}A^a F^j - \rho^i_a F^a \, ,
        \label{BFi}\\
      \ud   F^a&\equiv \frac12C^a_{bc,i}A^b A^c F^i -t^a_{D,i}B^D F^i
         -C^a_{bc} A^b F^c - t^a_D F^D\, ,  \label{BFa}\\
      \ud   F^D&\equiv -\Eco^D_{aC,i} A^aB^C F^i +\Eco^D_{aC} B^C F^a
        +\frac16H^D_{abc,i} A^aA^bA^c F^i  \nonumber\\
      &\quad  -\frac12H^D_{abc} A^aA^b F^c -\Eco^D_{aC} A^a F^C \, .
      \label{BFB}
\end{align}
Here the power of the employed super language becomes particularly transparent: 
Whereas in every explicit calculation one needs to make heavily use 
of the Eqs.~\eqref{B1} -- \eqref{B7} below, one may derive the above 
conditions within some lines by use of the following obvious two relations
(which
follow directly from the defining Equation~\eqref{CF} and the fact
that $\Qo$, $\Qt$ are nilpotent vector fields, cf.~also \cite{BKS,Str04b})
\ba  \Qo  \circ \CF &=& - \CF \circ  \Qt \label{dF} \\
   \CF(\phi) &=& F^\alpha \, a^*(\partial_\alpha \phi) \label{FLeibniz} 
\ea
where the second equation holds for all functions $\phi$ on the target
and $\partial_\alpha$ denotes a left derivation w.r.t.~$q^\alpha$. 
We illustrate this for $\alpha=i$: 
Applying the left hand side of \eqref{dF} to $x^i$, we get $\ud F^i$. On the 
other hand $\CF (\Qt x^i) = \CF(\rho^i_a\xi^a) = \CF(\rho^i_a) a^*(\xi^a)+
a^*(\rho^i_a) \CF(\xi^a) = \rho^i_{a,j} \CF(x^j) A^a + \rho^i_a F^a$.
(Following 
usual physics notation,
we did not explicitly display $a^*$ when acting on 
functions over $M$ such as  $\rho^i_{a,j}\equiv \partial_j \rho^i_a$,
which are then understood as
functions over $\Sigma$ via the ``scalar fields'' $X^i$).

   
The result above is readily specialized to particular cases of interest. 
E.g., for $M$ a point, we merely drop all terms containing $F^i$. If there is 
no $B$-field (i.e.~there are no degree two generators on $\CM_2$), we 
put $F^D$ and $B^D$ to zero wherever they occur. Combining both steps
we regain 
the Bianchi identities of ordinary Yang-Mills theory. 

\vskip\newpassage

Before continuing we briefly comment on some further
generalization. In the ansatz \eqref{Fi}--\eqref{FB} we could
have permitted also terms proportional to $\ud A^\alpha$, adding e.g.~a
term $\nu^a_i \ud X^i$ to the r.h.s.~of \eqref{Fa} etc. Since
such terms enter field strengths $F^\beta$ only for
values of $\beta$ corresponding to a higher form degree than the one
of $A^\alpha$, it is always possible to replace $\ud A^\alpha$ by $F^\alpha$ by
merely renaming coefficient functions in $F^\beta$. Thus, in the above
example of Eq.~\eqref{Fa}, we may instead equally well add $\nu^a_i
F^i$ to the ansatz, merely by
replacing $C^a_{bc}$ with new structure functions. Constructing
out of these new structural functions---except for the additional
coefficients $\nu^\b_\alpha$, which we may just ignore at this point---the 
vector field $Q_2$ as before, Eqs.~\eqref{Qi}--\ref{QB}, 
Theorem~\ref{theo:nilpot}
still holds true: While the ideal $\I$ generated by field strengths changes with 
the change of structural functions in the first step of the rewriting, it does not change by adding or dropping the terms with lower degree field strengths on the right hand sides in their definition and thus $\ud \I \subset \I$ gives the
same conditions.

More generally, the ideal $\I$ remains unaltered by every
redefinition of the form $\widetilde F^\alpha = M^\alpha_\b F^\b$ if the field
dependent matrix $M^\alpha_\b$ is invertible. If the form degree
of $A^\alpha$ is denoted by $|\alpha|$, then $M^\alpha_\b$ has form
degree $|\alpha| - |\b|\ge 0$. It thus is a lower triangular matrix, where
the non-vanishing off-diagonal components correspond to $\nu^\alpha_\b$
introduced above. The (block) diagonal pieces, those where $|\alpha| =
|\b|$, correspond to $X$-dependent coefficient matrices in front of 
$\ud A^\alpha$ in the definition of the new field strength $\widetilde
F^\alpha$---the matrix $M^\alpha_\beta$ is thus invertible, if and only if each of its 
diagonal components is nonzero everywhere. For example, instead of 
\eqref{Fi} we would consider
$\widetilde F^i = M^i_j \ud X^j - \widetilde \rho^i_a A^a$. Now, if the
matrix $M^i_j(X)$ is invertible (for every point in the target $M$),
we may introduce $\rho^i_a := (M^{-1})^i_j \widetilde \rho^j_a$, and
this is the first component of the vector field $Q_2$ of Theorem
\ref{theo:nilpot}, 
and so on.  Only
redefinitions of the above form with a non maximal and possibly even
varying rank of the
diagonal pieces of $M^\alpha_\b$ would require a new and in general more
intricate analysis. 

Also coordinate transformations on $M$, or, more
generally, on the supermanifold $\CM_2$, will induce
redefinitions of field strengths. In an active interpretation of coordinate transformations, i.e.~viewing them as (local expressions of) diffeomorphisms instead of a 
change of a chart (passive interpretation in physics terminology), the ideal $\I$ changes in general, but in a covariant way. In this active interpretation $\I$ remains unchanged only under very particular super-diffeomorphisms, namely those that leave the $Q$-structure invariant; we will come back to them in the context of gauge symmetries in section \ref{s:gauge} below. 

If $M$ is not a point, 
the field strength components $F^\alpha$ themselves do, however, \emph{not} transform covariantly, except for the lowest ones with $|\alpha| =0$. (Only the complete set
of the whole ideal $\I$ does). 
 This can be corrected by means of the choice of connections in
the vector bundles $E$ and
 $V$.\footnote{The bundles $E$ and $V$ and these connections are not to
be confused with the (total) bundle of the gauge theory and the gauge fields,
respectively. The first one, the total bundle, is assumed to be trivial
within this paper---for a generalization cf.~\cite{KS07}---and we
on purpose did not call $A^\alpha$ a connection since this would require
further geometrical justification, and in some way is even the wrong
notion. The choice of $M$, $E$, $V$, and possibly connections on the
latter two, are fixed background data; they replace the choice of a
Lie algebra in ordinary Yang-Mills theories. Only the gauge fields $A^\alpha$ 
are dynamical.} One way of seeing this is to first determine the
transformation behavior of all the structural coefficient functions 
under a change of coordinates on $\Mt$---we will do this within the
subsequent section---and then to correct for the overall behavior of
$F^a$ and $F^D$ by appropriate additions. There is a much simpler route,
however, using the fact
that $\CF$ is a (graded) Leibniz-type operator
over $a^*$, cf.~Eq.~\eqref{FLeibniz}:
A change of local frame in the vector bundle $E$, e.g.,
corresponds to a change of coordinates $\xi^a \mapsto
\widetilde{\xi^a} = M^{\tilde a}_b \xi^b$ on $\Mt$, where $M^{\tilde a}_b$ are
functions on the base\footnote{which is not to be confused with
spacetime $\S$---both $M$ and $E$ are data of the target.}
 $M$ of $E$. With \eqref{FLeibniz} this implies $F^{\widetilde{a}}= M^{\tilde
 a}_b(X) F^b +  M^{\tilde a}_{b,i}(X) F^i  A^b$, 
 where $M^{\tilde a}_{b,i} \equiv \partial
 M^{\tilde a}_{b}/\partial x^i$. Only the first term of the two
 corresponds to a tensorial behavior. From this it is obvious
 that if $\Gamma^a_{ib}$ and
 $\Gamma^D_{iC}$ denote connection coefficients in local frames in $E$
 and $V$, respectively, the redefined field strengths 
 \ba F^a_{(\Gamma)} &:=& F^a + \Gamma^a_{ib} F^i  A^b \label{FaG}
 \\
     F^D_{(\Gamma)} &:=& F^D + \Gamma^D_{iC} F^i  B^C
  \label{FBG} \ea
(together with the unaltered $F^i$) are tensorial. In some gauge invariant 
action functionals $S$  there may be
a contribution of the form $\Lambda_i F^i$, where $\Lambda_i$ is a
Lagrange multiplier field;\footnote{Cf, e.g., \cite{Str04b,KS08}.} in such a case one may replace $F^a_{(\Gamma)},
F^D_{(\Gamma)}$ by the simpler expressions
$F^a, F^D$ in the remaining part of the action $S$---at
the mere expense that the new, redefined Lagrange multipliers loose
their tensorial behavior. For other functionals these contributions
may be mandatory, however. 

Finally, an interesting alternative to the use of $\CF$ or the auxiliary connections needed for a tensorial behavior of the field strengths like in Eqs.~\eqref{FaG} and \eqref{FBG} is to introduce a map $f \colon \CM_1 \to T[1]\CM_2$ defined by means of $f^*(q^\alpha)=a^*(q^\alpha)$ and $f^*(\ud q^\alpha)=\CF (q^\alpha)$;\footnote{This was found later in the paper \cite{KS07}.} 
 equipping $T[1]\CM_2$ with an appropriate canonical Q-structure, $f$ becomes a Q-morphism for \emph{every} map $a \colon \CM_1 \to \CM_2$ and generalizes the Chern-Weil map in the case of non-trivial Q-bundles as is useful for constructing corresponding characteristic classes. In particular, $f^* \colon \smooth(T[1]\CM_2 \to \Omega^\bullet(\Sigma))$ is an algebra morphism then, $f^*(x^i\xi^a \ud \xi^D) = X^i A^a \wedge F^D$ etc. 

\newpage

\section{Q-structures for $p=2$ --- Lie 2-algebroids}
\label{s:examQ}
In this section we study general vector fields of the form
\eqref{Qi}--\eqref{QB}. Dropping the index 2 for notational
convenience within \emph{this} section,
\beq \label{Q} Q = \rho^i_a \xi^a \frac{\partial}{\partial x^i}
- \2 C^a_{bc} \xi^b \xi^c \frac{\partial}{\partial \xi^a} +t^a_D b^D
\frac{\partial}{\partial \xi^a}  -\Eco^D_{aC} \xi^a b^C
\frac{\partial}{\partial b^D}
+ \frac{1}{6} 
        H^D_{abc} \xi^a \xi^b \xi^c \frac{\partial}{\partial b^D} \eeq
is the most general degree one vector field on a degree two
supermanifold (cf.~also footnote \ref{fn:pos}). 
The range of indices is $i = 1, \ldots, n$, $a = 1,
\ldots r$, $D=1 \ldots s$, respectively,
where, in previous notation,  $n = \dim M$, $r = \dim V$, and $s  = \dim W$.
Note that $C^c_{ab}$ and $H^D_{abc}$ are completely
antisymmetric in their lower indices by construction. 
Requiring $Q$ to be homological, that is to square to
zero and thus to define a Q-structure, one obtains the following identities to
hold true: 
\begin{align}
  \rho^j_{[a}\rho^i_{b],j} -\frac12\rho^i_c C^c_{ab}  &= 0  \label{B1}\\
  \rho^i_a t^a_D &= 0  \label{B2}\\
  3C^e_{[ab} C^d_{c]e} +3\rho^i_{[c} C^d_{ab],i} - t^d_D H^D_{abc} &= 0
    \label{B3}\\
  \rho^i_c t^a_{D,i} -\Eco^F_{cD} t^a_F -C^a_{bc}t^b_D &= 0  \label{B4}\\
  \rho^i_{[b}\Eco^D_{a]F,i}  +\frac12 \Eco^D_{eF}C^e_{ab}
    -\Eco^D_{[aC}\Eco^C_{b]F} +\frac12H^D_{abc}t^c_F &= 0  \label{B5}\\
  t^a_{(A} \Eco^D_{aC)} &= 0  \label{B6}\\
  \Eco^D_{[aF}H^F_{bcd]} +\rho^i_{[a}H^D_{bcd],i} -\frac32H^D_{e[ab}C^e_{cd]}  &= 0  \label{B7}
\end{align}
The first two equations result upon twofold application of $Q$ to
coordinates $X^i$, the next three to $\xi^a$, and the last two when
$Q^2$ is applied to $b^D$ (the result being put to zero
in all cases). Round/square brackets indicate symmetrization/antisymmetrization
of enclosed
indices of equal type (in Eq.~\eqref{B5}, e.g., this concerns the indices $a$
and $b$ in the respective two terms, and only them).

We now discuss the meaning of these equations in more standard,
index-free terms. We start with $M =\pt$ and $s=0$, i.e.~there are no $X^i$
and $b^D$ coordinates. Correspondingly,  in \eqref{Q} only the term
containing $C^c_{ab}$  survives. From the above set of equations, the
only nontrivial one is Eq.~\eqref{B3}, reducing to $C^e_{[ab}
C^d_{c]e}=0$, which is the Jacobi identity for the coefficients
$C^a_{bc}$. Thus in this simple case, in which the supermanifold
$\CM$ has a
trivial base and is actually only of degree one,
one finds the Q-manifold
$(\CM,Q)= (\lie[1],\ud_{CE})$, where $\lie$ is some Lie algebra and
$\ud_{CE}$ the corresponding Chevalley-Eilenberg
differential. Super vector spaces of degree one
with Q-structure are equivalent to
Lie algebras. 

Let us in a next step drop the restriction that $M$ is a point,
i.e.~we deal with a general degree one supermanifold (which,
certainly, may also be viewed as a particular degree two
supermanifold) with Q-structure. This leaves the first two
terms in \eqref{Q}. The corresponding equations, \eqref{B1} and
\eqref{B3} with $t^d_D \equiv 0$, are recognized as the
structural equations for a Lie algebroid: A Lie algebroid is a vector
bundle $E$ over $M$ together with a bundle map $\rho \colon E \to TM$,
called the anchor map, 
and a Lie algebra defined on the sections of $E$ satisfying the
Leibniz rule $[\psi, f
\varphi] = f [\psi, \varphi] + \rho(\psi) f \, \varphi$.\footnote{For
later use we mention that we call $E$ an \emph{almost Lie algebroid} if it
is defined as above with the mere difference that the bracket $[
\cdot, \cdot ]$ needs to
be antisymmetric only, but not necessarily to satisfy the Jacobi
identity. If $\rho \equiv 0$ in an (almost) Lie algebroid, the bracket
defines a fiberwise
product and $E$ becomes a bundle of (almost) Lie algebras.}
{} From these
data one can infer that $\rho$ is also a morphism of Lie algebras, $\rho([\psi, 
\varphi]) =  [\rho(\psi), \rho(\varphi)]$. Choosing a local frame
$\xi_a$ in $E$ and local coordinates $x^i$ on $M$, the anchor gives
rise to structural functions $\rho^i_a=\rho^i_a(x)$ 
via $\rho(\xi_a) = \rho^i_a \partial_i$. Likewise, the bracket between
two sections from the basis induces $C^c_{ab}=C^c_{ab}(x)$:
$[\xi_a,\xi_b]= C^c_{ab} \xi_c$. Now it is easy to see that the
morphism property of $\rho$ implies Eq.~\eqref{B1} and the Jacobi
identity for the Lie bracket on $E$, together with the Leibniz rule,
yields (the respective remainder of) Eq.~\eqref{B3}. 

For the reverse direction, i.e.~to recover the index free geometrical
structure from the above formulas, one needs to address coordinate
changes on the Q-manifold $\CM$. By definition, coordinate changes on
a graded manifold need to be degree preserving; therefore, 
$\widetilde{x^i}= \widetilde{x^i}(x)$,
$\widetilde{\xi^a} = M^{\tilde a}_b(x) \xi^b$. Thus one deals with a
vector bundle $E$ over the base $M$ of $\CM$ and $\CM = E[1]$ (by
definition, for a vector bundle the shift in degree
concerns only the fiber coordinates). Implementing this coordinate
change on $Q$, we see that $\rho^i_a$ transforms like a tensor, but
$C_{ab}^c$ does not. Indeed, 
$\partial_i=\frac{\partial \widetilde{x^j}}{\partial x^i}
\widetilde{\partial_j} + \frac{\partial \widetilde{\xi^a}}{\partial
x^i} \widetilde{\partial_a}$, where 
$\partial_i = \partial/\partial x^i$, $\partial_a = \partial/\partial \xi^a$
etc. So, there is a contribution proportional to $\rho^i_a$ and the
$x$-derivative of $M^{\tilde a}_b$ to the transformation law of
$C_{ab}^c$. Since $C_{ab}^c \xi_c = [\xi_a, \xi_b]$ is to hold true
in every frame---after all we do not want to have frame-dependent
definitions of a bracket---this results in the Leibniz property of the
bracket. The rest is then obvious.

Summarizing, a degree one supermanifold with Q-structure, 
called Q1-manifold for simplicity,  is \emph{tantamount}
to a Lie algebroid, $(\CM,Q)=(E[1],{}^E\ud)$.\footnote{This was observed first by Vaintrob \cite{Vai97}, including also the compact formulation of Lie algebroid morphisms (a proof of equivalence of this definition with the standard, more complicated one of \cite{MaH93} can be found in \cite{BKS}). --- Recall \label{f:8}  that we 
restricted to non-negatively graded supermanifolds; Q-manifolds with 
this restriction are also called NQ-manifolds in the literature, cf., e.g., 
\cite{Sch93,AKSZ}.} Here ${}^E\ud$ is the
canonical differential that generalizes $\ud_{CE}$ from above for
$M=\pt$ and the de Rham differential for $E=TM$, $\rho = \id$,
the standard Lie algebroid (cf., e.g., \cite{CaWe99} for
${}^E\ud$  or just
use the above formulas for $Q$, reinterpreted correspondingly, as a
possible definition). 
 
In order to reveal the differential geometry contained in a general
``Q2-manifold'', i.e.~a degree two graded manifold with Q-structure, 
Eq.~\eqref{Q} together
with Eqs.~\eqref{B1}--\eqref{B7}, we follow the same strategy: We
first look at degree preserving coordinate changes on $\CM$,
\beq \widetilde{x^i}= \widetilde{x^i}(x), \qquad 
\widetilde{\xi^a} = M^{\tilde a}_b(x) \xi^b , \qquad
 \widetilde{b^D}  = N^{\tilde D}_C(x)  b^C +\frac12 L^{\tilde D}_
{ab}(x)\xi^a\xi^b  \, . \label{trafo} \eeq
 If there were no $\xi^a$ coordinates,
i.e.~$r=0$, similarly to before we would obtain $\CM = V[2]$ for some
vector bundle $V$ over $M$. The above $L$-terms imply that in general
$\CM$ is not just the direct sum of two vector bundles $E$ and $V$,
but instead one has the sequence of supermanifolds
\beq V[2] \to \CM \to  E[1] \, . \label{sequence} \eeq
(The first map is an embedding, characterized 
by $\xi^a =0$, and the second map a (surjective)
projection, dropping the $b^D$ coordinates). 
Only after the choice of a
section $\imath \colon E[1] \to \CM$, characterized locally by $b^D = L^D_{ab}(x) \xi^a \xi^b$, 
this gives rise to $\CM = E[1] \oplus V[2]$. The \emph{difference} between two 
sections is a 
section $B$ of $\Lambda^2 E^* \otimes V$---similarly to the fact that the
difference between two connections in a vector bundle
corresponds to a tensorial object.
In what follows the choice of a section 
shall be understood, so that in terms of ordinary differential geometry 
there is a total space corresponding to the Whitney sum  
$E \oplus V$ of two vector bundles over $M$. However, let us mention that when the Q-manifold under discussion carries 
additional structures, like a compatible symplectic form, this may not be the best adapted description. 


So, under the above assumption, a (non-negatively) graded manifold of (maximal) degree two gives rise to
a vector bundle $E \oplus V$ 
over $M$. Next we
determine the transformation properties of the coefficients in
\eqref{Q}. For later purposes we keep also $L$-contributions,
disregarding them only afterwards within the rest of this
section. With \eqref{trafo} one
obtains:
\begin{eqnarray}
  \widetilde{\rho^i_a} &=& \frac{\partial\tilde x^\imath}{\partial
x^j} \rho^j_b (M^{-1}) ^b_{\tilde a}  \label{trafo1} \\
  \widetilde{C_{ab}^c} M^{\tilde a}_d M^{\tilde b}_e     
   &=& M^{\tilde c}_f C_{de}^f
    -2M^{\tilde c}_{d],i}\rho^i_{[e}
    +M^{\tilde c}_f t^f_D (N^{-1})^D_{\tilde E} L^{\tilde E}_{de}
   \label{trafo2} \\
  \widetilde {t^a_D} &=& M^{\tilde a}_b t^b_C  (N^{-1})^C_{\tilde D}
  \label{trafo3}
  \\
  \widetilde {\Eco^C_{aE}} M^{\tilde a}_b N^{\tilde E}_D 
   &=&  N^{\tilde C}_E \Eco^E_{bD}
    -N^{\tilde C}_{D,i} \rho^i_b
    -L^{\tilde C}_{db} t^d_D  \label{trafo4} \\
  \widetilde {H^D_{def} } M^{\tilde d}_a M^{\tilde e}_b M^{\tilde f}_c
   &=& N^{\tilde D}_C H^C_{abc} 
    +3N^{\tilde D}_F \Eco_{a]C}^F (N^{-1})^C_{\tilde E} L^{\tilde E}_{[bc}
    +3\rho^i_{[a} L^{\tilde D}_{bc],i}
    -3L^{\tilde D}_{d[a}C^d_{bc]}
  \nonumber\\  &&
    -3N^{\tilde D}_{C,i}\rho^i_{[a} \big(N^{-1}\big)^C_{\tilde E}
    L^{\tilde E}_{bc]}
    -3L^{\tilde D}_{d[a} t^d_C (N^{-1})^C_{\tilde E} L^{\tilde E}_{bc]}
  \label{trafo5}
\end{eqnarray}
Note that $C_{ab}^c$, $H^D_{abc}$ and $\Eco_{aD}^C$ are affected by the affine 
$L$-contributions. Only the two quantities $\rho^i_a$ and $t_D^a$ are
not; they both correspond to bundle maps, namely $t \colon V
\to E$ and $\rho \colon E \to TM$ or, equivalently, to sections of
$V^* \otimes E$ and $E^* \otimes TM$, respectively.
 If $L$-terms are suppressed, $H$ is likewise just a 
section of $\Lambda^3 E^* \otimes V$. But the nature of all three
objects, $H$,
the bracket on $E$ induced by $C^c_{ab}$, and the object corresponding
to $\Eco_{aD}^C$ discussed further below 
can change if $L^D_{ab}$ is not forced to vanish. 
For the next steps we will first restrict to vanishing $L$-contributions, i.e.~to
a fixed choice of a section of \eqref{sequence}.

Even under this assumption, the nature of $\Eco_{aD}^C$ is more
intricate than that of $H$. 
As anticipated by the notation,  $\Eco_{aD}^C$ are
the coefficients of a kind of connection. More precisely, 
introduce an
$E$-connection $\Ena$ in $V$, i.e.\ for every $\psi \in \Gamma(E)$ an
operator $\Econn_\psi \colon \Gamma(V)
\to \Gamma(V)$ satisfying the Leibniz rule \mbox{\( \Econn_\psi (fv) =
\rho(\psi)f \cdot v +f \Econn_\psi v \)} such that the assignment $\psi
\mapsto \Econn_\psi$ is $C^\infty(M)$-linear in
$\psi$.\footnote{More abstractly, $\Econn_\cdot$ corresponds to a map from $E$
  to $A(V)$, the Atiyah algebroid of $V$. $A(V)$ is a vector bundle over $M$
  itself, namely the Lie algebroid
  corresponding to the canonical Lie groupoid of automorphisms of the
  vector bundle $V$.}
 This definition reduces to the one of an ordinary covariant
derivative in the case of $E=TM$. Then $\Eco_{aD}^C$ are just the
E-connection coefficients in a local frame, $\Econn_{\xi_a} b_D =
\Eco_{aD}^C b_C$; now Eq.~\eqref{trafo4}  follows from
\eqref{trafo} (both without the respective
$L$-contribution, \eqref{trafo} interpreted as a change of frames in
the two vector bundles) by means of the properties of $\Ena$ and vice versa.     

Finally the coefficients $C_{ab}^c$ induce a
bracket $[ \cdot , \cdot ]$ on sections of $E$  via $[\xi_a,\xi_b] :=
C_{ab}^c \xi_c$, which is antisymmetric by construction and, 
\emph{according to} \eqref{trafo2}
 without $L$-terms, satisfies a Leibniz
rule. 

Now, using the above data, we are to reinterpret the conditions
ensuring $Q^2=0$. In this way we arrive at:
\begin{prop} \label{prop:double}
A Q-structure on a supermanifold of degree two (a Q2-manifold)  
with section of the sequence \eqref{sequence} as described above 
is equivalent to a Lie 2-algebroid.
\end{prop}
\begin{vdef} \label{Lie2def} A \emph{Lie 2-algebroid} is a complex of vector bundles
\begingroup\makeatletter\ifx\SetFigFont\undefined%
\gdef\SetFigFont#1#2#3#4#5{%
  \reset@font\fontsize{#1}{#2pt}%
  \fontfamily{#3}\fontseries{#4}\fontshape{#5}%
  \selectfont}%
\fi\endgroup%
\beq   \mbox{
\setlength{\unitlength}{3947sp}%
\begin{picture}(1297,885)(601,-361)
\thicklines
{\put(1351,314){\vector( 1, 0){375}}
}%
{\put(1876,164){\vector(-2,-1){540}}
}%
{\put(751,314){\vector( 1, 0){375}}
}%
{\put(1276,164){\vector( 0,-1){300}}
}%
{\put(676,164){\vector( 2,-1){540}}
}%
\put(601,239){\makebox(0,0)[lb]{\smash{\SetFigFont{12}{14.4}{\rmdefault}{\mddefault}{\updefault}{V}%
}}}
\put(1201,239){\makebox(0,0)[lb]{\smash{\SetFigFont{12}{14.4}{\rmdefault}{\mddefault}{\updefault}{E}%
}}}
\put(1801,239){\makebox(0,0)[lb]{\smash{\SetFigFont{12}{14.4}{\rmdefault}{\mddefault}{\updefault}{TM}%
}}}
\put(1201,-361){\makebox(0,0)[lb]{\smash{\SetFigFont{12}{14.4}{\rmdefault}{\mddefault}{\updefault}{M}%
}}}
\put(826,389){\makebox(0,0)[lb]{\smash{\SetFigFont{12}{14.4}{\rmdefault}{\mddefault}{\updefault}{$t$}%
}}}
\put(1426,389){\makebox(0,0)[lb]{\smash{\SetFigFont{12}{14.4}{\rmdefault}{\mddefault}{\updefault}{$\rho$}%
}}}
\end{picture}
}
 \label{compl} 
\eeq 
together with an antisymmetric bracket $[ \cdot , \cdot ]$ on $E$, an $E$-connection $\Ena$ 
on $V$, and an ``anomaly'' $H \in \Gamma(\Lambda^3 E^* \otimes V)$ satisfying
the following compatibility conditions: 
\begin{align} [\psi_1,f\psi_2] &= \rho(\psi_1)f \, \psi_2 + 
f[\psi_1,\psi_2] \, ,  \label{Leibn}\\
  {}[\psi_1,[\psi_2,\psi_3]\,] +\textup{cycl(123)} &= t \big(H(\psi_1,
  \psi_2, \psi_3)\big) \, , 
    \label{Jac}\\{}
[\Econn_{\psi_1},\Econn_{\psi_2}]v - \Econn_{[\psi_1,\psi_2]}v &= 
H\big(\psi_1,\psi_2,t(v)\big) \, ,  \label{repr1}
\end{align}
as well as  $t(\Econn_\psi v) = [\psi, t(v)]$ and $\Econn_{t(v)} w = -
\Econn_{t(w)} v$. In addition, $H$ has to satisfy
\beq
  \Econn_{\psi_1} H(\psi_2,
  \psi_3, \psi_4) -\frac32 H([\psi_1,\psi_2],\psi_3,\psi_4)
  +\textup{Alt}(1234)= 0  \label{strange} \, .
\eeq
\end{vdef}
We will now make further remarks on the statement of the above
proposition and some parts of the above definition. In this way we
will arrive also at another, somewhat simpler form of the description of Q2-manifolds (or of what we just called Lie 2-algebroids), provided by Theorem
\ref{theo1} below.\footnote{In this part of the paper we provide a 
classical differential geometry description of a Q2-manifold. The relatively simple notion of a Q2 manifold in graded geometry encodes relatively involved data in the classical geometry setting. These are, however, chosen so as to mimic the standard definition of a Lie algebroid (corresponding to Q1-manifolds) and, simultaneously,  generalize existing definitions of strict and semi-strict Lie 2-algebras. It may be reasonable to call the above or any equivalent data in \emph{classical} differential geometry \emph{semi-strict} Lie 2-algebroids, so as to clearly distinguish from definitions like the ones given in the Appendix or increasingly popular conventions where Q$k$-manifolds are considered as a defintion of Lie $k$-algebroids (cf also the corresponding discussion in the Introduction following Theorem \ref{thm1.5}).}

In the above, $\psi_i \in \Gamma(E)$, $f \in
C^\infty(M)$, and $v,w \in \Gamma(V)$. \eqref{compl} defining a
complex just means that $ \rho \circ t = 0$, as corresponds to
Eq.~\eqref{B2}. The three projections going to one and the same
copy of the base $M$ means that the bundle maps $t$ and $\rho$ cover
the identity map on $M$. Eqs.~\eqref{Jac}, \eqref{repr1}, and
\eqref{strange} correspond to Eqs.~\eqref{B3}, \eqref{B5},
and \eqref{B7}, respectively; the two equations in the text are
\eqref{B4} and \eqref{B6}. From Eqs.~\eqref{Leibn} and
\eqref{Jac} together with $ \rho \circ t = 0$ one may show that
$\rho$ is a morphism of brackets:
\beq   \rho\big([\psi_1,\psi_2]\big) = [\rho (\psi_1), \rho (\psi_2)] \, , 
\label{rhomor}  \eeq
where the bracket on the r.h.s.~denotes the commutator of 
vector fields (i.e.~the bracket in the standard Lie algebroid $TM$). 
This equation, finally, reduces to \eqref{B1} in components.

Condition \eqref{Leibn} is the Leibniz rule for the bracket on
$E$, $t \circ H$ defining its Jacobiator according to
\eqref{Jac}. So, for $H=0$, or even for $H$ living only in the
kernel of $t$, $E$ becomes a Lie algebroid. As follows from
\eqref{repr1}, $H$ also governs the $E$-curvature $\ER$ of $\Ena$; indeed,
$$ \ER(\psi_1,\psi_2)  \equiv [\Econn_{\psi_1},\Econn_{\psi_2}] -
\Econn_{[\psi_1,\psi_2]} \quad , \qquad 
\ER \in \Gamma(\Lambda^2E^* \otimes \End(V)) \, , 
$$  where $[\Econn_{\psi_1},\Econn_{\psi_2}]$ denotes the ordinary commutator of derivations, 
is nothing but the generalization of
curvature to the context of $E$-connections in $V$ for the 
almost Lie algebroid $E$. Denoting by $\EO^p(M,V)$ the sections of
$\Lambda^p E^* \otimes V$, the $E$-$p$-forms with values in $V$, where $V$ is some
vector bundle over the base $M$ of $E$, $\ER$ becomes an $E$-2-form
with values in the endomorphisms of $V$. Note that we used
Eq.~\eqref{rhomor} to show 
that the operator $\ER(\psi_1,\psi_2)$ is
$C^\infty(M)$-linear and thus indeed defines an 
endomorphism of $V$. $H=0$ implies a flat E-connection, $ \ER=0$, which 
is a \emph{representation} of $E$ on $V$.  The condition
\eqref{strange} implies that the Jacobi condition for $[ \cdot ,
\cdot ]$  and the
representation property of $\Ena$ are violated in a relatively mild,
controlled way.

Eq.~\eqref{strange} can be rewritten in a more compact form.
Introduce the generalization of an exterior covariant derivative
$\ED \colon \EO^\cdot (M,V) \to \EO^{\cdot \, +1} (M,V)$ 
by means of a generalized Cartan 
formula
\begin{align} \label{genCar} \ED \omega (\psi_0, \ldots , \psi_p) = 
  &  \sum_i (-1)^i \Econn_{\psi_i}
  \omega(\psi_0,\ldots,\widehat{\psi_i},\ldots,\psi_p)  \nonumber \\
  & +\sum_{i<j} (-1)^{i+j} \omega\big( [\psi_i,\psi_j], \psi_0, \ldots,\widehat{\psi_i},\ldots\widehat{\psi_j},\ldots,\psi_p \big)
  \, .  
\end{align}
Note that this defines a map in the indicated manner only because
$\Ena$ and the bracket satisfy a Leibniz rule---with the same anchor
$\rho$. 
Then the condition to be satisfied by $H$ is nothing but simply
$\ED H=0$. 

At this point it is worthwhile to return to the question of changes 
of the section of \eqref{sequence}. As
already mentioned above it corresponds to an $E$-2-form $B$ with
values in $V$. For this purpose we thus need to interpret
Eqs.~\eqref{trafo1} to \eqref{trafo5} for the case that
both $M^{\tilde a}_b$ and
$N^{\tilde D}_C$ are the identity and $L$ equals $B$.\footnote{Note
that this is
\emph{not} the only option one has. Indeed, $L$ may be a function of
$M$ and $N$ and then changes of the frame in $E$ and $V$ induced by
$M$ and $N$ result in a different nature of e.g.~the bracket, which
then will in general \emph{not} satisfy a Leibniz rule as before. This
is what happens in the case of the Courant algebroid, which we will
discuss briefly in a later section; here additional structures, present
only in that case, give rise to a natural unique relation of $L$ to
changes of frames.} This induces the following changes of the bracket,
the $E$-connection, and the anomaly: 
\begin{eqnarray} [\psi_1,\psi_2] &\mapsto&\widetilde{ [\psi_1,\psi_2]}
\equiv
 [\psi_1,\psi_2] + t(B(\psi_1,\psi_2)) \label{split1}
\\ \Econn_\psi v &\mapsto& ^E{}\widetilde{\nabla}_{\!\psi\,} v \equiv  \Econn_\psi v +
B(\psi,t(v))\label{split2} 
\\ H &\mapsto& H +\ED B -C \, ,  \label{split3}
\end{eqnarray}
where $C(\psi_1,\psi_2,\psi_3) \equiv B(t(B(\psi_1,\psi_2)),\psi_3) +\cycl$
Note in this context that $\ED$ does not define a differential for two
reasons: first, $\Ena$ has a non-zero curvature, and, secondly, the
bracket on $E$ does not satisfy the Jacobi condition---and these two
facts do \emph{not} cancel against one another. Correspondingly, the
anomaly $H$ up to changes of the section is not just an element of a
cohomology. The new and the old $H$ are closed only w.r.t.~the
new and the old  $E$-derivatives, $\ED$ and  
$^E\!\!{}\widetilde{D}$, respectively. 


Finally,
$\Ena$ and $t$ may be used to define naturally a bracket on $V$ according to 
\([v,w]_V:= \Econn_{t (v)} w \). According to Prop.~\ref{prop:double}, 
this bracket is antisymmetric. Moreover, $\rho_V := \rho \circ t = 0$
defines a canonical 
anchor map from $V$ to $TM$ with respect to
which $[ \cdot , \cdot ]_V$ may be seen to 
satisfy a Leibniz rule, i.e.~just $[v, fw]_V = f[v,w]_V$. 
In this way also $V$ becomes an 
almost Lie algebroid, in fact a bundle of almost Lie algebras,
with Jacobiator determined by $H$ via
 \beq [v_1,[v_2,v_3]_V]_V +
\cycl(123) = H(tv_1,tv_2, tv_3) \, ,\label{Jac2} \eeq
as can be derived from the equations in Prop.~\ref{prop:double}.  
Also, $t \colon V \to E$ may be seen to be a morphism of brackets, 
$t\big([v,w]_V\big) = [t (v), t(w)]$. Finally, it remains to check in what sense $\Econn_\psi$ 
respects the bracket on $V$; one obtains
\beq \Econn_\psi ([v,w]_V) = [ \Econn_\psi (v),w]_V + [v, \Econn_\psi (w)]_V + 
H(\psi, t(v),t(w)) \, . \label{der} \eeq
In summary, 
\begin{thm} \label{theo1}
A Q2-manifold with section or a Lie 2-algebroid is equivalent to an
almost Lie algebroid $(E,[ \cdot , \cdot ])$ and a bundle of almost
Lie algebras $(V,[ \cdot , \cdot ]_V)$ together with a morphism $t
\colon V \to E$, covering the identity on $M$, an $E$-connection on
$V$ such that \beq \Econn_{t(v)} w= [v,w]_V \quad , \qquad t(\Econn_{
\psi} v) = [\psi, t(v)] \, , \label{repr2} \eeq as well as an
anomaly $H\in \EO^3(M,V)$ such that the Jacobiators on $E$ and $V$ are
$t \circ H$ and $H \circ t$, respectively, and the representation and
derivation property of $\Ena$ are violated according to the
$E$-curvature $\ER \cdot = H( \, , \, , t(\cdot))$ and
Eq.~\eqref{der}, respectively.  $H$ fulfills the consistency
condition $\ED H=0$.

A change of section of \eqref{sequence} corresponds to $L \in
\EO^2(M,V)$ and induces Eqs.~\eqref{split1} to \eqref{split3} as
well as $[v,w]_V \mapsto [v,w]_V + L(t(v),t(w))$, with all remaining
structures unmodified. 
%
\end{thm}
Instead of the first three lines we could have said also 
``A Q2-manifold with section is equivalent to a complex of almost 
Lie algebroids 
\begingroup\makeatletter
\ifx\SetFigFont\undefined%
\gdef\SetFigFont#1#2#3#4#5{%
  \reset@font\fontsize{#1}{#2pt}%
  \fontfamily{#3}\fontseries{#4}\fontshape{#5}%
  \selectfont}%
\fi\endgroup%
\beq \mbox{
\setlength{\unitlength}{3947sp}%
\begin{picture}(2120,885)(601,-361)
\thicklines
{\put(1276,314){\vector( 1, 0){300}}
}%
\put(1351,389){\makebox(0,0)[lb]{\smash{\SetFigFont{12}{14.4}{\rmdefault}{\mddefault}{\updefault}{$t$}%
}}}
{\put(2251,314){\vector( 1, 0){300}}
}%
\put(2251,389){\makebox(0,0)[lb]{\smash{\SetFigFont{12}{14.4}{\rmdefault}{\mddefault}{\updefault}{$\rho$}%
}}}
{\put(1726,164){\vector( 0,-1){300}}
}%
{\put(906,172){\vector( 2,-1){749.200}}
}%
{\put(2699,169){\vector(-2,-1){870}}
}%
\put(601,239){\makebox(0,0)[lb]{\smash{\SetFigFont{12}{14.4}{\rmdefault}{\mddefault}{\updefault}{(V,[.,.])}%
}}}
\put(1651,-361){\makebox(0,0)[lb]{\smash{\SetFigFont{12}{14.4}{\rmdefault}{\mddefault}{\updefault}{M}%
}}}
\put(2626,239){\makebox(0,0)[lb]{\smash{\SetFigFont{12}{14.4}{\rmdefault}{\mddefault}{\updefault}{(TM,[.,.])}%
}}}
\put(1651,239){\makebox(0,0)[lb]{\smash{\SetFigFont{12}{14.4}{\rmdefault}{\mddefault}{\updefault}{(E,[.,.])}%
}}}
\end{picture}
}
  \label{compl2} 
\eeq
together with an $E$-connection on $V$ such that \ldots''. 
In this sequence $TM$ is regarded as 
the standard Lie algebroid over 
$M$ and arrows are morphisms. Since $t$ and $\rho$ cover the identity
map, the 
morphism of almost Lie algebroids then implies that
$\rho_V = \rho \circ t$ and this has to  vanish due to the complex 
property. Note also that we did not try to formulate minimal 
conditions: e.g.~Eqs.\ \eqref{Jac2} and  \eqref{der} 
may be established from the rest as obvious from our derivation. 

Of particular interest is certainly the special case of vanishing 
anomaly, where the situation simplifies considerably: 
\begin{prop} 
 A Q2-manifold with a section such that $H =0$ is equivalent to a 
complex of Lie algebroids 
\begingroup\makeatletter
\ifx\SetFigFont\undefined%
\gdef\SetFigFont#1#2#3#4#5{%
  \reset@font\fontsize{#1}{#2pt}%
  \fontfamily{#3}\fontseries{#4}\fontshape{#5}%
  \selectfont}%
\fi\endgroup%
\beq \mbox{
\setlength{\unitlength}{3947sp}%
\begin{picture}(1297,885)(601,-361)
\thicklines
{\put(1351,314){\vector( 1, 0){375}}
}%
{\put(1876,164){\vector(-2,-1){540}}
}%
{\put(751,314){\vector( 1, 0){375}}
}%
{\put(1276,164){\vector( 0,-1){300}}
}%
{\put(676,164){\vector( 2,-1){540}}
}%
\put(601,239){\makebox(0,0)[lb]{\smash{\SetFigFont{12}{14.4}{\rmdefault}{\mddefault}{\updefault}{V}%
}}}
\put(1201,239){\makebox(0,0)[lb]{\smash{\SetFigFont{12}{14.4}{\rmdefault}{\mddefault}{\updefault}{E}%
}}}
\put(1801,239){\makebox(0,0)[lb]{\smash{\SetFigFont{12}{14.4}{\rmdefault}{\mddefault}{\updefault}{TM}%
}}}
\put(1201,-361){\makebox(0,0)[lb]{\smash{\SetFigFont{12}{14.4}{\rmdefault}{\mddefault}{\updefault}{M}%
}}}
\put(826,389){\makebox(0,0)[lb]{\smash{\SetFigFont{12}{14.4}{\rmdefault}{\mddefault}{\updefault}{$t$}%
}}}
\put(1426,389){\makebox(0,0)[lb]{\smash{\SetFigFont{12}{14.4}{\rmdefault}{\mddefault}{\updefault}{$\rho$}%
}}}
\end{picture}
}  
  \label{compl1} 
\eeq 
together with a representation $\Ena$ of 
$E$ on $V$ satisfying the compatibility conditions \eqref{repr2}. 
\label{cor:Liecomp}
\end{prop}
In the notion of a representation of one Lie algebroid on another one we 
may include $\Econn_\psi ([v,w]_V) = [ \Econn_\psi (v),w]_V + [v,
\Econn_\psi (w)]_V$, or otherwise we may even derive 
this condition from the other assumptions in the proposition.

It remains to specialize to the case of $M=\pt$, 
i.e.~to ``Q-vector spaces''. E.g.~then Prop.~\ref{cor:Liecomp} reduces to 
\begin{prop} A Q2-vector space with section such that $H=0$ is 
equivalent to a Lie 2-algebra. 
\end{prop}
\begin{vdef} A Lie-2-algebra (a differential crossed module) 
is a pair $(\lie,\h)$ of Lie algebras 
together with a homomorphism $t \colon \h\to \lie$ and a  representation
$\alpha \colon \lie\to\Der(\h)$ such that $\, t((\alpha x)v) =[x,t(v)]$   
and $ \alpha(t(v))w =[v,w]_\h$  for all $x \in \lie$, $v,w \in \h$ hold true.  
\end{vdef}
Here we adapted to the notation of \cite{Baez02jn}. The relation is 
established easily via $\lie =E$, $\h = V$ and $\alpha(x) \equiv \Econn_x$ for 
$x \in \lie \cong \Gamma(\lie)$, $\alpha_{aD}^C \equiv \Gamma_{aD}^C$.\footnote{$\alpha$ 
is not an arbitrary representation on $\h$, 
$\alpha \in \End(\h)$ and $[\alpha(x),\alpha(y)] = \alpha([x,y])$---the flatness property of $\Ena$, 
but the fact that $\h$ is a Lie algebra itself is taken into account 
by restricting to $\Der(\h) \subset  \End(\h)$: 
$\alpha(x) \in \Der(\h) \Leftrightarrow \alpha(x) ([v,w]) =  [\alpha(x)(v),w]
+[v,\alpha(x)(w)]$---which follows from our Condition~\eqref{der} with $H=0$.} 
The case of a Q2-vector space with non-vanishing $H$ (and fixed
section $\imath$) corresponds to a 
semistrict Lie 2-algebra in the nomenclature of Baez, cf.~e.g.~\cite{Baez03vi}. 


\newpage

\section{Gauge transformations} \label{s:gauge}
Here we introduce a notion of infinitesimal gauge
transformations. The present section is intended to
parallel Section~\ref{s:Bian}, where we introduced an as general as possible
notion of field strength of gauge fields, imposing conditions that a
physicist would generically consider mandatory---including certainly
ordinary Yang--Mills gauge theories as particular example, but also
known generalizations like those of \cite{Str04b}, \cite{SaWei:2005bp,SaWit:2005hv}, \cite{BM05}, \cite{ACJ05} or \cite{BS05}. 
Such as we required the field strengths $F^\alpha$ to start with $\ud
A^\alpha$, likewise now we want to require that the gauge
transformations 
for $A^\alpha$ start with $\ud \epsilon^\alpha$ (except those for the zero
forms $X^i$ certainly). Again this is satisfied in all the 
examples we just mentioned---at least after choosing appropriate parameters. Correspondingly, infinitesimal gauge
transformations are
parameterized by quantities $(\epsilon^\alpha) :\equiv
(\epsilon^a,\mu^D)$ which are differential
forms on $\Sigma$ of form degree one less
than the respective $A^\alpha$ (so $\mu^D$ are 1-forms); due to the index
$\alpha$, in general $\epsilon^\alpha$ should be considered as functions of
$X^i(\s)$ as well---otherwise formulas are at least frame
dependent---but we postpone this issue to a more careful
discussion below. We thus arrive at the following rather general ansatz: 
\begin{align}  \delta_\e X^i &= \bar\rho^i_a  \e^a  \nonumber\\
  \delta_\e A^a &= \ud\e^a +\bar C^a_{bc}A^b \e^c -\bar t^a_D \mu^D +O(F^i) \nonumber\\
  \delta_\e B^D &= \ud\mu^D -\bar\Eco^D_{aC} \e^a B^C +\bar\Eco^{\prime\,D}_{aC}
    A^a \mu^C +\frac12 \bar H^D_{abc} A^a  A^b \e^c\,
    +O(F^i,F^a) \,.  \label{gsymm0}
\end{align}
Here $O(F^\alpha)$ is supposed to denote contributions that vanish for
(the respective)
$F^\alpha=0$. They are required already to obtain frame-independence of
formulas, while there exist several options to achieve this 
(cf.~also the discussion in \cite{BKS}). Moreover, there may be
additional, model dependent contributions necessary to render a
particular action functional $S$ invariant w.r.t.~gauge transformations.
Such contributions may use additional data present in $S$,
including even a metric $h$ on $\S$; an illustration is provided by
the Dirac sigma model \cite{DSM}, where the $A^a$ variations contain
contributions proportional to $\ast F^i$, where $\ast$ denotes the
Hodge duality operation induced by $h$ on $\S$, which is 2-dimensional
in this case (moreover, such a contribution makes sense only in this
dimension). 
In this way one
could even consider completely general
$F^a$ contributions to $\delta A^a$, or likewise
contributions to $\delta X^i$. We excluded them in the above (thus
only ``rather general'') ansatz, where we assume that the additions live in the ideal $\I$ and are homogeneous of the correct form-degree (without any additional structures); that is why, e.g., only $F^i$ is permitted to enter $\delta A^a$ in \eqref{gsymm0}. 

The coefficients in the ansatz~\eqref{gsymm0} are arbitrary
functions of $X(\s)$ at this point. The notation chosen anticipates
the result one obtains, where one simply drops the bars
(and also the prime in the third contribution to the last formula).

In analogy to the Bianchi identities \eqref{bian}, \eqref{ideal}
we require that gauge transformations of field strengths result into
combinations that \emph{contain} field strengths again, i.e.~we
require: $\delta_\e F^\alpha = \L_\b^\alpha 
F^\c$ for all possible choices of gauge parameters 
$\e^\alpha$ and field configurations $A^\alpha$---where  the coefficients
$\L_\b^\alpha$ depend 
linearly on $\e^\alpha$ (and, in principle, its first derivatives) and are
permitted to be also field dependent---or,  in other
words, we require
\beq  \boxed{\delta_\e F^\alpha|_{F^\b=0} \equiv 0 \quad \textrm{or}  \quad \delta_\e
\I \subset \I } \; .
\label{gauge}\eeq
So the ideal $\I$, defined in Section~\ref{s:Bian}, should not only
be invariant w.r.t.~the operator $\ud$ but also w.r.t.~the operator
$\delta_\e$ (for arbitrary choices of the gauge transformation
parameters $\e^a$ and $\mu^D$). 

Such relations are certainly always satisfied in all the examples
mentioned above. In particular, even if some of the $X^i$ fields are
considered as scalar fields (and not as part of the gauge field) and,
correspondingly, the respective $F^i = \ud X^i - \rho^i_a A^a$ then
has the interpretation of a covariant derivative of $X^i$, again it is
a minimal requirement on what one would usually like to
call a covariant derivative.
 
Following the same strategy as in the Section~\ref{s:Bian}, the
Condition~\eqref{gauge} is now used to restrict the coefficient
functions in the ansatz~\eqref{gsymm0}. There we obtained 
Theorem \ref{theo:nilpot}. Consequently, 
we now assume that $Q_2$
squares to zero in the definition of field strengths.\footnote{At the
end of the following proof we will
make a brief remark on what still can be said if we drop this
requirement.}
We obtain
\begin{thm} \label{theo:gauge1}  
For a spacetime dimension of at least 3, equality of the coefficients
in \eqref{gsymm0} with the respective unbared (and unprimed)
coefficient in the
field strengths \eqref{Fi}--\eqref{FB}---or, likewise, 
in the (nilpotent) super vector field $\Qt$
in Eqs.~\eqref{Qi}--\eqref{QB}---is necessary and
sufficient for the validity of \eqref{gauge}.
\end{thm}
Again, a likewise statement holds true for every choice of $p$, provided
only the dimension of $\Sigma$ exceeds $p$ by at least one (otherwise one
still has sufficiency, but necessity is maybe the more striking part). 
Note that we assumed the Bianchi
identities to be satisfied, so \emph{any}
contribution of the form $O(F^\alpha)$, i.e.~lying inside the ideal $\I$, to \eqref{gsymm0} does not
modify the result!

To prove the statement we will content ourselves with a halfway
Q-language at this point, not yet
introducing any further abstractness beyond what we already did
so far, while still avoiding any lengthy direct calculation in this manner.
In a condensed notation we may summarize the ansatz~\eqref{gsymm0}
as follows:
\beq \delta^{(0)}_\e A^\alpha \equiv
\ud \e^\alpha + \e^\b a^*V^\alpha_\b \label{gauge0}
\, ,
\eeq
where $V^\alpha_\b$ is a function on $\CM_2$ (of degree $|\alpha| - |\b| +1\ge
0$, vanishing otherwise as well as for $|\b|=0$),
collecting all the yet undetermined
coefficients of the ansatz. The superscript $(0)$
was used, moreover, so as to indicate that we dropped the
terms $O(F^\alpha)$, which are irrelevant for the present
purposes. We assume for simplicity, moreover,
that $\e^\alpha$ 
depends on $\s$ only (and not also
on fields).\footnote{Both steps, dropping $O(F^\alpha)$ terms and having $\e^\alpha$
depend on $\s$ only, imply, among others, that we fixed a frame or target space coordinates for
the calculation. But this is legitimate at this point.} We also kept
the slightly more suggestive notation $\ud$ instead of $Q_1$ here.

Now we are to calculate $ \delta^{(0)}_\e F^\alpha$, where, according to
\eqref{CF}, $F^\alpha= \ud A^\alpha -a^* Q_2^\alpha$ with  $Q_2^\alpha$ denoting 
 the (left) components of the super vector field
$Q_2$ in a basis $\partial_\alpha \equiv \partial /\partial q^\alpha$. A
rather simple calculation (using at this point only $\ud^2=0$) yields
\begin{eqnarray}\delta^{(0)}_\e F^\alpha  &= 
&  \ud \e^\b \, a^* (V^\alpha_\b - \partial_\b
Q_2^\alpha) \nonumber \\ 
  & & - \e^\b \, a^* \left[ \partial_\b (Q_2^\gamma \partial_\gamma
  Q_2^\alpha) + O(V^\alpha_\b - \partial_\b
Q_2^\alpha) \right] 
\nonumber \\ 
  & &- (-1)^{|\b|} \e^\b \CF(V^\alpha_\b) \, , \label{Fchangeold} 
\end{eqnarray}
where $ O(V^\alpha_\b - \partial_\b
Q_2^\alpha)$ denotes terms vanishing with $V^\alpha_\b = \partial_\b
Q_2^\alpha$ and, as before, $\CF \equiv \ud a^* - a^* Q_2$.
So, obviously the sum of the first two lines has to vanish.
For every point $\s_0 \in \S$ we may choose $\e^\alpha$ such that $\ud \e^\alpha(\s_0)$
is non-zero and $\e^\alpha(\s_0)$ vanishes and vice versa. Also the map $a \colon
T[1]\S \to \CM_2$ is at our disposal. Correspondingly,
if the dimension of $\S$ exceeds $p$, this requires
\beq V^\alpha_\b = \partial_\b
Q_2^\alpha \qquad \textrm{and} \qquad  \partial_\b (Q_2^\gamma \partial_\gamma
  Q_2^\alpha) = 0 \, \qquad \textrm{for} \quad
   |\b| \ge 1 \, 
   . \label{consistency}
   \eeq
The second equation is fulfilled identically for $(Q_2)^2=0$, while
the first one gives the claimed necessity. 

On the other hand, if we restrict ourselves to gauge transformations of the form \eqref{gauge0} (and in general only then!), one may even
\emph{conclude} $(Q_2)^2=0$ from \eqref{gauge}. In this case one needs
only $\dim \S \ge p+1$. The reason is that according to the second Equation~\eqref{consistency} $Q_2^\gamma \partial_\gamma
  Q_2^\alpha \equiv (Q_2 \circ Q_2)^\alpha$ can only be a function of degree zero, which contradicts with the fact that it must have a strictly positive degree by knowing that $Q_2$ has degree plus one. This gives the result summarized in Theorem \ref{thm:phys2}.

For convenience of the reader, we display the resulting
formulas 
explicitly: 
\begin{align}  \delta^{(0)}_\e X^i &= \rho^i_a  \e^a \, ,  \label{g1a}\\
  \delta_\e^{(0)} A^a &= \ud\e^a + C^a_{bc}A^b \e^c +t^a_D \mu^D \, , 
    \label{g2a}\\
  \delta_\e^{(0)} B^D &= \ud\mu^D -\Eco^D_{aC} \e^a B^C +\Eco^D_{aC}
    A^a \mu^C +\frac12 H^D_{abc} A^a A^b \e^c \, ,  \label{g3a}
\end{align}
while, according to \eqref{gsymm0} and the above result, the full
transformation then takes the form 
\begin{align}
  \delta_\e X^i &= \delta^{(0)}_\e X^i  \label{g1b}\\
  \delta_\e A^a &= \delta^{(0)}_\e A^a +F^i \lambda^a_i  \label{g2b} \\
  \delta_\e B^D &= \delta^{(0)}_\e B^D +F^i \lambda^D_i 
    +F^a \lambda^D_a   \,, \label{g3b}
\end{align}
for some coefficients $\lambda^\alpha_\b$ depending on $\e$ (and the gauge fields). We will
specify two covariant options for a choice of $\lambda$s below. Again we
mention that this is not yet the most general possible form; there may
be further model- and/or dimension-dependent contributions of the more
general form that they just vanish with vanishing $F^\alpha$\footnote{If, for example, we add a term containing $*F^\alpha$ linearly to, say, $\delta A^\beta$, $*$ denoting the Hodge dual with respect to some metric on $\S$, then the gauge transformation of $F^\beta = \ud A^\beta + \ldots$ produces a term linear in $\ud *F^\alpha$ from which, in general, we cannot say that it belongs to the ideal $\I$ (while it still would vanish together with $F^\alpha=0$).}---but in any case the above $\delta_\e$ (as
well as $\delta^{(0)}_\e$) is always an important part of such gauge
transformations as well.

\vskip\newpassage

In the following we will now refine the Q-language so as to also be
able to cover the gauge transformations of the form \eqref{g1a}--\eqref{g3a} or even \eqref{g1b}--\eqref{g3b}. We restrict ourselves
to infinitesimal ones and we will not strive to present their
geometrical meaning here (this is rather subject of a more
mathematical accompanying paper \cite{KS07}); the main goal here
is to provide a technical tool, which e.g.~permits one to calculate
the commutator of two gauge transformations in an efficient manner.
Again we want to be model independent here and thus we aim at
calculating the generic part of such a commutator,
$[\delta^{(0)}_\e,\delta^{(0)}_{\bar \e}] A^\alpha$;
every explicit calculation, making use of \eqref{B1}--\eqref{B7}, 
would be rather tedious and lengthy---the method presented
below\footnote{Cf.~also \cite{BKS} for the same idea, illustrated
there at the example of the
Poisson sigma model.} shortens the calculation considerably; also it
provides a relation to so-called derived brackets \cite{YKS96},
which were studied in the mathematical literature and which we will
then develop further for our particular case in a subsequent section. 

In Section~\ref{s:Bian} we found that $A^\alpha = a^* (q^\alpha)$, the
map $a \colon \Mo \to \Mt$ containing all the information about our
gauge fields. Likewise, we now could attempt
to cast gauge transformations into the form
$\delta_\e A^\alpha = (\delta_\e a)^* (q^\alpha)$ with $(\delta_\e
a)^* = a^* \circ  V_\e$  for some degree preserving vector field $V_\e$ on
$\Mt$ (parameterized somehow by $\e$). Observing moreover that $\e^\alpha =
(0,\e^a,\mu^D)$ naturally carries an upper index of $\Mt$, it is natural
to regard these coefficients as components of a vector field
on $\Mt$, $\e := \e^\alpha \partial_\alpha$. There are two obvious problems
with this ansatz: First, $\e$ should depend on $\s$ not just via the
map $X$ and, second, in this manner we would not be able to have $\ud = Q_1$ enter respective formulas. 

Both issues are cured simultaneously by the following construction
\cite{BKS,KS07}: we can trivially extend the map $a$ from before
to a map \beq a \colon \Mo \to
\M := \Mo \times \Mt \, , \label{aext}
\eeq
which, by definition, is the identity on the
first factor. For notational simplicity we do not introduce a new
symbol for this extended map, henceforth $a$ always corresponds to
this extended map; as before $A^\alpha = a^* (q^\alpha)$, since (local)
functions $q^\alpha$ on $\Mt$ can be regarded also as functions on $\M$
(together with coordinate functions $\s^\m$ on $\Mo$ they provide
coordinate functions on the product supermanifold $\M$). Obviously,
$\M$ naturally becomes a Q-manifold itself; just add the two nilpotent
vector fields to obtain
\beq Q = Q_1 + Q_2 \, . \eeq
This is again a nilpotent degree one vector field since by definition
$Q_1$ and $Q_2$ supercommute (i.e.~they anticommute). In other words,
given the two {complexes} $(C^\infty(\Mo), Q_1)$ and  $(C^\infty(\Mt),
Q_2)$, one forms a double complex in the standard way,  
$C^\infty(\Mo) \otimes C^\infty(\Mt) \cong C^\infty(\CM)$. In this
extended setting, the operator $\CF$ of
Equation~\eqref{CF} is replaced by
\beq \CF := \Qo \circ a^* - a^* \circ Q \qquad \Rightarrow \qquad \CF (q^\alpha) 
\equiv F^\alpha \, , \label{CFext}\ee
since essentially all remains the same only that the target
$(\Mt,\Qt)$ is replaced by $(\M,Q)$---merely with an additional restriction on 
\eqref{aext} that $a$ is the identity on the first factor.\footnote{In other
words, $a$ is a section of the bundle $\M \to \Mo$. Cf.~\cite{KS07}
for the obvious generalization to nontrivial (total) bundles, which,
despite its simplicity,
may be seen to generalize the notion of an ordinary
connection in a nontrivial principal bundle (while certainly
in this case a connection is not just
a globally defined Lie algebra valued 1-form on the base of that bundle).} Note that
this extended $\CF$ vanishes identically on functions coming from
$\Mo$ (i.e.~constant along $\Mt$).

The space of vector fields $\X_\bullet(\M)= \Gamma(T\M)$ of every Q-manifold $\M$ carries canonically the structure of a complex. If $\CM$ is a supermanifold of degree $p$
(cf.~Section~\ref{s:examQ}), the space $\X_\bullet(\M)$ has also elements of
negative degree, bounded from below by $-p$. The differential is
provided by the graded commutator with $Q$, $\ud_Q \colon
V \mapsto [Q,V]$, since the bracket
$[ \cdot , \cdot ]$ maps derivations of $C^\infty(\M)$ into itself,
$\ud_Q \equiv [Q,\cdot]$ is of degree one since $Q$ is, and
$(\ud_Q)^2=0$ due to $Q^2 = \frac{1}{2} [Q,Q]=0$ as one shows easily by means of
the (graded) Jacobi identity of the bracket.

For the purpose of gauge symmetries, we will be interested in a
subcomplex $(\Xv,\ud_Q)\subset (\XM,\ud_Q)$,  spanned by ``vertical''
vector fields, i.e.~vector fields parallel to $\Mt$.\footnote{Here we made
use of the fact that $\ud_{Q_1}$ maps $\Xv$ into itself: $\ue =\ue^\alpha
\partial_\alpha \in \Xv
\Rightarrow [Q_1,\ue] = Q_1(\ue^\alpha)
\partial_\alpha \in \Xv$.}
 This complex is
isomorphic to $(C^\infty(\Mo),Q_1) \otimes (\X_\bullet(\Mt),\ud_{Q_2})$,
as one verifies easily (cf.~also Eq.~\eqref{complex}  below for
more details). In this form we recognize that both of the potential problems of the ansatz 
of symmetries generated by vector fields are cured at the same time: while still dealing essentially with vector fields on the target, the coefficient functions now depend also on $\Mo$ and there is $Q_1\cong \ud$ in the game as well. 

Let $V \in \Xz$, a vertical vector field of degree zero,
and set $\delta A^\alpha = (a^* \circ V) (q^\alpha)$. Together with
\eqref{CFext} this implies
\beq \delta F^\alpha = ( \Qo \circ a^*\circ V - a^*\circ V \circ Q)(q^\alpha)
\equiv \CF ( V(q^\alpha)) + (a^* \circ \ud_Q(V))(q^\alpha) \, . \label{Ftrafo1}
\eeq
Thus if $\dim \S \ge p+1$, $\delta \I \subset \I$ is tantamount 
to requiring $V$ to be $\ud_Q$-closed. The solution found before in
Theorem \ref{theo:gauge1} corresponds to an exact $V$, i.e.~$V=\ud_Q
\ue$, where $\ue$ is 
a degree minus one vertical vector field,
$\ue \in \Xmo$. We will now make this more explicit (and thereafter
add a remark about the $\ud_Q$-cohomology).

The degree $k$ elements of
$C^\infty(\Mo)$ may be identified with $\Omega^k(\S)$ and $Q_1$ is the
de Rham differential $\ud$ on $\S$. Correspondingly,
\begin{align}
 (\Xmo,\ud_Q)
\cong & \; (\Omega^0(\S),\ud) \otimes (\X_{-1}(\Mt),\ud_{Q_2})
  \nonumber \\ &  \oplus
 (\Omega^1(\S),\ud) \otimes
(\X_{-2}(\Mt),\ud_{Q_2})  \; \oplus \ldots \nonumber \\
  &  \oplus  (\Omega^{p-1}(\S),\ud) \otimes
(\X_{-p}(\Mt),\ud_{Q_2})
  \, .  \label{complex}
\end{align}
Thus, for $p=2$, the most general $\ue \in \Xmo$ takes the form
\beq \ue = \ue^a \frac{\partial}{\partial \xi^a} + \ue^D
\frac{\partial}{\partial b^D} 
\eeq
with $\ue^a \in C^\infty(\S) \otimes C^\infty(M)\cong
C^\infty(\S\times M)$ and  $\ue^D = (\ue^D)_{\mathrm{main}} + \xi^a
\ue^D_a$, where $(\ue^D)_{\mathrm{main}}
\in \Omega^1(\S) \otimes C^\infty(M)$ and $\ue^D_a \in
C^\infty(\S\times M)$. Here we used that the degree zero
part of $C^\infty(\Mt)$ is $C^\infty(M)$ and that $\xi^a
\frac{\partial}{\partial b^D} \in \X_{-1}(\Mt)$; 
everything is understood as
defined only locally over the target base $M$ certainly. We now
introduce
\beq \e^a := a^* \ue^a \; , \quad  \mu^D \equiv \e^D:= a^* \ue^D \, ;
\eeq
thus obviously, $\e^a \in  C^\infty(\S)$ and $\mu^D \in
\Omega^1(\S)$.\footnote{Note: $\e^a(\s) = \ue^a(\s,X(\s))$---it was,
among others, this extra \emph{explicit}
$\s$-dependence that made it necessary to extend the target as
in \eqref{aext}: For a fixed choice of $X \colon \S \to M$ that is
an embedding this would not have been necessary, but, on the other
extreme side---which, however, is \emph{always} realized in ordinary YM
gauge theories---if the
image of $X$ is just a point in $M$, one now
still has the possibility for a nontrivial
$\s$-dependence of $\e^a$. The situation with $\mu^D$ is similar, but
a bit more involved: The coefficient $\ue^D_a$ corresponds to a part
of $\mu^D$ that is not only $X$, but also $A^a$--dependent. We mostly
will consider only the case with $\ue^D_a=0$.}
Now, to prove the above claim, we just
perform the computation:
\begin{align}
\left(a^* \circ \ud_Q(\ue)\right)(q^\alpha)& =
  a^* \left(\ue(Q_2^\alpha)\right)  + \left(a^* \circ Q\right) (\ue^\alpha) 
  \nonumber \\ 
& =  a^*(\ue^\b\partial_\b Q_2^\alpha) - \CF (\ue^\alpha) + \left(Q_1 \circ
a^*\right)  (\ue^\alpha) \nonumber \\
& = \ud \e^\alpha + \e^\b \, a^*(\partial_\b Q_2^\alpha) - F^\beta \, a^*
( \partial_\b \ue^\alpha ) \, .
  \end{align}
Comparison with Eqs.~\eqref{gauge0}, \eqref{consistency}, and 
\eqref{g1b}--\eqref{g3b} shows that indeed
\beq \boxed{\delta_\e A^\alpha = (a^* \circ [Q,\ue] )\, (q^\alpha)}
\label{gaugebox} \eeq
with particular coefficients $\lambda^\alpha_\b$ in
\eqref{g1b}--\eqref{g3b}, namely
\beq \lambda^\alpha_\b = 
-a^*(\partial_\b \ue^\alpha)\, , \label{lambdas1} \eeq
i.e., in 
detail, $\lambda^a_i = 
-a^*(\partial_i \ue^a)$, $\lambda^D_i = 
-a^*(\partial_i \ue^D) \equiv -
a^*\left(\partial_i (\ue^D)_{\mathrm{main}}\right) -
 A^a a^*(\partial_i \ue^D_a)$, and $\lambda^D_a = 
-a^*(\ue^D_a)$. For the particular case of prime interest that
$\ue^D =(\ue^D)_{\mathrm{main}}$, this simplifies to
\beq \lambda^a_i = 
-a^*(\partial_i \ue^a) \, , \quad  \lambda^D_i = 
-a^*(\partial_i \ue^D) \, , \quad \lambda^D_a =0 \, ,   \label{lambdas2}
\eeq
where $\lambda^D_i$ depends on the choice of $X$ but not also on
$A^a$.

Note that, given a fixed map $a$, i.e.~a
fixed choice of $A^\alpha$, we can always adapt for an arbitrary choice of
$\e^\alpha$ \emph{and}  $\lambda^\alpha_\b$ in \eqref{g1b}--\eqref{g3b}
by an appropriate choice of
$\ue^\alpha$. On the other hand, from the point of view of
Eq.~\eqref{gaugebox}, gauge transformations are governed by
``vertical vector fields'' on the target, which in turn determine, for
\emph{every} independent choice of $A^\alpha$, both $\e^\alpha$ and
$\lambda^\alpha_\b$. 

We use the opportunity for a side remark about \emph{one} alternative covariant
choice of $\lambda$s in  \eqref{g1b}--\eqref{g3b}. At the
expense of introducing a connection in both  $E$ and $V$, say those
used to define the field strengths \eqref{FaG} and \eqref{FBG},
one may choose
\beq \lambda^a_i=\e^b\Gamma^a_{ib} \qquad
\lambda^D_i=\mu^C\Gamma^D_{iC} \label{lambdas3}
\eeq 
besides $\lambda^D_a =0$. This  has the advantage that now already $\e^\alpha \equiv a^*(\ue^\alpha)$
determines $\lambda^\alpha_\b$ uniquely (in \eqref{lambdas1} also the first germ of $\ue^\alpha$ at the section $a$ is needed). To see that this choice is indeed
covariant, one can argue by comparison with Eqs.~\eqref{lambdas1},
\eqref{lambdas2}: By construction, Eq.~\eqref{gaugebox}
gives a covariant gauge transformation---note that here covariance
means that $\delta_\e$ is to act as a derivation, but the r.h.s.~of
\eqref{gaugebox} obviously does so since $[Q,\e]$ is a
vector field.\footnote{Cf.~also \cite{BKS} for a  more detailed
discussion of this issue.} In detail this implies that when we change
local frames in $E$ and $V$ according to \(\tilde{\xi^a}=M^a_b \xi^b\),
\(\tilde{b^D}=N^D_C b^C\), every transformation of $\lambda$s such as
those of Eqs.~\eqref{lambdas1}, \eqref{lambdas2}, namely
\begin{align}  \widetilde{\lambda^a_i} &= 
    M^{\tilde a}_b \lambda^b_i - M^{\tilde a}_{b,i}\e^b 
  \label{lambdatrafo1} \\
  \widetilde{\lambda^D_a} &=
    N^{\tilde D}_F \lambda^F_c
   (M^{-1})^c_{\tilde a}  \label{lambdatrafo2} \\
  \widetilde{\lambda^D_i} &= 
    N^{\tilde D}_C\lambda^C_i - N^{\tilde D}_{C,i}\mu^C
    \label{lambdatrafo3}
  \, ,
\end{align}
give a covariant gauge transformation in 
\eqref{g1b}--\eqref{g3b}.
(We suppressed writing pullback by $X$ here: E.g., $M^{\tilde a}_b$ in
\eqref{lambdatrafo1} is rather $X^*(M^{\tilde a}_b) \equiv M^{tilde a}_b(X)$ and likewise
$M^{\tilde a}_{b,i}$ denotes $X^*(\partial_i M^{\tilde a}_b)$.)
The above choice of $\lambda$s using background connections
obviously fulfills the transformation requirements
\eqref{lambdatrafo1}--\eqref{lambdatrafo3}. Note that the
method of obtaining these transformation properties is
incomparably simpler than an elementary one where, in a first step, one would 
compute all changes in \eqref{g1a}--\eqref{g3a} using 
\eqref{trafo1}--\eqref{trafo5}. 

In the considerations around Eq.~\eqref{Ftrafo1} we obtained only that
$V$ needs to be $\ud_Q$-closed, but not necessarily $\ud_Q$-exact. This
reconciles with the result of Theorem \ref{theo:gauge1} since 
underlying \eqref{gauge} is a particular form of
gauge transformations, namely the one specified in
Eqs.~\eqref{gsymm0}---starting with exact forms on $\S$.
This already gives the restriction
to trivial elements of the otherwise nontrivial cohomology of
$\ud_Q$. It may be interesting to relax this condition, however.

Note that the $\ud_Q$--cohomology on vertical vector fields
can be computed by means of an algebraic
version of the K\"unneth formula, in detail:  $H_{\ud_Q}^0(\M) =
\bigoplus_{k=0}^2 \left(H_{\textrm{deRham}}^k(\S) \otimes
H_{\ud_{Q_2}}^{-k}(\Mt)\right)$. E.g.~a change of coordinates on the
supermanifold $\Mt$ induces a symmetry only if it leaves the Q-structure invariant, i.e.~if its (degree zero) generator $V=V^\alpha \partial_\alpha$ commutes with $Q_2$. 
Thus such a generator corresponds to an element with $k=0$
in the previous sum. 


\vskip\newpassage

We now come to the announced application of the above framework for
gauge transformations. As a warm-up exercise we return to
Eq.~\eqref{Ftrafo1} and use it to
compute the explicit form of the change of field strengths $F^\alpha$
under the generic part $\delta^{(0)}_\e$ of gauge transformations of
our gauge fields. Needless to say that again every explicit
calculation would be very lengthy, making heavy use of
\eqref{B1}--\eqref{B7}. We first observe that locally and for every
fixed choice of $(\e^a,\mu^D)$ we can choose frames in $E$ and $V$
such that $\ue^a$ and $\ue^D$ are constant along $\Mt$---correspondingly,
in \emph{this} particular frame, $\delta_\e A^\alpha = \delta^{(0)}_\e
A^\alpha$ and we may make use of \eqref{Ftrafo1} with $V=\ud_{Q}\ue$:\footnote{Alternatively, here we may equally well
use \eqref{Fchangeold} together with Eqs.~\eqref{consistency}.}
\ba
  \delta^{(0)}_\e F^i &=& \e^a \, \rho^i_{a,j} \, F^j \, , \label{deltaFi} \\
  \delta^{(0)}_\e F^a &=&-C^a_{bc}\, \e^b F^c + \left(C^a_{cb,i}\, \e^bA^c
  -t^a_{D,i}\, \mu^D \right)F^i \, , \label{deltaFa} \\
  \delta^{(0)}_\e F^D &=& -\Eco^D_{aC} \, \e^a F^C + \left(\Eco^D_{aC}
  \, \mu^C  - H^D_{abc} \,
  \e^b A^c \right)  F^a \nonumber \\
  && + \left[\tfrac{1}{2} H^D_{abc,i}\, \e^a A^bA^c -\Eco^D_{aC,i}
  \left(\e^a B^C + \mu^CA^a \right) \right] F^i \, . \label{deltaFB}
\ea
For example, the first line was obtained as follows:
$\delta_\e F^i=\CF([Q,\ue]x^i)=\CF(\rho^i_a\ue^a) =
X^*\left(\partial_j(\rho^i_a\ue^a)\right) F^j$,
where in the last step we  used \eqref{FLeibniz}. 


We now come to the commutator of gauge transformations. For this, we first make
the following simple but important observation
\begin{thm} \label{theo:Lie}
Infinitesimal gauge transformations, written in form of
\eqref{gaugebox},
form a Lie algebra isomorphic to $(\Xmo/\ker\ud_Q,[\cdot ,\cdot]_Q)$,
with 
the derived bracket:
\beq [\ue,\bar \ue]_Q = [\ud_Q(\ue),\bar \ue] \equiv  [[\ue,Q],\bar
\ue] \, . \label{induced}
\eeq 
\end{thm}
To prove this we observe that infinitesimal gauge transformations
\eqref{gaugebox} are generated by  degree minus one vertical
vector fields \emph{modulo} elements that are annihilated by
$\ud_Q \equiv [Q, \cdot]$. Such elements are called \emph{coexact} in
\cite{YKS96}, in generalization of coexact differential forms. 
It is now easy to convince oneself that the commutator of two gauge
transformations \eqref{gaugebox} gives rise to a new bracket
between vector fields:
\beq 
 [\delta_\ue , \delta_{\bar \ue} ] A^\alpha = \left(a^* \circ
 [\ud_Q(\ue), \ud_Q(\bar \ue)]\right) (q^\alpha) =
 \left( a^* \circ \ud_Q ( [\ud_Q(\ue),\bar \ue])\right) (q^\alpha) \, ,
\label{comm}
 \eeq
where we made use of the (graded) Leibniz property of $\ud_Q$ 
w.r.t.~$[ \cdot , \cdot ]$---noting that
$\ud_Q(\ue)$ is of degree zero---and of 
$(\ud_Q)^2=0$. This proves \eqref{induced}. The induced or
\emph{derived} bracket
$[\ue,\bar \ue]_Q = [\ud_Q(\ue),\bar \ue]$ on $\Xmo$ is \emph{not}
a Lie bracket; however, as a result of a general construction,
recalled in the subsequent section, on the
quotient by the kernel of the differential $\ud_Q$ it is,
cf.~Proposition \ref{p:YKSder} below.

It is possible to characterize Q2-manifolds 
also by means of its derived bracket---as an alternative to the
description found in Section~\ref{s:examQ}. Since
we were now naturally led to this bracket, and since this formulation
is well-suited also for comparison with the usual formulation of what
is a Courant algebroid, we will reanalyze
Q2-manifolds from this perspective in the subsequent section.

The above theorem
implies in particular that if gauge transformations are
interpreted---in generalization
of \cite{BKS}---in the sense advocated in the present section, they close
off-shell---and this also in theories which in the \emph{usual}
terminology have an
\emph{open algebra} of gauge transformations. In fact,
as a kind of corollary of Theorem \ref{theo:Lie}, we will now verify
\begin{prop}  For the 
gauge transformations \eqref{g1a}--\eqref{g3a} one finds
\ba [\delta^{(0)}_\e,\delta^{(0)}_{\bar \e}] X^i &=&
\delta^{(0)}_{\widetilde \e}
X^i  \label{gcom1}\\{}
  [\delta^{(0)}_\e,\delta^{(0)}_{\bar \e}] A^a &=&
  \delta^{(0)}_{\widetilde \e}
  A^a -C^a_{bc,i} \e^b
    \bar \e^{c} F^i  \label{gcom2}\\{}
  [\delta^{(0)}_\e,\delta^{(0)}_{\bar \e}] B^D &=&  
  \delta^{(0)}_{\widetilde \e}
 B^D 
     +H^D_{abc} \e^b  \bar \e^{c} F^a \nonumber \\ 
  & & +\,\Eco^D_{aC,i}
     (\e^a\bar \mu^{C}-\bar \e^{a}\mu^C) F^i
  - H^D_{abc,i} \e^{a}\bar\e^b A^c F^i  \label{gcom3}
\ea
where $\,{\widetilde \e}^a 
\equiv C^a_{bc}   \e^b \bar \e^c$ 
and  
    $\,{\widetilde \e}^D
    \equiv  \Eco^D_{aC}(
 \e^a \bar \e^{C}-\bar \e^{a}\e^C ) 
 -  H^D_{abc}\e^a\bar \e^{ b}
  A^c \, .  $
\label{prop:onshellclosed}  
\end{prop}
The idea is simple: As when computing the variation of
field strengths w.r.t.~$\delta^{(0)}_\e$, for a given choice of 
$\e^\alpha \in C^\infty(\Mo)$ 
and given a particular coordinate system $q^\alpha$ on $\Mt$, we may
define $\ue := \e^\alpha \partial_\alpha$, which now is a vertical vector
field. In that coordinate system, we then have
$\delta^{(0)}_\e A^\alpha = \delta_\ue A^\alpha$ and likewise so for the
barred variations.\footnote{Here we
wrote  $\delta_\ue A^\alpha$ instead of  $\delta_\e A^\alpha$
as in \eqref{gaugebox} or \eqref{g1b}--\eqref{g3b}, so as to
make clear that the gauge transformation \eqref{gaugebox} is
parameterized not only by $\e^\alpha(\s)=\ue^\alpha(\s,A^\b(\s))$, but also by the
first derivative of $\ue^\alpha$
w.r.t.~$q^\b$---cf.~Eq.~\eqref{lambdas1}. This is in contrast to the
situation when choosing \eqref{lambdas3} in
\eqref{g1b}--\eqref{g3b}.} 
Thus also for
the commutators we have $[\delta^{(0)}_\e,\delta^{(0)}_{\bar \e}] A^\alpha
= [\delta_\ue,\delta_{\bar \ue}] A^\alpha$. Now we can apply
\eqref{comm} and Theorem \ref{theo:Lie}. This contains the ``heavy
calculation'', so to say (which, in the present formulation, is almost a
triviality). We are left with splitting  $\delta_{[\ue,\bar \ue]_Q}
A^\alpha$ into the part $\delta^{(0)}_{[\ue,\bar \ue]_Q} A^\alpha$ and, using
\eqref{lambdas1}, the remaining 
terms containing $F^\alpha$s according to \eqref{g2b} and \eqref{g3b}.
The main point is that even if $\ue$ and $\bar \ue$ are constant along
$\Mt$, their induced bracket is not! Let us be explicit about this:
\beq [ [\e^\alpha \partial_\alpha , Q_1 + Q_2 ], \bar \e^\b \partial_\b ] =
\bar \e^\b \e^\alpha  [[\partial_\alpha ,Q_2 ],\partial_\b ]
\eeq
since $\e^\alpha$ and $\bar \e^\b$ are constant along $\Mt$ and
$[Q_1(\e^\alpha) \partial_\alpha, \bar \e^\b \partial_\b ] =0$. Thus we need to
compute $[\partial_\alpha,\partial_\b ]_{\Qt}$, which is easily found to 
equal 
$(-1)^{|\beta|+1}(\partial_\alpha \partial_\b \Qt^\gamma) \partial_\gamma$. So,
\ba 
\widetilde \ue := [\e^\alpha \partial_\alpha,  \bar \e^\b \partial_\b ]_Q 
  &=& -(-1)^{|\beta|}\bar \e^\b \e^\alpha
(\partial_\alpha \partial_\b \Qt^\gamma) \, \partial_\gamma
\label{derivedsymm} \\
&\equiv&   \e^a \bar \e^b C^c_{ab} \partial_{c} 
\,+ (\e^a \bar \mu^D- \bar \e^a \mu^D ) \, \Eco^C_{aD}
 \partial_{C}  
 \, - \e^a\bar \e^{ b}H^D_{abc}
  \xi^c \partial_{D}
 \, , \nonumber 
 \ea
where we used the general form \eqref{Q} for $\Qt$ and $\e^D \equiv
\mu^D$. Despite the fact that $\ue$ and  $\bar \ue$
were vertical vector fields ``constant along $\Mt$'' (in the chosen
frame $\partial_\alpha$), the vector field $\widetilde \ue = \widetilde
\ue^\alpha \partial_\alpha$ has coefficients which \emph{do} depend
non-trivially on $\Mt$. By comparison of coefficients it is now easy to
conclude the above statement from
\beq [\delta^{(0)}_\e,\delta^{(0)}_{\bar \e}] A^\alpha
= [\delta_\ue,\delta_{\bar \ue}] A^\alpha = \delta_{\widetilde \ue} A^\alpha =
\delta^{(0)}_{\widetilde \e} A^\alpha + F^\b \widetilde \lambda_\b^\alpha
\, , 
\eeq
with $ \widetilde \e^\alpha = a^* \widetilde \ue^\alpha$
and $\widetilde \lambda_\b^\alpha = - a^*(\partial_\b \widetilde
\ue^\alpha)$.
It is seen also that the nontrivial
$\xi^a$--dependence of $\widetilde \ue^D$ is responsible for both, the
$A^a$-dependence of $\widetilde \e^D$ and the $F^a$-contribution to 
\eqref{gcom3}. 

We conclude this section by illustrating how the previous considerations can be  used to also compute the commutator of symmetries of the form  
\begin{align}  \delta^{}_\e X^i &= \rho^i_a  \e^a \, ,  \label{g1an}\\
  \delta_\e^{} A^a &= \ud\e^a + C^a_{bc}A^b \e^c +t^a_D \e^D +\Gamma^a_{ib}\e^bF^i  \label{g2bn} \\
  \delta_\e B^D  &= \ud\e^D -\Eco^D_{aC} \e^a B^C +\Eco^D_{aC}
    A^a \wedge \e^C +\frac12 H^D_{abc} A^a \wedge A^b \e^c  -  \Gamma^D_{iC}  \epsilon^C \wedge F^i \,, \label{g3bn}
\end{align}
induced by (ordinary, ``non-dynamical'') connections on the (target) bundles $E$ and $V$ 
and which result from implementing \eqref{lambdas3} (together with $\lambda_a^D=0$) into the general symmetries \eqref{g1b}--\eqref{g3b}. The symmetries on the 0-forms $X^i$ coincide with the ones in \eqref{g1a} and we thus can take recourse to \eqref{gcom1} for the corresponding commutator. Changes occur when calculating the commutator on the 1-form fields $A^a$ and the 2-form fields $B^D$, where it is useful to split the infinitesimal symmetries according to $\delta_\e= \delta^{(0)}_\e+ \delta^{(1)}_\e$, where $\delta^{(0)}_\e$ is given by \eqref{g2a} and \eqref{g3a}. In the following we restrict ourselves to the 1-form fields, leaving the rest as an exercise to the reader. Evidently \begin{equation}
\delta^{(1)}_{\e} A^a = \Gamma^a_{ib}\e^b F^i \, .
\end{equation}
Moreover, 
\begin{equation}
[\delta^{}_\e,\delta^{}_{\bar \e}] A^a 
= [\delta^{(0)}_\e,\delta^{(0)}_{\bar \e}]A^a 
 + [\delta^{(1)}_\e,\delta^{(1)}_{\bar \e}]A^a 
 + \left(\delta^{(0)}_\e \delta^{(1)}_{\bar \e}A^a 
    +\delta^{(1)}_\e \delta^{(0)}_{\bar \e}A^a  
	- (\e \leftrightarrow \bar{\e})\right)\, .
\end{equation}
The first commutator is given by \eqref{gcom2}. For the second commutator, one observes    $\delta^{(1)}X^i=0$ such that  $\delta_{\e}^{(1)}\left(\Gamma^a_{ib}\bar{\e}^b F^i\right)=\Gamma^a_{ib}\bar{\e}^b\delta_{\e}^{(1)}\!F^i= 
-\Gamma^a_{ib}\bar{\e}^b\rho^i_c \Gamma^c_{jd}\e^d F^j$. Since we also know how $F^i$ transforms with respect to $\delta^{(0)}$, cf.~Eq.~\eqref{deltaFi}, we are now left with merely collecting all terms so as to arrive at\footnote{There were some later developments about such commutators in the mean time. First, in \cite{Str09} the terms violating closeness were given a geometric interpretation in the case that $E$ is a Lie algebroid: An ordinary connection on $E$ induces an $E$-connection generalizing the adjoint representation of a Lie algebra \cite{Str04b}; its $E$-curvature factorizes over the anchor map  $\rho$,  $\ER^a{}_{bcd}=\rho_b^i \, S^a{}_{icd}$. With this one obtains $ [\delta_\epsilon,\delta_{\bar{\epsilon}}]A^a = \delta_{\widetilde{\epsilon}}^{}A^a+S^a_{ibc}\epsilon^b\bar{\epsilon}^{c}F^i$.
Second, the methods used here were extended so as to perform calculations about such commutators more directly, replacing $a^*$ by $f^*$---cf., e.g., \cite{SaS13,SLS14}.}
\begin{align}
 [\delta_\epsilon,\delta_{\bar{\epsilon}}]A^a =& \delta_{\widetilde{\epsilon}}^{}A^a
 -(C^a_{bc,i}+\Gamma_{id}^aC^d_{bc} ) \epsilon^b\bar{\epsilon}^{c}F^i
  \\
 & +\left(C^a_{db}\Gamma^d_{ic} -\Gamma^a_{jb}\rho^j_{c,i}-\Gamma^a_{ib,j}\rho^j_c
 +\Gamma^a_{jb}\rho^j_d\Gamma^d_{ic} \right)( \epsilon^b\bar{\epsilon}^{c}-\bar{\epsilon}^{b}\epsilon^c)F^i , \nonumber
\end{align}
where the new parameter $\widetilde{\epsilon}$ is given by 
Proposition \ref{prop:onshellclosed}. Note that the part of the symmetries on the $A^a$ fields parameterized by the 1-forms $\e^D\equiv \mu^D$, cf.~Eq.~\eqref{g2bn}, do commute
evidently. Still, these symmetries do not commute with those parameterized by the 0-forms $\epsilon^a$. This explains why on the r.h.s.~the parameter $\widetilde{\epsilon}^D$ is non-vanishing in general---in fact, as we learn from Proposition \ref{prop:onshellclosed}, I 1-form part of the symmetries is even induced if the initial two symmetries $\epsilon$ and $\bar{\epsilon}$ have vanishing 1-form parts (like in a central extension).  


%

\newpage

\section{Derived bracket and extended Courant algebroids} \label{sec:derived}
In Section~\ref{s:examQ}
we studied Q-manifolds $\CM$ in terms of ordinary differential
geometric notions by analyzing the components of the vector field $Q$
in some chosen (graded) coordinate system, taking into account their
transformation properties. There is, however, also an alternative
procedure using the space of vector fields on $\CM$ with its graded
commutator bracket as well as its derived one, as is inspired by the
previous section on gauge symmetries. While in the case of
Q1-manifolds or Lie algebroids this provides an alternative route
leading to an \emph{evidently} equivalent picture, for Q2-manifolds
one obtains an alternative description of Lie 2-algebroids, closer to
the description of Courant algebroids (they will be a particular case
of what we will call vector bundle twisted Courant algebroids), with a
nontrivial translation of one picture into the other one,
cf.~definition \ref{def:VCourant} and, e.g., Theorem \ref{thm:der}
below (or, equivalently, Theorem \ref{thm1.6} in the summary).

As mentioned, the bracket between gauge symmetries \eqref{induced} is
a so called derived bracket \cite{YKS03en}. In Subsection~\ref{sec:Der}
we first recall some generic features of derived brackets, starting
from a general differential graded Lie algebra (DGLA). We extend the
considerations of the literature to the context of algebroids by
introducing an anchor map; vector fields on a Q$p$-manifold equipped
with the differential $\ud_Q$ are an example of this construction for
arbitrary $p \in \N$. We show that applying this perspective to
$p=1$ reproduces the respective considerations of section
\ref{s:examQ} (i.e.~it provides the usual definition of a Lie 
algebroid) in a rather straightforward way. We also provide the
definition of Courant algebroids in this subsection as another example
of the derived bracket construction. 

Inspired by the definition of Courant algebroids and the derived
bracket construction applied to vector fields on Q$p$-manifolds for
arbitrary $p$ we are led to define Vinogradov algebroids and vector
bundle twisted Courant algebroids in Subsection~\ref{subsec:Vino}. The
former ones are also inspired by the Vinogradov bracket (which
generalizes the Courant or Dorfman bracket to higher form degrees). 
In vector bundle twisted
Courant algebroids, which 
are special Vinogradov algebroids, the inner product in the definition of an
ordinary Courant algebroid is replaced by a bilinear form taking
values in a second vector bundle $V$ over the same base. Although in part
inspired by Q-manifolds, it is not clear that all Vinogradov and V-twisted
Courant algebroids have a description in terms of Q-manifolds. 

In Subsection~\ref{subsec:Lie2} we show that particular V-twisted
Courant algebroids are in one-to-one correspondence to (split)
Q2-manifolds, and, consequently (cf.~Theorem \ref{theo1}) to Lie
2-algebroids. This equivalence, cf.~Theorem \ref{thm:der} below, was
suggested by the changed perspective on Q-manifolds as stemming from
the gauge transformations in our context and would not be that obvious
to invent without such a motivation. The subsection also contains a
classification of these type of V-twisted Courant algebroids, which is
somewhat analogous to the one of exact Courant algebroids.

\subsection{Derived brackets and anchored DGLAs}
\label{sec:Der}
The prototype of a derived bracket goes back to E.~Cartan. Consider the space of differential 
forms $\Omega^\cdot(M)$. Besides the de Rham operator $\ud$ we can also consider contractions 
$\imath_v$ with a vector field $v$ as an operator. The Lie derivative  $\L_v$ 
then results as the graded commutator $\L_v = [\imath_v, \ud]$. On the other hand, certainly, 
$[\L_v, \imath_w] = \imath_{\L_v(w)}$, by the Leibniz rule of the Lie derivative. Using 
$\L_v(w)=[v,w]$ and the above formula for the Lie derivative, we see that the bracket $[v,w]$ 
of two vector fields can be viewed as arising from the left hand side of the following 
equation
\beq[[\imath_v,\ud], \imath_w] = \imath_{[v,w]}\,.\label{derbasic}\eeq
Note the similarity of the left hand side with Eq.~\eqref{induced}. In fact, the latter 
equation can be almost literally  translated into that setting: Consider the
supermanifold $T[1]M$ with the canonical Q-structure $\ud$ on it. Every vertical vector
field on $T[1]M$ corresponds to some $\imath_v$ for $v \in \Gamma(TM)$. It is a nontrivial
but easy-to-check fact that the result of the l.h.s.~of \eqref{derbasic} is a vertical
vector field again; therefore it comes from a vector field on $M$, defining a product
$[v,w]$. 

The above construction can be easily extended to include also exterior multiplication 
with differential forms. Let $\alpha,\beta \in \Omega^1(M)$ and compute\footnote{A.~Alekseev,
private communication 1999.} 
\beq [[\imath_v+\alpha\wedge,\,\ud],\,\imath_w+\beta \wedge]
  =\imath_{[v+\alpha,w+\beta ]_W} \, ; \label{Anton} \eeq
In this way one obtains the  (non-skewsymmetric) Courant, i.e.~the Dorfman bracket \cite{Cour90,Dor87,Dor93},
\beq [v\oplus \alpha,w\oplus \b]_W = [v,w]\oplus \L_v\b -\imath_w\ud\alpha \, . \label{Cour}
\eeq
Also note that the mere bracket,
$[\imath_v+\alpha\wedge,\imath_w+\beta \wedge]$ induces a natural
bilinear, symmetric and non-degenerate pairing: when acting on
$\Omega^\cdot(M)$ this operator reduces to multiplication with $\alpha(w)
+ \b(v)$. Viewing $v \oplus \alpha$ as a section of $W=TM\oplus T^*M$, one
can axiomatize \cite{Xu97,SevLett} rules that can be derived for these
operations  so as to end up with the definition
 
 \newpage
\begin{vdef} \label{def:Cour2} A {\rm Courant algebroid} is a vector bundle $W \to M$ together with a bundle map
$\rho \colon W \to TM$, a pairing of its sections $[ \cdot , \cdot ]_W \colon \Gamma(W)
\otimes \Gamma(W) \to \Gamma(W)$ as well as a non-degenerate fiber metric $\< \cdot, \cdot\>$
such that the following three identities hold true 
\begin{align}
  [\psi,[\phi_1,\phi_2]_W]_W &= [[\psi,\phi_1]_W,\phi_2]_W +[\phi_1,[\psi,\phi_2]_W]_W
    \label{cJacobi}  \\
  [\psi,\psi]_W &= \frac12\uD\<\psi,\psi\>  \label{cnSkew}\\
  \rho(\psi)\<\phi_1,\phi_2\> &=
    \<[\psi,\phi_1]_W,\phi_2\> + \<\phi_1,[\psi,\phi_2]_W\> 
\, ,  \label{cInvar}
\end{align}
where $\uD$ is induced by means of $\rho$ and the inner product via $\<\uD f , \phi\> =
\rho(\phi) f$.  
 \end{vdef}
The first axiom is the (left-)Leibniz rule of the bracket with respect
to itself. From \eqref{cInvar} one can also conclude \cite{LiB11} a
Leibniz rule w.r.t.~multiplication of sections by functions
\be [\psi,f\phi]_W = \rho(\psi)f\phi +f[\psi,\phi]_W
\label{cLeibniz} \, . \ee
The properties \eqref{cJacobi} and \eqref{cLeibniz} render  
$(W,\rho,[ \cdot , \cdot ]_W)$
what we want to call a Leibniz--Loday algebroid (cf.~also \cite{KS08})\footnote{In fact, in \cite{KS08} we called it Loday algebroid. We finally decided for the present paper to call an algebra with the property \eqref{cJacobi} a Leibniz-Loday algebra and, if it is an algebra of sections and the condition \eqref{cLeibniz} is satisfied, a Leibniz-Loday algebroid.}. 
The second line brings in the fiber metric so as
to control the violation of the antisymmetry of the bracket, while the last line is the
ad-invariance of the inner product. 
Eqs.~\eqref{cJacobi} and
\eqref{cLeibniz} permit to conclude that $\rho$ is a morphism of
brackets, cf.~e.g.~the corresponding generalization in
Prop.~\ref{p:Vinogrelem} below; together with \eqref{cLeibniz} this
gives the five axioms often demanded in the literature for the
definition of a Courant algebroid. 

For a general Courant algebroid one always has the complex
\beq 0 \to T^*M \to W \to TM \to 0 \, , \label{seq}\eeq
induced by $\rho$ and its adjoint. Exact Courant algebroids, which, by
definition, means that the above sequence is exact, are classified by
an element $[H] \in H^3_{{\mathrm{dR}}}(M)$ \cite{SevLett}; up to a
contribution $H(v,w,\cdot)$  added to the r.h.s.~of
\eqref{Cour}, this reproduces the explicit formulas on $W=TM\oplus
T^*M$ introduced above (with $\rho$ being projection to the first
factor and a change of splitting in \eqref{seq} corresponding to a
change $H \mapsto H + \ud B$).

The construction in \eqref{Anton} can be translated easily into
super-language. As before we may use $\CM=T[1]M$ and consider
$\imath_v$ as a vertical vector field and $\alpha \wedge$ as
multiplication with a function of degree one. If we want to have
$\alpha$'s enter as vector fields as well, getting a model for what
happens in the previous section, we can simply extend $T[1]M$ by one
(graded) copy of $\Real$. Denoting this coordinate by $b$, and
choosing local $x^i$ on $M$ with its induced odd coordinates $\xi^i =
\ud x^i$ on $T[1]M$, we then associate to every $\psi = v \oplus \alpha$
the vector field
\be \psi = v^i \frac{\partial}{\partial \xi^i} + \alpha_i \xi^i  \frac{\partial}{\partial b} \, ; \ee
declaring $b$ to have degree two, this becomes a homogeneous vector
field on the graded manifold $\CM = T[1]M \times \Real [2]$ of degree
minus one (and, moreover, it has the form of the most general vector
field of this degree). The bracket \eqref{Cour} now follows as a
derived bracket of (graded) vector fields on $\CM$, $[ v \oplus \alpha , w
\oplus \b]_W = [\psi,\psi']_Q$ where the latter is defined by 
\be [ \psi, \psi']_Q := [[ \psi , Q ], \psi'] , \label{Cour2} \ee
with $Q$ being the
canonical vector field on $T[1]M$ corresponding to the de Rham
differential on $M$ extended trivially to the product and the bracket
denoting the graded Lie bracket of vector fields.
Obviously, this
can be ``twisted'' by a closed 3-form $H$ by replacing the above $Q$
by \be Q = \xi^i  \frac{\partial}{\partial x^i} + \frac{1}{3!} H_{ijk}
\xi^i \xi^j \xi^k  \frac{\partial}{\partial b} \label{Qtwist1} \ee
which again yields a Q-structure on $\CM$. This reproduces the general
situation of an exact Courant algebroid mentioned above. Note also
that in this language the inner product on $W=TM \oplus T^*M$ results
from the (graded) bracket of the corresponding vector fields on
$\CM$: 
\be [\psi,\psi'] = \<\psi,\psi'\>   \frac{\partial}{\partial b} \, .
\label{inner2}\ee

To have vector fields $v$ and 1-forms $\alpha$ enter more symmetrically,
one may alternatively lift the graded construction on $\CM = T[1]M$
canonically to the (super)symplectic manifold $T^*[2]\CM$: The local
coordinates $x^i$ and $\xi^i$ of degree 0 and 1 on $\CM$ are then
accompanied by canonically conjugate momenta $p_i$ and $\xi_i$ of
degree 2 and 1, respectively. $\psi= v + \alpha = v^i \xi_i + \alpha_i \xi^i$
then describes the most general function of (total) degree
one\footnote{A vector field on $\CM$ is always a fiber linear function
on its cotangent bundle and a function can be pulled back canonically
by the projection. The shift in grading by two has been chosen so as
to have them both of the same degree.}, and the canonical Hamiltonian $\tilde Q$ for
lifting the vector field $Q=\xi^i \partial/\partial x^i$ to the
cotangent bundle, $\tilde Q = \xi^i p_i$, then reproduces
\eqref{Anton} and the inner product by means of the (graded)
Poisson brackets
\begin{align}
\{ \, \{ \psi , \tilde Q \} , \psi' \} &= [ \psi ,  \psi' ]_W 
\label{dBrack1}\\
\{  \psi , \psi' \} &= \< \psi , \psi' \> \, .
\label{dInner}\end{align}
In fact, such a construction can be extended to the general setting of
a Courant algebroid as defined above. One then finds that Courant
algebroids are in {\em one-to-one} correspondence to {\it symplectic}
Q2-manifolds \cite{Royt02} (where, by definition, the symplectic form
is required to be of degree two and to also be compatible with the
Q-structure, i.e.\ to be preserved by the vector field $Q$).

One observes the similarity of the formulas \eqref{Cour2},
\eqref{inner2} with \eqref{dBrack1}, \eqref{dInner}. In both
cases the original bracket is a (graded) Lie bracket, in the second
case of degree minus two, in the first case of degree
zero.\footnote{\label{f:24} In fact, the second bracket is even a
(graded) Poisson bracket, i.e.~there exists a second compatible graded
commutative multiplication $\cdot$, usually taken to be of degree zero. (Here
compatibility means a graded Leibniz condition: $\{F,G \cdot H \} = \{F,G
\} \cdot H + (-1)^{(|F|+d)|G|} G \cdot \{F, H \}$, where $d$ is the degree of the bracket, cf.~also the subsequent footnote).  If one
prefers, also the bracket of vector fields can be embedded into a
Poisson algebra, by viewing vector fields as fiber linear functions on
the cotangent bundle and thus replacing $\CM$ by the graded manifold
$T^*\CM$, if one prefers, shifted also by some degree (by degree one
if the canonical Poisson bracket is to reproduce the
Schouten--Nijenhuis  bracket of (graded)
vector fields on $\CM$). We will come back to this perspective below.}
In the previous section, moreover, we viewed the left-adjoint action
of $Q$ as a (compatible) differential $\ud_Q$. Thus, to cover all
relevant cases simultaneously, one is led to regard the derived
bracket construction of a differential graded Lie algebra (DGLA),
i.e.\ a graded Lie algebra with a differential (nilpotent odd
operator) compatible with the bracket:\footnote{We call
$(A^\bullet,[.,.])$ a (graded) \emph{Leibniz--Loday algebra} of degree $d$, if
$[A^k,A^l]\subset A^{k+l+d}$ and $[a,[b,c]]=[[a,b],c]
+(-1)^{(|a|+d)(|b|+d)}[b,[a,c]]$. It is a \emph{differential graded
Leibniz--Loday algebra} (DGLoA) if in addition there is an operation $D \colon
A^\bullet \to A^{\bullet+1}$ squaring to zero and satisfying
$D([a,b])=[Da,b] + (-1)^{(|a|+d)} [a,Db]$. If the bracket $[.,.]$ is
also graded antisymmetric, i.e.~$[b,a]=-(-1)^{(|a|+d)(|b|+d)}[a,b]$, we
can replace ``Leibniz--Loday'' by ``Lie'' in these definitions.}

\begin{prop}[Kosmann-Schwarzbach \cite{YKS96}]\label{p:YKSder} Let
$(A^\bullet,[.,.],\uD)$ be a DGLA (with bracket
of degree $d$) and define its derived bracket as:
\be\label{Dder} [\phi,\psi]_D = (-1)^{|\phi|+d+1}[\uD\phi,\psi] \;,
\ee
where $|\phi|$ denotes the degree of $\phi$ (assumed to be
homogeneous).  

\noindent Then the following facts are true: 
\begin{enumerate}
\item This new bracket has degree $d+1$ and is a graded
Leibniz-Loday bracket,
\be\label{Loday} [\phi,[\psi_1,\psi_2]_D]_D = [[\phi,\psi_1]_D,\psi_2]_D
  +(-1)^{(|\phi|+d+1)(|\psi_1|+d+1)}[\psi_1,[\phi,\psi_2]_D]_D \;.
\ee
\item If $|\phi|$, $|\psi_!|$, and $|\psi_2|$ all have the opposite parity of $d$, one has 
\begin{eqnarray}\label{dSkew} [\phi,\phi]_D &=&\frac12 \uD[\phi,\phi] \;, \\
             {}[\phi,[\psi_1,\psi_2]]_D &=& [[\phi,\psi_1]_D,\psi_2]+ 
[[\psi_1,[\phi,\psi_2]_D]\label{Dcomp} \; .  \label{adInvar}
\end{eqnarray}
\item The quotient of  $(A^\bullet,[.,.]_D)$ by $\ker D$ is a graded Lie
algebra. 
\end{enumerate} 
\end{prop}
These facts are easy to verify and recommended to the reader as an
exercise. (To see that $A^\bullet/\ker D$ is well-defined one observes
that $[\phi,\psi]_D$ equals $[\phi,D\psi]$ modulo a $D$-exact term,
which can also be used to establish the graded antisymmetry of the
induced bracket). 

In fact, the Leibniz-Loday Property~\eqref{Loday} also follows in the more
general context of $(A^\bullet,[.,.],\uD)$ being a DGLoA only, cf.~the
preceding footnote and Lemma \ref{p:der2} below. The choice of the
sign in \eqref{Dder} is essential and can be motivated also from
examples where $\uD$ is an inner derivation, $\uD=[Q,.]$ with
$[Q,Q]=0$, and the brackets arise as in \eqref{Cour2}. This
definition also agrees with (the first part of) Eq.~\eqref{induced}
with the identification $\uD=\ud_Q$, since there the bracket has
degree zero and the elements are odd.

\bigskip

We now turn to the alternative description of Q-manifolds when
studying the vector fields on them and their derived bracket. To
explain this procedure, we start with a Q1-manifold, $\CM = E[1]$:
Sections $\psi$ in $E$ can be identified with degree minus one vector
fields on $\CM$, $\psi = \psi^a
\partial/\partial \xi^a$ in our notation.  Replacing $\ud$ in Cartan's
formula \eqref{derbasic} with the differential of the Lie algebroid (acting
 on $\EO^\bullet(M)$) every Lie algebroid bracket occurs naturally as a
derived bracket. In super-language this means explicitly that the
Lie algebroid bracket on $E$, which for clarity we now denote by $[
\cdot , \cdot ]_E$ is a derived bracket w.r.t.~$Q$,
\be [\phi,\psi]_E = [[\phi,Q],\psi] \, , \label{Eder}
\ee
where, for notational simplicity, we did not introduce a symbol for
the map identifying sections in $E$ (as they appear on the l.h.s.~of
this equation) with the respective degree minus one vector field (as
understood on the r.h.s.)---in \eqref{derbasic} the symbol
$\imath_\cdot$ could be interpreted as such an isomorphism. We will
keep such a notational convention also further on. From the previous
general considerations about DGLAs and their derived brackets, we can
now at once conclude about two of the three defining properties for
$[\cdot , \cdot ]_E$ to be a Lie algebroid bracket: Since there are
\emph{no} degree minus two vector fields on $\CM = E[1]$, the right
hand side of \eqref{dSkew} vanishes identically. With the fact that
$\phi$, $\psi$ have odd degree and the natural bracket of super vector
fields degree $d=0$, Eqs.~\eqref{dSkew} and \eqref{Loday} yield
$[\cdot , \cdot ]_E$ to be (ungraded) antisymmetric and to satisfy an
(ungraded) Jacobi identity, respectively. To retrieve the Leibniz 
property w.r.t.~an or the anchor map of $E$,
we also need to make use of the module structure of vector fields
w.r.t.~functions on $\CM$, where we may restrict to degree zero
functions $C^\infty_0(\CM) \cong C^\infty(M)$ so as to preserve the
degree of our vector fields. This will be provided by
Eq.~\eqref{eq:Leibniz} below, then establishing that we obtain a Lie
algebroid structure on $E$ by this procedure. 

Since we will need such type of anchor maps also in the general
context of higher degree $Q$-manifolds and their compatibility with
the derived bracket construction, we will extend the
considerations on DGL(o)As to this setting. A rather minimalistic way of
doing so is 
\begin{vdef}
We call $(A^\bullet,[.,.],\uD, R, \rho_0)$ a  \emph{weakly anchored} DGL(o)A
over the ring $R$ (commutative, with unity) if $(A^\bullet,[.,.],\uD)$
is a DGL(o)A of degree $d$ and $\cdot \colon R \times A^k
\to A^k$ and $\rho_0 \colon A^{-d} \to \End_{\Real} (R)$
are maps such that 
\be [\phi, f \cdot \psi] = \left(\rho_0(\phi)\right)(f) \cdot \psi + f
\cdot [\phi, \psi] \label{ring}
\ee
holds true for every $\phi \in  A^{-d}$, $\psi \in A$, and
$f\in R$.

We call it merely \emph{anchored} iff in addition $\rho_0$ is $R$-linear and
so is its composition with $\uD$, i.e.~for all $\phi \in  A^{-d-1}$
one has
\[ \label{rhoBund} \rho_0(\uD(f\phi))=f\rho_0(\uD\phi) \,. \]
We call it \emph{strongly anchored} if also ($\forall f,g \in R, \psi \in
A$) 
\[  \uD(fg\psi) +fg\uD \psi= f\uD(g\psi) +g\uD(f\psi) \,. \label{1stder} \]
\end{vdef} 

In fact, \eqref{1stder} is \emph{equivalent} to the
existence of an operator  $\CD \colon R\to \Hom_R(A^\bullet,A^{\bullet+1})$ such that
\begin{align}
  \uD(f\phi) &= (\CD f)(\phi) +f\uD \phi  \label{DLeib} \,, \\
  \CD(fg) &= f\CD g +g\CD f \label{D0Leib}
\end{align}  holds true. \eqref{D0Leib} implies that  $\CD$ is a first 
order differential operator and \eqref{DLeib} that likewise is $\uD$
w.r.t.~$R$. In fact Eq.~\eqref{DLeib} also shows constructively how to
define $\CD$ once \eqref{1stder} is fulfilled. We have chosen the
formulation as in \eqref{1stder} since it is intrinsic and does not
need the introduction of new operations such as the above map $\CD$. 

A graded Lie algebroid over an ungraded manifold $M$ gives an example
of a strongly anchored DGLA with $\uD=0$; here $A^\bullet$ are the
sections in the graded vector bundle and $R=C^\infty(M)$, as an
$R$-module $A^\bullet$ is projective in this example. Obviously the
above definition can be easily extended to graded commutative rings
$R$ in which case general graded Lie algebroids can be covered as
well; this is however not needed for the applications we are having in
mind here (cf.~Proposition \ref{p:tangentLie} below). An example with
nontrivial $\uD$ can be constructed from Lie bialgebroids satisfying
Eq.~\eqref{rhoBund}; here $A^\bullet$ consists of the sections of the
exterior powers of the (ungraded) vector bundle. But more important
for our context is the following observation:\footnote{Cf.~also
footnote
\ref{f:8} concerning our simplified nomenclature and footnote \ref{f:24} for
a possible relation of the two parts of the proposition. A differential
graded Poisson algebra $(A^\bullet, \{ . , . \}, \cdot , \uD)$ of degree $d$ is
a DGLA $(A^\bullet, \{ . , . \}, \uD)$ and a graded Poisson algebra
$(A^\bullet, \{ . , . \}, \cdot)$, both of degree $d$, such that $\uD(F \cdot
G) = \uD F \cdot G + (-1)^{|F|} F \cdot \uD G$ for all $F
\in A^\bullet$, $G \in A$.}
\begin{prop}\label{p:tangentLie}
$\,$
\begin{enumerate}
\item $(A^\bullet,[.,.])=\X(\CM)$ with $\uD=\ud_Q$
of a Q-manifold is a strongly anchored DGLA with bracket of degree zero 
w.r.t.\ $R=C^\infty_0(\CM)\cong\smooth(M)$, where $\rho_0(\phi)[f]=\phi(f)$.
\item 
A differential graded Poisson algebra (DGPA) on $C^\infty(\CM)$ is a weakly
anchored DGLA w.r.t.\ $R=C^\infty_0(\CM)$ and $\rho_0(\phi)f=\{\phi,f\}$. It
is strongly anchored iff $\{f,g\}=0=\{\phi \cdot \uD f,g\}$ for all $f,g\in
R$ and $\phi \in C^\infty_{-d-1}(\CM)$.
\end{enumerate}
\end{prop}
\begin{proof} For the case of vector fields use $\D \phi = [Q,\phi]$ and find
\eqref{DLeib} with 
$(Df)(\phi)$ being the multiplication of $\phi$ by $Q(f)$. $Q$ being a vector
field, this proves also \eqref{D0Leib} and thus \eqref{1stder}. From its
definition it is immediate that $\rho_0$ is $R$-linear in $\phi
\in \X_0(\CM)$. Moreover using \eqref{DLeib} and observing that $\rho_0
\left(Q(f) \phi\right)=0$, which follows since here $\phi$ is degree minus
one and vanishes when applied to degree zero functions, also \eqref{rhoBund}
follows.

In the Poisson case the two additional conditions are found by directly
checking the desired conditions using the definition of
$\rho_0$. E.g.~$\{f,g\}=0$ follows immediately from $R$-linearity of
$\rho_0$. 
\end{proof}

\pagebreak[2]These additional structures are compatible with the
derived bracket construction in the following sense:\nopagebreak[3]
\begin{lemma}\label{p:der2}
  Let $(A^\bullet,[.,.],\uD,R,\rho_0)$ be a (weakly/strongly) anchored DGLoA, 
  then $(A^\bullet,[.,.]_\uD,\uD,R,\rho)$ is also a (weakly/strongly)
  anchored DGLoA with 
\be \rho(\phi) = \rho_0(\uD\phi)  \;. \label{eq:rho1}
\ee
\end{lemma}
This can be proven by simple straightforward calculations. The Lemma contains
three different statements. 

We now apply these considerations first by returning to Q1-manifolds. As a
consequence of Proposition \ref{p:tangentLie} and the preceding lemma,
we notice that the derived bracket \eqref{Eder} satisfies a Leibniz
rule \be [\phi, f \cdot \psi]_E =\rho(\phi)(f) \cdot \psi + f
\cdot [\phi, \psi]_E \label{eq:Leibniz} 
\ee with an $R$- or $C^\infty(M)$-linear anchor map 
$\rho$ given by $\rho(\phi)f = [Q,\phi] f$. $C^\infty(M)$-linearity
implies  that $\rho$ indeed comes from a vector bundle map $E\to TM$.


Thus, indeed, the bracket on $\Gamma(E)$ defined via the derived
bracket construction
\eqref{Eder} induces a Lie algebroid structure on $E$. Moreover, this is the 
\emph{same} Lie algebroid structure as found before in Section~\ref{s:examQ}, 
which one easily verifies on a basis: Using the above formulas as definitions 
for the bracket and the anchor, one verifies easily that $[\partial_a,
\partial_b]_E = C^c_{ab} \partial_c$ and $\rho(\partial_a)=\rho^i_a 
\partial_i$ (with $\partial_a \equiv \frac{\partial}{\partial \xi^a}$ and 
$\partial_i \equiv \frac{\partial}{\partial x^i}$ for the parameterization used 
in \eqref{Q}).

We now want to imitate this second procedure, illustrated at the
example of Q1-manifolds and in spirit close to the considerations of
the previous section on gauge symmetries, so as to find an alternative
characterization of degree two Q-manifolds. In fact, we want to
perform some of the considerations common to all choices of degree $p$
first and later specialize the discussion to $p=2$. Courant algebroids
should be particular examples of the latter case, moreover, and we
will find a description of Q2-manifolds or Lie-2-algebroids
complementary to the one given in Section~\ref{s:examQ} and much closer to
the usual formulation of Courant algebroids.

\subsection{Vinogradov and $V$-twisted Courant algebroids}
\label{subsec:Vino}
To obtain a prototype of a bracket for higher $p$ than two, we
may look at a generalization of the Courant-Dorfman bracket, the
so-called Vinogradov
bracket. It is in fact the bracket \eqref{Cour} with $\alpha,\b
\in \Omega^k(M)$ for some arbitrary $k \in \N$, and it can be 
obtained precisely as for the case $k=1$. Accordingly, we again have
at least two options translating this scenario into super-language. In
the $\CM=T[1]M \times \Real[p]$ picture\footnote{We have found this description in 2005. In the mean time it appeared also in \cite{Uri13}}, sections $v\oplus\alpha$ correspond
to \be  \label{psi}
\psi = \imath_v + \alpha \frac{\partial}{\partial b} \, , \ee 
where $b$ obviously needs to have degree $p=k+1$ for $\psi$ to be
homogeneous of degree minus one. The Vinogradov bracket then follows
again as a derived bracket \eqref{Cour2} with $Q$ being the de Rham
differential. In generalization of Eq.~\eqref{Qtwist1} this may
certainly again be twisted, \be Q = \ud + H \frac{\partial}{\partial
b} \, , \label{Qtwist2} \ee where now $H$ needs to be a closed
($k+2$)-form obviously. Like in the Courant algebroid, also here the
original (graded) Lie bracket plays an important role,
cf.~Eq.~\eqref{inner2} (or \eqref{dInner}), generating a
$C^\infty(M)$ bilinear pairing; it comes from the contraction of the
vector fields with the $k$-forms, as one easily verifies using
\eqref{psi}. Thus, this time the pairing maps two elements in
$W=TM \oplus \Lambda^k T^*M$ to an element of another bundle, namely
of $V=\Lambda^{k-1} T^*M$. This, the appearance of a second bundle $V$
of relevance, will be one feature we want to keep also for the more
general setting in what follows.

Before axiomatizing these data, in analogy to the step from
\eqref{Cour} to Definition \ref{def:Cour}, into what one may want
to call a Vinogradov algebroid, we briefly translate the above picture
also into a graded Poisson language. We therefore regard 
$\CM=T^*[p]T[1]M$ with its canonical Poisson bracket $\{ \cdot , \cdot
\}$ of degree $-p$. Local coordinates $x^i$ and $\xi^i$ of degree 0
and 1, respectively, are then accompanied by momenta $p_i$ and $\xi_i$
of degree $p$ and $p-1$, respectively. Sections of $W$ correspond to
functions on this $\CM$ of degree $k=p-1$, $\psi= v^i \xi_i +
\frac{1}{k!}\alpha_{i_1
\ldots i_k} \xi^{i_1} \ldots \xi^{i_k}$, and the bracket and pairing 
on $W$ results from formulas \eqref{dBrack1} and \eqref{dInner},
respectively, where $\tilde{Q}=\xi^i p_i +H$.

We now return to the task of axiomatization. In particular, one may
want to check properties of the bracket and pairing, finding
appropriate generalizations of those written in definition
\ref{def:Cour}. There are natural candidates for the analogues of the
objects appearing in \eqref{cLeibniz} and \eqref{cnSkew}, which will
also turn out to be realized for what one may call the Standard
Vinogradov algebroid ($W=TM \oplus \Lambda^k T^*M$, $V=\Lambda^{k-1}
T^*M$), possibly twisted by some $H\in \Omega_{cl}^{k+2}(M)$. The
anchor map $\rho$ will be projection to $TM$, for $\uD$ one can choose
the de Rham differential (succeeded by the natural embedding of
$k$-forms into sections of $W$). For the ad-invariance of the pairing
(cf.~Eq.~\eqref{cInvar}), however, we can no more use solely the
anchor map, since the pairing lands in sections of $V$. The operator
on the l.h.s.~of
\eqref{cInvar} will then be found to be replaced by the Lie derivative
w.r.t.~the anchor of the section $\psi$. Note that this is for $k>1$
(or, likewise, $p>2$) no more $C^\infty(M)$-linear in $\psi$. So, in
the axioms there will be besides $\rho$ and $\uD$ one more new object,
which we will call $\rho_D$, not necessarily coming from a bundle map,
but being restricted appropriately otherwise. Also, the relation
between $\rho$ and $\uD$ being adjoint to one another does not make
any sense anymore.\footnote{Only the generalization that $\rho_D$
would be adjoint of $\D$ w.r.t.~the pairing can be meaningful. In
fact, although it is not satisfied in the Standard Vinogradov
algebroid for $k>1$, posing this additional condition will be studied
in detail the context of $p=2$ below, cf., e.g., definition
\ref{def:VCourant} and Proposition \ref{p:FW}.}  $\uD$ needs to be
restricted then by something like Equation~\eqref{1stder} so as to
ensure that it is a first order differential operator. Likewise for
$\rho_D$ we require it to take values in the module $\Gamma(\CDO)(V)\equiv \Gamma(\CDO(V))$ of covariant differential operators  on $V$: Keeping in
mind that the $\Real$-linear first order differential operators on a
vector bundle $V\to M$ are themselves a projective module over $M$, we
call the generating bundle $\CDO(V)$, the bundle of covariant
differential operators on $V$; elements in  $\Gamma(\CDO)(V)$ are then sections of this bundle. There is a canonical map from $\CDO(V)$
to $TM$ which we may use to require $\rho_D$ to cover the
map from $\Gamma(W)$ to $\Gamma(TM)$ induced by $\rho$ (and
conventionally denoted by the same letter), i.e.~in formulas we require 
\be 
  \rho_D(\phi)[fv] = \rho(\phi)[f]v +f\rho_D(\phi)[v] \, .
\label{vLeibn3} 
\ee where $\phi \in \Gamma(W)$, $f \in C^\infty(M)$, and $v \in \Gamma(V)$.
We thus are led to the following

\newpage
\begin{vdef} \label{def:Vino} A \emph{Vinogradov algebroid} are two vector bundles $W$ and $V$
together with a bracket $[.,.]_W:\Gamma(W)\times\Gamma(W)\to \Gamma(W)$,
a map $\rho\:W\to TM$, a non-degenerate\footnote{A (skew)-symmetric bilinear form $\<.,.\>$ is non-degenerate if for every $\phi\in E_x$, $\phi\ne0$ there is a $\psi\in E_x$ such that $\<\phi,\psi\>\ne0$.} surjective inner product
$\<.,.\>\: W\otimes_M W\to V$, an $\Real$-linear map
$\uD:\Gamma(V)\to \Gamma(W)$, 
and a map $\rho_D: \Gamma(W)\to \Gamma(\CDO)(V)$ covering the map
$\rho \colon  \Gamma(W)\to \Gamma(TM)$ subject to the
following axioms.
\begin{align}
  [\psi,[\phi_1,\phi_2]_W]_W &= [[\psi,\phi_1]_W,\phi_2]_W +
[\phi_1,[\psi,\phi_2]_W]_W
    \label{vJacobi}  \\
  [\psi,\psi]_W &= \frac12\uD\<\psi,\psi\>  \label{vSkew}\\
  \rho_D(\psi)\<\phi,\phi\> &=
  2\<[\psi,\phi]_W,\phi\>  \label{vInvar}\\
  \uD(fgv) +fg\uD(v) &= f\uD(gv) +g\uD(fv)  \label{vLeibn2}
\end{align} where $\phi_i,\psi \in \Gamma(W)$, $f,g \in \smooth(M)$,
and $v \in \Gamma(V)$.
\end{vdef}
It is easy to see that \eqref{vInvar} and \eqref{vLeibn3} imply 
\be
[\psi,f\phi]_W = \rho(\psi)f\phi +f[\psi,\phi]_W \label{vLeibniz} \;
\ee and, conversely, 
one may also conclude \eqref{vLeibn3} from \eqref{vInvar}
and \eqref{vLeibniz}.  One may consider adding $[\uD v,\phi]_W =
0\label{brwExact}$ to the axioms as this rule is suggested from
several perspectives (e.g.~it is satisfied for the Vinogradov
bracket). For the case of 
 $\rho_D(\phi)[v]=\<\phi,\D v\>$, on the other hand,
this follows automatically, cf.~Proposition \ref{p:Royt} below.

To require surjectivity of the pairing goes almost without loss of
generality since $V$ enters the axioms only via the image of  $\<
\cdot , \cdot \>$; only if this image changes its rank, one cannot
consistently restrict to it as a vector bundle. On the other side, it
is suggested from some perspectives to drop the condition of
non-degeneracy, 
which, however, then would go with a partially
different formulation of the axioms and we
will not pursue this idea in the present article. Let us mention that
in the application to Q-manifolds the inner product is automatically
surjective and non-degenerate, cf.~Proposition \ref{p:genStdVinogr}
below.

Note that the Equation~\eqref{vInvar} is nothing but the invariance of the
inner product under sections of $W$, being equivalent to
\begin{align*}
  \rho_D(\psi)\<\phi_1,\phi_2\> &=
    \<[\psi,\phi_1]_W,\phi_2\> + \<\phi_1,[\psi,\phi_2]_W\> \,
\end{align*}
by a standard polarization argument. It is an instructive exercise to
compute the analog of the Leibniz rule (Eq.~\eqref{vLeibniz}) for
the left hand side of the bracket, which shows that the bracket
is also a first order linear partial differential operator in the l.h.s.
Also one has the following 
\begin{prop}\label{p:Vinogrelem}  The maps $\rho$ and $\rho_D$ are morphisms of
brackets.\footnote{In the case of Courant algebroids it was apparently
first K. Uchino \cite{Uchi02} who observed that the morphism property
of $\rho$ can be deduced from the other axioms.} Furthermore, their
composition with $\D$ vanishes ($\rho\circ\D =0=\rho_D\circ\D$).
\end{prop}

\begin{proof}
The morphism property of $\rho$ can be concluded easily from
\eqref{vJacobi} and \eqref{vLeibniz}, cf., e.g., \cite{HS08}. 
For the morphism property of $\rho_D$ write a section of
$V$ locally as scalar product of two sections of $W$.  Then use \eqref{vInvar}
and \eqref{vJacobi} to proceed in a similar manner.
The last two properties follow from \eqref{vSkew} and the first or second property, respectively.
\end{proof}

Let us now check that the axioms are indeed all fulfilled for the
(twisted standard) Vinogradov bracket described above.  Since we
already introduced some super language for the bracket, let us start
from there: 
\begin{prop}\label{p:derVino}  Given a strongly anchored DGLA structure
  $(A^\bullet,[.,.],\uD,\rho_0)$ with bracket of degree $d$ over
  $R=\smooth(M)$ of a smooth manifold $M$, where $A^{-d-1}$ and
  $A^{-d-2}$ are projective modules over $M$, then there exist vector
  bundles $W\to M$ and $V\to M$ such that $A^{-d-1}$ is (canonically
  isomorphic to) the space of sections of $W$ denoted
  $\phi,\psi,\ldots$ in what follows and $A^{-d-2}$ is the space of
  sections of $V$ a typical element of which is denoted by $v$ below.
  Thus $\D \colon \Gamma(V)\to \Gamma(W)$. Now the following operations
  can be defined by their right hand sides:
\begin{align}
  [.,.]_W &: \Gamma(W)\otimes\Gamma(W)\to \Gamma(W),&  [\phi,\psi]_W&:=[\uD\phi,\psi] \label{dBrack2}\\
  \rho &: W\to TM,&
 \rho(\phi)[f]&:= \rho_0(\uD\phi)[f]  \label{dAnchor}\\
  \<.,.\> &: \Gamma(W)\otimes\Gamma(W)\to \Gamma(V),&
    \<\phi,\psi\>&:=[\phi,\psi] \label{dInner2}\\
  \rho_D &: \Gamma(W)\to \Gamma(\CDO)(V), \;& \rho_D(\phi)[v]&:= [\uD\phi,v] \,. \label{rhoQ}
\end{align}
  $\<.,.\>$ is $\smooth(M)$-linear iff $[\psi,f\phi]=f[\psi,\phi]$ for
  all functions $f$ on $M$. Requiring this and that the induced inner product
  $\<.,.\>\:W\otimes_M W\to V$ is non-degenerate and surjective, these
  operations form a Vinogradov algebroid.
\end{prop}
\begin{proof}  Follows from Proposition \ref{p:YKSder} 
  and Lemma \ref{p:der2} with the interpretation that \eqref{vLeibn3} is just
  the Leibniz rule for the derived bracket restricted to $A^{-d-2}$ in the
  right-hand argument.
\end{proof}

Due to Proposition \ref{p:tangentLie} the vector fields on our super manifold
$\CM=T[1]M\times\Real[k+1]$ together with the above $Q$-structure fulfill the
first condition.  $\smooth(M)$-linearity of the inner product is
automatically fulfilled for Q-manifolds (since $d=0$ and thus $\Gamma(W)$
is identified with degree minus one vector fields in this case). 
Non-degeneracy follows by direct inspection, noting that  
$W \cong TM\oplus \Lambda^kT^*M$ for
$1\le k\le\dim M$. This shows that the (twisted) Vinogradov
bracket and the canonical pairing indeed provide an example of a
Vinogradov algebroid; we will call this the ($H$-twisted) standard 
Vinogradov algebroid.

An immediate generalization of this is to replace the tangent bundle
in the above construction by a Lie algebroid $A$.  Starting with a Lie
algebroid $(A,[.,.]_A,\rho_A)$ and a positive integer $k\le\rk A$ we
consider the bundle $W:=A\oplus\Lambda^kA^*$ with the canonical
symmetric (non-degenerate) pairing $\<X\oplus\alpha,Y\oplus\b\>
:=\imath_{X}\alpha+\imath_{Y}\b$.  The choice for $V$ is therefore
$\Lambda^{k-1}A^*$.  The choice of the super manifold $\CM$ is analog
to the tangent case $\CM=T^*[k+1]A[1]$.  $A[1]$ also has a canonical
$Q$-structure, the differential $\ud_A$ of the algebroid $A$.  Its
Hamiltonian lift to $\CM$ is the $Q$-structure we use.  Since the
Hamiltonian is itself of degree $k+2$, we can add an $A$-$(k+2)$-form
$H$, satisfying $\ud_AH=0$ so as to render the total charge further on
nilpotent.  This leads us to
\begin{prop}\label{p:genStdVinogr}  Let $(A,[.,.]_A,\rho_A)$ be a Lie algebroid over
$M$. Let $W=A\oplus\Lambda^kA^*$ and $V= \Lambda^{k-1}A^*$. Then  
\begin{align}
  (X\oplus\alpha,Y\oplus\beta) &:= \imath_X\beta+\imath_Y\alpha \\
  [X\oplus\alpha,Y\oplus\beta]_W &:= [X,Y]_A\oplus \L_X\beta
    -\imath_Y\ud_A\alpha +\imath_X\imath_Y H  \label{stdVbrack}\\
  \rho(X\oplus\alpha) &:= \rho_A(X)  \\
  \D(v) &:= 0\oplus\ud_Av  \label{stdVder}\\
  \rho_D(X\oplus\alpha) &:= \L_X  \, ,  \label{stdVrhoQ}
\end{align}  where $X,Y \in \Gamma(A)$, $\alpha,\beta \in
\Gamma(\Lambda^kA^*)$,  and $v \in \Gamma(V)$, equips $(W,V)$ with the
structure of a Vinogradov algebroid
.
\end{prop}
\begin{proof}  Follows from Proposition \ref{p:derVino} with the
realization as a Q-structure. 
\end{proof}
We will call this example an ($H$-twisted) generalized standard
Vinogradov algebroid, reducing to an ($H$-twisted) standard Vinogradov
algebroid for the choice of $A$ being a standard Lie algebroid, $A
=TM$. Note that for every generalized standard Vinogradov algebroid we
have two quite distinct super-geometric descriptions, one as described
above where $\CM = T^*[k+1]A[1]$ and another one on the supermanifold
$\CM' = A[1] \times
\Real[k+1]$  with $Q'=\ud_A + H \partial/\partial b$, $b$ being the
canonical coordinate on $\Real[k+1]$.\footnote{Also these two pictures
are not just related by the Hamiltonian reformulation of $(\CM',Q')$
according to footnote \ref{f:24}, $T^*[\bullet]\CM'$ does not agree
with $\CM$.} 

An important class of Vinogradov algebroids is provided by those where
the map $\rho_D$ is $C^\infty(M)$-linear (note that for a generalized
standard Vinogradov algebroid this is the case only for $\rk
V=1$). A special case of this is provided by the following definition
(cf.~Prop.~\ref{p:FW} below):
\begin{vdef}[{$V$}-twisted Courant algebroid]\label{def:VCourant}
Let $V,W$ be vector bundles over $M$, $W$ equipped with an anchor map $\rho
\colon W\to TM$, a bracket $[ . , . ]_W$ on its sections, and a non-degenerate
surjective symmetric product $\<.,.\>$ taking values in $V$. Let
 $\Wna$ be a $W$-connection on $V$. We call this a vector bundle twisted (or
V-twisted) Courant algebroid if these data are subject to the
following axioms
\begin{align}
  [\psi,[\phi_1,\phi_2]_W]_W &= [[\psi,\phi_1]_W,\phi_2]_W +[\phi_1,[\psi,\phi_2]_W]_W
    \label{fwJacobi}  \\
  \<[\psi,\phi]_W,\phi\> &=\frac12\Wconn_\psi \<\phi,\phi\> =  \<\psi,[\phi,\phi]_W\>
  \label{fwInvar} \;.
\end{align}
\end{vdef}

We remark that the usual Leibniz rule \eqref{vLeibniz} follows from \eqref{fwInvar} after polarization, $\<[\psi,\phi_1],\phi_2\>+\<\phi_1,[\psi,\phi_2]\> = \Wconn_\psi\<\phi_1,\phi_2\> $, from the Leibniz rule for $\Wna$.

\begin{prop} \label{p:FW} A V-twisted Courant algebroid
  is a Vinogradov algebroid with   $\rho_D(\phi):=\Wconn_\phi$ and $\D$
its adjoint, i.e.~$\D$ being defined via
\begin{align} 
  \label{FW}  \<\phi,\D v\> =\Wconn_\phi v \;.
\end{align}
Conversely, a Vinogradov algebroid is a V-twisted Courant algebroid if
\begin{align} 
  \label{FW'}  \<\phi,\D v\> =\rho_D(\phi) v \;.
\end{align}
\end{prop}
\begin{proof} Given \eqref{FW} then \eqref{fwInvar} is equivalent to 
  \eqref{vSkew}--\eqref{vInvar}.
  It remains to check that $\D$ fulfills the Leibniz rule \eqref{vLeibn2},
  but this follows since \eqref{vLeibn3} shows that the adjoint
  $\Wna$ of $\D$ is a first order linear PDO. For the other direction note that \eqref{FW'} replaces \eqref{FW} thus \eqref{vSkew}--\eqref{vInvar} are fulfilled.  Further $\Wna=\rho_D$ the adjoint of $\D$ implies that it is $\smooth$-linear in $\phi$.
\end{proof}

Note that the l.h.s.~ of \eqref{FW} also permits the interpretation that
$\D\:\Gamma(V)\to{}^W{}\Omega^1(M,V)$ and this 1-form is then applied to a vector
$\phi\in W$, i.e.\ $\D$ is an exterior $W$-covariant derivative. From
Proposition \ref{p:Vinogrelem} we learn that $\Wna$ is a
morphism of brackets,
i.e.~$$\Wconn_{[\phi,\psi]_W}=[\Wconn_\phi,\Wconn_\psi] \, .$$  In other words,
the $W$-connection of a V-twisted Courant algebroid is always
flat and thus also $\D^2=0$. Furthermore, one has 
\begin{prop}\label{p:Royt}  In a twisted generalized standard Vinogradov 
  algebroid as well as in a V-twisted Courant algebroid 
   brackets for $\D$-exact sections satisfy \\(for all $v,v' \in \Gamma(V)$ and
$\phi \in \Gamma(W)$):
\begin{align*} \<\D v,\D v'\>&=0 \\ [\D v,\phi]_W &=0\\
  [\phi,\D v]_W &= \D\<\phi,\D v\>  
\end{align*}
%
\end{prop}
\begin{proof} The first equality follows from $ \<\D v,\D
v'\>=\Wconn_{\D v}v' $ and Proposition \ref{p:Vinogrelem}. For the
remaining two equations one can use the explicit formula
\eqref{stdVbrack} for the derived bracket 
in the case of twisted generalized standard Vinogradov algebroids and
generalize the proof in \cite[p.20, lemma 2.6.2]{Royt99} to
$V$-twisted Courant algebroids in a straightforward way, respectively.
\end{proof}

As the terminology suggests, ordinary Courant algebroids
$(W,[.,.],\rho,\<.,.\>)$ provide examples of $V$-twisted ones: one just
takes $V$ to be  the trivial $\Real$ bundle over the base $M$ so that its
sections can be identified with functions on $M$ and identify $\Wna$
with $\rho$ in this case. For the reverse direction we need an
additional condition: Consider a V-twisted Courant algebroid with $V$
of rank one. Suppose we can find a section $v$ of $V$ which vanishes
nowhere (consequently $V$ has to be trivial) and which is annihilated
by $\D$, $\D v=0$. Then we can define a non-degenerate symmetric
bilinear form $\<.,.\>_C$ on $W$ by means of
$\<\phi,\psi\>=\<\phi,\psi\>_Cv$.  Moreover, $\<\D (fv),\phi\>=
\left(\rho(\phi) f\right) v$. It now is easy to see that $([.,.]_W,
\rho, \<.,.\>_C)$ defines a Courant algebroid. 

Let us remark, however, that a line-bundle twisted Courant algebroid,i.e.~a $V$-twisted Courant algebroid with $\rk V = 1$, is a strictly more general notion than the one of an ordinary Courant algebroid.  We intend to come back to this elsewhere.

Ordinary Courant algebroids over a point reduce to quadratic Lie
algebras. This is no more the case for the above $V$-twisted
generalizations, where even the bracket need not be antisymmetric. In
this case $W$ is what one calls a Leibniz--Loday algebra and $V$ becomes a
$W$-module by means of $\Wconn_\cdot \colon W \to \End(V)$. Whereas the
r.h.s.~of \eqref{cnSkew} vanishes identically for $M$ being a point,
the analogous second equality in \eqref{fwInvar} does \emph{not} lead
to $[\phi,\phi]_W=0$ for any representation $\Wconn_\cdot$. An explicit
example with a bracket having a symmetric part is given by the
following construction which was motivated by the example in \cite[chap.\ 1]{San07}

\begin{prop}  Let $W:=\End(\Real^2)$, $V:=\Real^3$
\begin{align*}
 P\: W&\to W: (x_{ij})\mapsto\pmatrix{cc}{x_{11}&0\\0&0} \\
 [X,Y]_W &:= P(X)Y -YP(X) \\
 \<X,Y\> &:= (x_{11}y_{12}+x_{12}y_{11}, -x_{11}y_{21}-x_{21}y_{12},x_{22}y_{22}) \\
 \Wconn_Y(p,q,r) &:= (2y_{11}p,2y_{11}q,0)
\end{align*}  Then these 5 structures form a V-twisted Courant algebroid over a point.
\end{prop}
\begin{proof}  Straightforward calculations show that
 $$P(P(X)Y)=P(X)P(Y)=P(XP(Y)) \;. $$  Due to \cite{San07} $(W,[.,.]_W)$ forms a Leibniz--Loday algebra.  By inspection the given inner product is non degenerate.  It is now a straightforward calculation that the given connection fulfills \eqref{fwInvar}.  Thus the given structure is a V-twisted Courant algebroid.
\end{proof}

In fact, even the following statement is true: 
\begin{prop} \label{canonical}
\emph{Any} Leibniz--Loday algebra $(W,[ \cdot , \cdot ]_W)$ becomes a $V$-twisted Courant algebroid over a point in a \emph{canonical} way. 
\end{prop}
\begin{proof}  For this purpose one takes $V$ to be the two-sided ideal generated by quadratic elements, $V=\langle [w,w]_W| w \in W \rangle$, and the inner product to be nothing but the symmetrization of the original bracket, $(w_1, w_2) := \frac{1}{2} [w_1 , w_2 ]_W + \frac{1}{2}[w_2 , w_1 ]_W$, mapping elements from $S^2W$ into $V\subset W$. With the evident choice $\Wconn_{w_1} (w_2) := [w_1,w_2]_W$, the conditions \eqref{fwInvar} follow directly from \eqref{fwJacobi}. 
\end{proof}

In the present context the most important examples of $V$-twisted Courant algebroids are, however, provided by Q2-manifolds. This is the subject of the
following subsection.

\subsection{Lie-2-algebroids and  V-twisted Courant algebroids}
\label{subsec:Lie2}
Every Q2-manifold gives rise to a V-twisted Courant algebroid as follows:
\begin{prop}\label{p:derVcour} For the strongly anchored DGLA $\X_\bullet (\CM)$ of
a Q2-manifold (cf.~Prop.~\ref{p:tangentLie}) the conditions in 
Prop.~\ref{p:derVino} are satisfied and the resulting Vinogradov
algebroid is even a V-twisted Courant algebroid.
\end{prop}

\begin{proof}  First we recall that $d=0$ and thus vector fields of
degree -2 are sections of $V\to M$ while those of degree -1 are
sections of the vector bundle $W\to M$
(cf.~Prop..~\ref{p:tangentLie}). We have 
 $[\psi,f\phi]=f[\psi,\phi]$ for all $\psi, \phi \in \Gamma(W) \cong \X_{-1}(\CM)$ and
 $f \in \smooth(M)\cong C^\infty_0(\CM)$ since there are no negative
degree functions on $\CM$. In every  (local) splitting of
\eqref{sequence} we find $W=E\oplus E^*\otimes V$ and  the bracket 
between tangent vector fields of degree -1 is just the canonical pairing
on this direct sum \be \<X\oplus a,Y\oplus b\> = a(Y)+b(X)
\label{pair} \, ,  \ee
where $X,Y \in \Gamma(E)$ and $a,b \in \Gamma(E^* \otimes V)\cong {}^E\Omega^1(M,V)$.
This is
obviously non-degenerate and surjective. Thus by Prop.~\ref{p:derVino}
we conclude to have a Vinogradov algebroid. According to
Prop.~\ref{p:FW} it now suffices to verify Eq.~\eqref{FW'}: Indeed,
this follows from the compatibility of the Lie bracket with $\D$ and
the fact that $[\psi,v]\equiv 0$ $\forall \psi \in \Gamma(W) \cong \X_{-1}(\CM)$ and $v \in 
\Gamma(V) \cong \X_{-2}(\CM)$, since there are no vector fields of
degree -3.\end{proof} In Section~\ref{s:examQ} we found an
interpretation of the components of the vector field $Q$,
Eq.~\eqref{Q}, on a (split) Q2-manifold in terms of a Lie
2-algebroid. Using the same $Q$ to calculate explicitly the derived
bracket given in the previous proposition, one arrives at the
following
\begin{prop}\label{p:FWder} Let $(E,V, [.,.]_E,\rho_E,t,\Ena,H)$ be a Lie
2-algebroid as defined in Def.~\ref{Lie2def}
. Then $W:=E\oplus E^*\otimes V$ with 
\begin{align} \rho_W(X\oplus a):=&\; \rho_E(X) \, , \label{Rho}\\
  \Wconn_{X\oplus a}v := &\Econn_Xv -\imath_{t(v)}a \label{rhoQCont} \\
  \<X\oplus a,Y\oplus b\> :=&\; \imath_Y a + \imath_X b \, ,  \label{Paar}\\
  [X\oplus a, Y\oplus b]_W :=&\; \left([X,Y]_E  \oplus
    [\imath_X,\ED]b -\imath_Y\ED a  + \imath_X \imath_Y H  \right) \nonumber\\
  &-\left(t(a(Y))\oplus T(a,b)\right)\, , \label{brCont}
\end{align} where $X,Y \in \Gamma(E)$, $a,b \in 
\EO^1(M,V)$, and $v \in \Gamma(V)$, is a $V$-twisted Courant
algebroid. Here $\ED$ has been defined in Eq.~\eqref{genCar} and the
operator $T\colon (E^*\otimes V)^{\otimes2} \to E^*\otimes V$ is induced by $t$ according to the composition of maps 
\begin{equation} 
T(a,b)\equiv a\circ t\circ b - b\circ t\circ a \, . \label{Tdefn}
\end{equation}
\end{prop}
\begin{proof}  Instead of checking all axioms separately one shows that the given algebroid is the V-twisted Courant algebroid from Proposition \ref{p:derVcour}.

The Lie 2-algebroid induces a Q-structure via Theorem \ref{theo1}.  To prove the above formulas one can e.g.\ use local coordinates on the Q-manifold.
\end{proof}

It turns out that under certain circumstances Proposition \ref{p:FWder} can be reversed. For that purpose we now consider a generalization of exact Courant algebroids.
Let $W$ be an ordinary Courant algebroid. Then the anchor map $\rho\colon W\to TM$ together with the inner product induce a map $\rho^* \colon T^*M\to W$, where we used the inner product to identify $W^*$ with $W$.  Due to Proposition \ref{p:Vinogrelem} the composition $\rho\circ\rho^*$ vanishes. Thus for \emph{any} Courant algebroid $W$ one has the complex
 \begin{equation}
 0\to T^*M\xrightarrow{\rho^*} W\xrightarrow{\rho} TM\to 0 \, ; \label{Wexact}
\end{equation}   
$W$ is said to be exact iff the above short sequence is exact.
Examples of exact Courant algebroids occur on $W=TM\oplus T^*M$ via the choice of a closed 3-form $H$ with the bracket \eqref{Cour}.  Conversely given an exact Courant algebroid we can choose an isotropic splitting $j\:TM\to E$\footnote{We will be more explicit on this in the more general context below.} and define a 3-form $H$ as
 $$ H(X,Y,Z):= \<[j(X),j(Y)],j(Z)\> \;. $$  \eqref{cnSkew} shows that $H$ is skew-symmetric in the arguments $X$ and $Y$, while \eqref{cInvar} implies that it is invariant under cyclic permutations of $X$, $Y$ and $Z$. Thus $H$ is a 3-form. The remaining  identity on the bracket implies that this 3-form has to be closed under the de Rham differential.  Finally, a change of splitting, which is necessarily induced by a skew-symmetric 2-form $B$, changes $H$ by an exact term, $H \mapsto H+\ud B$.  Thus, exact Courant algebroids are in 1:1 correspondence to classes of the third real cohomology on $M$, an observation going back to \v{S}evera \cite{SevLett}.

We wish to generalize this result to certain V-twisted Courant algebroids $W$.  Since the inner product on $W$ has values in $V$, the map $\rho \equiv \rho_W \colon W \to TM$ induces a map $\rho^* \colon T^*M \otimes V\to W$ by means of $( \rho^*(\alpha \otimes v), \psi ) = \alpha(\rho(\psi))$ for all $\psi \in W$. Thus, \emph{any} $V$-twisted Courant algebroid $W$ gives rise to the complex
\begin{equation}
 0\to T^*M\otimes V \xrightarrow{\rho^*} W\xrightarrow{\rho} TM\to 0 \, . \label{VWexact}
\end{equation}
Let us call a $V$-twisted Courant algebroid exact if this complex is acyclic, in generalization of ordinary exact Courant algebroids (which fit into this definition since for $V=M\times \Real$ \eqref{VWexact} reduces to \eqref{Wexact}). More generally, consider a map $\pi\colon W\to E$. Together with the V-valued inner product it induces a map $\pi^*\colon E^*\otimes V\to W$.  Here $E^*\otimes V$ is the $V$-dual of E and $W$ the $V$-dual of itself. We are thus led to 
\begin{vdef}
Suppose that in a V-twisted 
Courant algebroid $(V,W)$ the anchor map $\rho_W \colon W \to TM$ factors through a bundle map $\pi\colon W\to E$ for some vector bundle $E$. If the then induced complex 
\[  0 \rightarrow E^*\otimes V \xrightarrow{\pi^*} W \xrightarrow{\pi} E \to 0
   \label{Vexact}\]
is acyclic, we call $(V,W)$ an \emph{exact $V$-twisted (or $V$-exact) Courant algebroid over $E$} and for $E=TM$, $\pi=\rho_W$ simply a \emph{$V$-exact Courant algebroid}.
\end{vdef}
There then is the following remarkable fact:
\begin{thm}\label{thm:der} Let $(V,W)$ be a $V$-exact Courant algebroid over some $E$ with $\rk V>1$ and let $j\: E\to W$ be an isotropic 
splitting  of \eqref{Vexact}.

Then $E$ and $V$ carry canonically the structure of a semistrict Lie 2-algebroid, $(E,V,[\cdot, \cdot]_E,\rho,t,\Ena,H)$, and the structural data of the $V$-twisted Courant algebroid $W$ are necessarily of the form  of Proposition \ref{p:FWder} above. 

A change of splitting $j(X)\mapsto j(X)-B^\#(X)$ results into a change of the underlying data $[.,.]_E$, $H$, and $\Ena$ according to \eqref{split1}--\eqref{split3} furthermore.
\end{thm}

We first remark that an isotropic splitting of \eqref{Vexact} always exists. This is seen as follows: First, as for any exact sequence of vector bundles, we can choose \emph{some} splitting $j_0\:E\to W$ of the short exact sequence \eqref{Vexact}.  This permits us to identify $E\oplus E^*\otimes V\xto\sim W$ via $(X,a)\mapsto j_0(X)+\pi^*(a)$ for $X,Y\in E_x$, $a,b\in E^*\otimes V_x$ and $x\in M$. Computing the symmetric 2-form $L(X,Y):=\<j_0(X),j_0(Y)\>$, $L\:E\otimes E\to V$, it may or may not vanish.  If it does, $j_0$ provides an isotropic splitting.  If it does not, it induces a map $L^\#\:E\to E^*\otimes V,\, X\mapsto L(X,\bullet)$.  If we change $j_0$ by $-\pi^*\circ L^\#$, we obtain an isotropic splitting $W\cong E\oplus E^*\otimes V$ and the inner product on $W$ is identified with the standard product on $E\oplus E^*\otimes V$, $\<j(X)+\pi^*(a),j(Y)+\pi^*(b)\>=a(Y)+b(X)$.

We henceforth do no more write $j$ and $\pi^*$ explicitly, but use the identification  $W\cong E\oplus E^*\otimes V$ induced by the isotropic splitting to write elements of $W$ in the form $X\oplus a$ (instead of $j(X)+\pi^*(a)$) with $X \in E$ and $a \in \Hom(E,V)$.

We also remark that the condition on the rank of $V$ is necessary. The 1:1 relation between $V$-exact Courant algebroids over $E$ and semi-strict Lie 2-algebroids $(E,V)$ breaks down explicitly in the case of $\rk V=1$. While due to Proposition \ref{p:FWder} also for $V$ being a line bundle, one obtains a $V$-exact Courant algebroid over $E$ from any semi-strict Lie 2-algebroid, not any such a generalized Courant algebroid is of this form, even if the line bundle is trivial, $V=M \times \R$.
This will become clear within the proof below, which will also permit us to provide an explicit example of this fact right after the proof.
\begin{proof}
%
%
Our task now is to retrieve the data
$[.,.]_E$, $\rho_E$, $t$, $\Ena$, and $H$ from the data of the $V$-twisted exact Courant algebroid over $E$ and to show that 
they satisfy all the identities they have to so as to form a Lie 2-algebroid.

We start with Equation \eqref{rhoQCont} and \emph{define} $\Ena$ and $t$ implicitly by means of
\[
  \Wconn_{X\oplus a} v =: \Econn_Xv -a(t(v)) \, , \label{naTdefn}
\]
which has to hold true for all $v \in \Gamma(V)$ and $X \oplus a \in \Gamma(E) \oplus \Gamma(\Hom(E,V)) \cong \Gamma(W)$. We also know by assumption that $\rho_W = \rho_E \circ \pi$ for some $\rho_E \colon E \to TM$. Noting that after the choice of a splitting, $\pi$ just becomes projection to the first factor, $\pi(X \oplus a)=X$, the Leibniz property of $\Wna$ and Definition \eqref{naTdefn} (for simplicity one may set $a=0$) ensure that $\Ena$ is an $E$-covariant derivative on $V$ with the anchor map $\rho_E$, as anticipated by the notation chosen.  Returning with this information to \eqref{naTdefn}, it implies that the map $t$ is $C^\infty(M)$-linear, and thus comes from a bundle map $t\colon V\to E$.

Next, anticipating what we want to achieve in \eqref{brCont}, we \emph{define} $[.,.]_E$ and some  ${H}$ from $[X,Y]_W$ by means of
$$  [X\oplus0,Y\oplus0]_W =: [X,Y]_E\oplus -{H}(X,Y) \, .$$
First we can conclude from the known properties of the left-hand-side that $[X,Y]_E$ satisfies a Leibniz property with respect to the second argument for the anchor map $\rho_E$. Likewise we also follow that $H(\cdot, \cdot)$ is $C^\infty$-linear in the second argument. Moreover, as a consequence of the isotropy of the chosen splitting, $(X\oplus0,Y\oplus0)=0$, and the second equation of \eqref{fwInvar} yields antisymmetry of the $E$-bracket and $H$ in $X$ and $Y$. By definition, $H(X,Y) \in \Gamma(E^* \otimes V)$. The first equation of \eqref{fwInvar} then implies that $H(X,Y,Z) := \imath_Z H(X,Y)$ 
is also antisymmetric in $Y$ and $Z$. Thus it is antisymmetric in all three arguments and the $C^\infty$-linearity extends to all these arguments. All together this implies that  $H\in\EO^3(M,V)$.

We now have all of the ingredients $[.,.]_E$, $\rho_E$, $t$, $\Ena$, and $H$ necessary to define the bracket \eqref{brCont}. Let us denote the r.h.s.~of this equation by $[X\oplus a,Y\oplus b]_W^0$ and parameterize the difference to this form of the bracket by several bilinear maps $\delta_i$ as follows:
\begin{align}\label{ansatzbr}
  [X\oplus a,Y\oplus b] =: &\;[X\oplus a,Y\oplus b]_W^0 +\\ &+\delta_1(a,Y)+\delta_2(b,X)+\delta_3(a,b)  
    \oplus  \delta_4(X,b)+\delta_5 (Y,a) +\delta_6(a,b) \, . \nonumber
\end{align}
These maps are certainly $\R$-bilinear, but due to the property of $[X\oplus a,Y\oplus b]_W^0$ they are even $C^\infty(M)$-bilinear (due to the already established features of the quantities entering this bracket).

We now want to show that for $\rk V\ne1$, the Axiom~\eqref{fwInvar} implies that $\delta_1 = \ldots = \delta_6=0$, while the Leibniz--Loday property \eqref{fwJacobi} implies the three Equations \eqref{HJacobi}, \eqref{HRep}, and \eqref{DH} needed to endow $(E,V)$ with the structure of a semi-strict Lie 2-algebroid. 

The term in the center of  \eqref{fwInvar} does not depend on the bracket at all and these terms are precisely balanced by the terms coming from $[X\oplus a,Y\oplus b]_W^0$ in the right equation. This implies that for equal arguments, $X=Y$ and $a=b$, all the delta-contributions have to vanish in \eqref{ansatzbr}. This is the case iff
\begin{equation}
\delta_1=-\delta_2 \; , \quad \delta_4 = -\delta_5 \; , \label{delta4} 
\end{equation}
and  $\delta_3$ and $\delta_6$ are antisymmetric in their two arguments. We now turn to the l.h.s.~of \eqref{fwInvar}. The left Equation \eqref{fwInvar} then gives the following additional constraints:
\begin{eqnarray}
\delta_5 &=& 0 \, ,\label{delta5}\\
b(\delta_2(b,X)) &=& 0 \, ,\label{delta2}\\
b(\delta_3(a,b)) &=& 0\, , \label{delta3}\\
\iota_Y\delta_6(a,b)+ b(\delta_1(a,Y)) &=& 0 \, ,\label{delta6} 
\end{eqnarray}
to hold true for all choices of arguments. To exploit Equations \eqref{delta2} and \eqref{delta3}, we prove the following
\begin{lemma}\label{p:Jh=0} 
 Let $V_1$ and $V_2$ be two vector spaces and let $\Delta \in \left(\Hom(V_1, V_2)\right)^* \otimes V_1 \cong (V_1)^{\otimes 2} \otimes V_2^*$ satisfying \begin{equation} \label{aa}
  a(\Delta(a))=0
\end{equation} for all $a \in \Hom(V_1, V_2)$. Then, for $\dim V_2 >1$, necessarily $\Delta =0$, while for $V_2 \cong \R$ any $\Delta \in \Lambda^2 V_1$ has this property.
\end{lemma}
\begin{proof}[{ (of Lemma \ref{p:Jh=0})}] 
Let us first consider the case of a factorizing $a$: $a = \alpha \otimes v$ with $\alpha \in V_1^*$ and $v \in V_2$. Then \eqref{aa} holds true iff $\Delta \in \Lambda^2(V_1) \otimes V_2^*$. It remains to check \eqref{aa} for $a = a_1 + a_2$ with both summands such products. This is satisfied iff the contraction of $\Delta$ in the $V_2$-slot with the first factor of any term of the form $v_1 \wedge v_2 \in \Lambda^2 V_2$ has to vanish. If $\dim V_2 = 1$, this condition is empty and the statement follows from $\Lambda^2(V_1) \otimes \R^* \cong \Lambda^2(V_1)$. In the case that $V_2$ has at least two directions, on the other hand, evidently one is forced to a vanishing $\Delta$.
\end{proof}
By $C^\infty(M)$-linearity, the Equations \eqref{delta2} and \eqref{delta3} are pointwise in $M$ and for $\rk V>1$ we can conclude the vanishing of $\delta_2$ and $\delta_3$ by setting $V_1 := E|_x$, $V_2 :=V|_x$ for any fixed $x \in M$ and $\delta(\cdot)$ equal to 
$\delta_2(\cdot , X)|_x$ and $\delta_3(a,\cdot)|_x$, respectively (again for any fixed $X$ and $a$, respectively). 

Thus, if the rank condition on $V$ is satisfied, Equations \eqref{delta4}---\eqref{delta6} imply the vanishing of $\delta_i$, $i=1,\ldots,6$. It now remains to exploit Equation \eqref{fwJacobi}. 

It is a straightforward calculation to check that  \eqref{fwJacobi}---together with the already established form of the bracket $[ \cdot , \cdot ]_W = [ \cdot , \cdot]_W^0$, i.e.~\eqref{brCont} holding true,---implies the remaining identities \eqref{HJacobi}, \eqref{HRep}, and \eqref{DH}. For example, regarding the $\Gamma(E)$-part of  \eqref{fwJacobi} yields 
\begin{gather}
 [[X,Y]_E,Z]_E -[t(\imath_Y a),Z]_E-t\left( \imath_Z\left([\imath_X,\ED]b -\imath_Y\ED a  + \imath_X \imath_Y H  -T(a,b)\right)\right)= \nonumber \\
  = [X,[Y,Z]_E]_E -[X,t(\imath_Zb)]_E -t\left( a\left([Y,Z]_E-t(\imath_Z b)\right)\right) - \left(X \leftrightarrow Y \,\,\text{and}\,\, a \leftrightarrow b\right)\, ,
\end{gather}
which, after the dust clears and using the definitions \eqref{genCar} and \eqref{Tdefn} for $\ED$ and $T$, respectively, yields the the first of the three equations  
and the right hand equation of \eqref{repr2}. Establishing the three quadratic equations on the structural data for $E$ and $V$ shows that they form a semi-strict Lie 2-algebroid.

Alternatively one may also proceed as follows in the last step: Given the above data, we can define a vector field of the form \eqref{Q} on  $E[1]\oplus V[2]$. Vector fields of degree -1 can then be identified with sections of $W$. Equation~\eqref{vJacobi} then translates into $0=\frac12[[[[Q,Q],\psi],\phi_1],\phi_2]$. 
 It remains to check that
this equation contains the three Equations \eqref{B3}, \eqref{B5}, and \eqref{B7}. In fact, it contains even all seven components \eqref{B1}--\eqref{B7} of
$Q^2=0$ as one may convince oneself.

A change of splitting of the degree 2 Q-manifold used above corresponds precisely to a change of splitting of the sequence \eqref{Vexact} if one regards the vector fields of degree -1 on that graded manifold.  This then proves the  theorem. 
\end{proof}

We conclude with some related remarks. First, we saw that in the case $\rk V =1$ the discussion of $V$-exact Courant algebroids over some $E$ is at least more involved. Let us restrict to a trivial bundle $V=M \times \R$ so that we can identify $\Hom(E,V)$ with $E^*$. Then, according to the above considerations, the data appearing in Proposition \ref{p:FWder} can be extended by a modification of the bracket by means of two tensors: $\delta_1 = - \delta_2 \sim \delta_6 \in \Gamma(\Lambda^2E \otimes E^*)$ and $\delta_3 \in \Gamma(\Lambda^3 E)$. While the general analysis of this special case gets somewhat involved---one still needs to analyze the consequences of \eqref{fwJacobi}---, we just make two remarks at this point here: First, \emph{every} Lie 2-algebra gives rise to a $V$-exact Courant algebroid over $E$ according to Proposition \ref{p:FWder}, also in the case of $\rk V =1$. Second, one simple example of a $V$-exact Courant algebroid over $E$ with $\rk V = 1$ \emph{not} covered by this parameterization may be provided by the following construction: Choose the Lie 2-algebroid to have all structural data vanish and take, in the above parameterization also $\delta_1=0$ but $\delta_3\equiv h \in \Gamma(\Lambda^3 E)$ arbitrary.  
Then for arbitrary sections in $E \oplus E^*\otimes V \cong E \oplus E^*$ one has 
\begin{equation}
[X \oplus \alpha , Y \oplus \beta ]_W := h(\alpha, \beta) \oplus 0 \, . 
\end{equation} 
It is evident that double brackets always vanish with this definition, and thus Equation \eqref{fwJacobi} is satisfied trivially. But also each of the three terms in \eqref{fwInvar} vanish by the antisymmetry of $h$ in all three arguments and the vanishing of $\Wna$, so that also these two equations are satisfied. 

We intend to come back to a more systematic investigation of line-bundle-twisted exact Courant algebroids over some $E$ elsewhere.

\newpage
\appendix
\section{$L_\infty$-algebroids and Q-structures}\label{s:Linf}
\subsection{$L_\infty$-algebras}\label{s:Linf1}
The notion of an $L_\infty$-algebra goes back to Lada and Stasheff in \cite{SHLA}.  Before recalling their definition, we remind the reader of the Koszul sign.  Let $V_\bullet$ be a $\Z$-graded vector space.  We define the graded symmetric powers as
\[  S^\bullet V:= \bigoplus_{n\ge0} \otimes^n V/\langle v\otimes w-(-1)^{|v|\,|w|}w\otimes v \, | \; v,w\in V\rangle \, ,
\] where $\langle \dots\rangle$ denotes the two-sided ideal generated by the given elements.  We denote the product in $S^\bullet V$ as $\cdot$ or omit the dot for simplicity. The Koszul sign of a permutation $\sigma\in S_n$ w.r.t.~$n$ homogeneous elements $v_1,\dots,v_n\in V_\bullet$ is now
\[  v_1\dotsm v_n = \Koszul(\sigma)v_{\sigma1}\dotsm v_{\sigma n}.
\]  Note that $\Koszul(\sigma)$ depends on $\sigma$ as well as on the degrees of all $v_k$.  By $(-1)^\sigma$ we denote the usual sign of a permutation $\sigma\in S_n$.

Correspondingly we introduce the graded skew-symmetric products
\[  \wedge^\bullet V = \bigoplus_{n\ge0} \otimes^n V/\langle v\otimes w \boldsymbol{+}(-1)^{|v|\,|w|}w\otimes v\rangle.
\]  We denote the induced product in $\wedge^\bullet V$ as $\wedge$.

\begin{vdef}  An \emph{$L_\infty$-algebra} is a $\Z$-graded vector space $V_\bullet$ together with a family of maps $l_k\:\wedge^kV\to V$ of degree $2-k$ subject to the equations, for all $n\ge1$,
\[\label{LinftyJacobi1}  \sum_{i+j=n+1}(-1)^{i(j+1)}\mkern-20mu \sum_{\sigma\in\Unsh(j,i-1)} (-1)^\sigma\Koszul(\sigma) l_i(l_j(v_{\sigma1},\dots,v_{\sigma j}),v_{\sigma j+1},\dots,v_{\sigma n}) = 0
\] where $v_1,\dots,v_n\in V_\bullet$ are homogeneous, $\Unsh(j,i-1)$ are the unshuffles of type $(j,i-1)$.

An \emph{$n$-term $L_\infty$-algebra} is an $L_\infty$-algebra concentrated in degrees $-n+1,\dots,0$.
\end{vdef}
There is an equivalent definition which requires fewer signs in the Jacobi identities and which we will use henceforth.  We borrow it from \cite{Vor05,MZ12}.
\begin{vdef}  An $L_\infty[1]$-algebra is a graded vector space $V_\bullet=\bigoplus_{n\in\Z} V_n$ together with a family of multilinear operators $a_k\:S^k V\to V$ of degree $+1$, such that the following relations, called Jacobi identities, hold for all $n\ge1$,
\[\label{LinftyJacobi}  \sum_{i+j=n+1} \sum_{\sigma\in\Unsh(j,i-1)} \Koszul(\sigma) a_i(a_j(v_{\sigma1},\dots,v_{\sigma j}),v_{\sigma j+1},\dots,v_{\sigma n}) = 0\, .
\] Here all $v_k\in V$ are homogeneous and $\Unsh(j,i-1)$ are the unshuffles of type $(j,i-1)$.

An \emph{$n$-term $L_\infty[1]$-algebra} is an $L_\infty[1]$-algebra concentrated in degrees $-n,\dots,-1$.
\end{vdef}

\begin{example}{}$\quad$
\vspace{-2mm}
\begin{enumerate}\zero  Given a Lie algebra $(\g,[.,.])$, then $a_2=[.,.]$, $a_k=0$ for $k\ne2$  and $V_{-1}=\g$, $V_k=0$ for $k\ne-1$ is a 1-term $L_\infty[1]$-algebra.
\item  Given a complex $(V_\bullet,d)$, then $a_1=d$ and $a_k=0$ for $k\ne1$ is an $L_\infty[1]$-algebra.
\item  Given a DGLA $(\g_\bullet,[.,.],d)$, then $a_1=d$, $a_2(v,w)=(-1)^{|w|}[v,w]$, $a_k=0$ for $k\ne1,2$ with $V_\bullet=\g_{\bullet-1}$ makes it an $L_\infty[1]$-algebra.
\end{enumerate}
\end{example}

Given a Q-structure on a graded vector space $V_\bullet$ 
with $Q|_0=0$,\footnote{$Q|_0$ means that we compute the vector field at the origin of the vector space, i.e.\ all the coordinates $q^\alpha=0$, but $\partial_{q^\alpha}$ remains.} we denote the components of $Q$ as
\[  Q= \sum_{k\ge1} \tfrac1{k!}C_{\alpha_1,\dots,\alpha_k}^\beta q^{\alpha_1}\dots q^{\alpha_k}\partial_{q^\beta} \;.
\] 
Note that the above sum for Q starts at $k=1$.  The most general vector field of degree 1 would start with a constant term $k=0$ that corresponds to an element $a_0\in V_1$.  Our $L_\infty$-algebras correspond to pointed Q-structures.  This is automatically fulfilled for $\N_-$-graded Q vector spaces.

\begin{lemma}  Given a $\Z$-graded vector space $V$ together with a vector field $Q$ with $Q|_0=0$ and the identification of the components 
\[  a_k(q^*_{\alpha_1},\dots,q^*_{\alpha_k}) = C^\beta_{\alpha_1\dotsm \alpha_k} q^*_{\beta}
\] then $Q$ is a Q-structure iff the $a_k$ form an $L_\infty[1]$-algebra.
\end{lemma}

\begin{proof}  Given the formula for $Q$, we can write down a formula for $Q^2$
\begin{gather*}
 \begin{split} Q^2 &= \sum_{n\ge k\ge1} \frac1{k!(n-k-1)!} C^\beta_{\alpha_1\dotsm\alpha_k} C^\gamma_{\beta,\alpha_{k+1}\dotsm\alpha_n}q^{\alpha_1}\dotsm q^{\alpha_n}\partial_{q^\gamma}.  \\
   &= \sum_{n\ge k\ge 1}\sum_{\substack{\alpha_1\le\alpha_2\le\dots\le\alpha_n\\ \beta,\gamma}}\sum_{\sigma\in\Unsh(n-k+1,k-1)} \Koszul(\sigma)C^\beta_{\alpha_{\sigma1},\dots\alpha_{\sigma k}} C_{\beta\alpha_{\sigma k+1},\dots,\alpha_{\sigma n}}^\gamma q^{\alpha_1}\dotsm q^{\alpha_n} \partial_{q^\gamma}.
 \end{split}
\intertext{where $\le$ is an arbitrary ordering of the indices.  The Jacobi identities for base vectors read as}
 \begin{split} &\sum_{k+l=n+1}\sum_{\sigma\in\Unsh(l,k-1)} \Koszul(\sigma) a_k(a_l( q^*_{\alpha_{\sigma1}},\dots,q^*_{\alpha_{\sigma l}}),q^*_{\alpha_{\sigma l+1}},\dots,q^*_{\alpha_{\sigma n}})  \\
 &= \sum_{\dots}\Koszul(\sigma) C_{\alpha_1\dots\alpha_l}^\beta C_{\beta\alpha_{l+1}\dots\alpha_n}^\gamma q^*_\gamma
 \end{split}
\end{gather*}  Due to the symmetry of the $a_k$ it is sufficient to prove the Jacobi identities for $\alpha_1\le\alpha_2\le\dots\le\alpha_n$.  These can be extracted from the expression for $Q^2$.  Conversely given the Jacobi identities, it is easy to see that all independent terms in $Q^2$ vanish.
\end{proof}

\begin{vdef}  Given two $L_\infty$-algebras $(V,a)$ and $(V',a')$, then a morphism between them is a map $\tilde\phi\:SV\to SV'$ that projects to a family of maps $\phi_k\:S^k V\to V'$, each of degree 0, such that it intertwines the brackets $a$ and $a'$:
\begin{gather}
  \begin{split}
	  \sum_{k=1}^n a'_{n+1-k}\circ\tilde\phi^{(k)}(v_1,\dots,v_n)
	  = \sum_{k=1}^n \phi_{n+1-k}\circ\tilde{a}_k(v_1,\dots,v_n)
  \end{split}
\intertext{where $\tilde{a}_k$ means that we extend the $k$-ary function to $n$ arguments as a derivation, i.e.}
  \begin{split}	\tilde{a}_k(v_1,\dots,v_n)= \sum_{\sigma\in\Unsh(k,n-k)} \Koszul(\sigma) (a_k(v_{\sigma1},\dots,v_{\sigma k}),v_{\sigma k+1},\dots,v_{\sigma n})
  \end{split}
\end{gather}
Let $\Part(n)$ be the partitions of the positive integer $n$ in positive integer summands ordered by size.  For any particular Partition $p=(p_1,\dots,p_k)$, let further $\Unsh(p)$ be the poly-unshuffles of $n$ elements into groups according to the partition.  Finally, let $k=|p|_0$ be the number of summands in $p$.
\[
  \begin{split}
  \tilde\phi^{(k)}(v_1,\dots,v_n) = \sum_{\substack{p\in\Part(n)\\ k=|p|_0}}
  \sum_{\sigma\in\Unsh(p)}
  \Koszul(\sigma) \big(
    &\phi_{p_1}(v_{\sigma 1},\dots,v_{\sigma p_1}), 
     \phi_{p_2}(v_{\sigma(p_1+1)},\dots,v_{\sigma(p_1+p_2)}), \\
	&\dots, \phi_{p_k}(v_{\sigma(n-p_k+1)},\dots,v_{\sigma n}) \big)
  \end{split}
\]
If the higher $\phi_k$, $k\ge2$ vanish, we say that it is a \emph{strong homomorphism}.  If $\phi_1$ is invertible, then we say that $\phi$ is an \emph{isomorphism}.  If $\phi_1=\id_V$ (and $V'_\bullet=V_\bullet$) we say that $\phi$ is an \emph{essential isomorphism}.
\end{vdef}
Explicitely in terms of $(\phi_k)$, the first formulas read:
\begin{gather*}
	a'_1\circ\phi_1 = \phi_1\circ a_1, \\
	a'_2\circ(\phi_1\otimes\phi_1) +a'_1\circ\phi_2 = \phi_1\circ a_2 +\phi_1\circ (\tilde{a}_1\otimes\id) \pm\phi_1\circ(\id\otimes a_1),  \\
  \begin{split}
	a'_3\circ(\phi_1\otimes\phi_1\otimes\phi_1) 
	+&\,a'_2\circ(\phi_1\otimes\phi_2+\phi_2\otimes\phi_1) +a'_1\circ\phi_3 \\
	=&\; \phi_1\circ a_3 +\phi_2\circ(\tilde{a}_2\otimes\id) \\ 
	&\pm\phi_2\circ(\id\otimes a_2) \pm\phi_2\circ(a_2\otimes\id)\circ\sigma_{23} \\ 
	&+\phi_3\circ(\tilde{a}_1\otimes\id\otimes\id) \pm\phi_3\circ(\id\otimes a_1\otimes\id) \pm\phi_3\circ(\id\otimes\id\otimes a_1)
  \end{split}
\end{gather*} where $\sigma_{23}$ is the flip of the second with the third argument.

We see that $\phi_1$ itself is a chain map (
 since $a_1$ squares to zero as a consequence of Eq.~\eqref{LinftyJacobi}, $(V_\bullet, a_1)$ and $(V'_\bullet, a'_1)$ are each a complex). $\phi_2$ is the correction for $\phi_1$ to be a homomorphism of 2-brackets $a_2$ and $a'_2$.  $(\phi_2,\phi_3)$ is the correction for $\phi_1$ to be a homomorphism of 3-brackets $a_3$ and $a'_3$, and so on.

Also the above notion of a morphism of $L_\infty$-structures simplifies considerably in terms of the Q-language: it just reduces to a Q-morphism, i.e.~to a chain map between the graded symmetric algebras (of ``functions''). Similarly one can formulate it in terms of coalgebras \cite[Defn.~4.4]{Kon97}.

The above defintion is compatible with the following one for a morphism between $L_2$-algebras from \cite{Baez03vi}. 
\begin{example} Given two 2-term $L_\infty$-algebra structures $(a_1,a_2,a_3)$ and $(a'_1,a'_2,a'_3)$ on the same $\N_-$-graded vector space $(V_{-1},V_{-2})$, then an \emph{essential isomorphism} of $L_2$-algebras consists of a map $\phi_2\:S^2V_\bullet\to V_\bullet$ of degree 0 subject to the rules
\begin{gather*}  a'_1 = a_1, \\
  \begin{split} &\forall v_i\in V_\bullet:  \\
  &a'_2(v_1,v_2) -a_2(v_1,v_2) = (-1)^{|v_1|}\phi_2(v_1,a_1(v_2)) +\phi_2(a_1(v_1),v_2) -a'_1\circ\phi_2(v_1,v_2),
  \end{split} \\
 \begin{split} a'_3(x,y,z) - a_3(x,y,z) =\;& -\phi_2(x,a_2(y,z)) +\phi_2(a_2(x,y),z) +\phi_2(y,a_2(x,z)) \\
   &-a'_2(\phi_2(x,y),z) +a'_2(y,\phi_2(x,z)) -a'_2(x,\phi_2(y,z)), \\&\forall x,y,z\in V_{-1}.  \end{split}
\end{gather*}
\end{example}
The somewhat lengthy formulas arise naturally in the description with Q-structures.  Namely, a change of splitting of \eqref{sequence} is described by coefficients $(L_{ab}^D)$ as $\tilde{b}^D = b^D +\tfrac12L_{ab}^D\xi^a\xi^b$.  Since we work over a point, these form a linear map $\phi_2\:\wedge^2E\to V$ and then the Formulas~\eqref{split1}--\eqref{split3} give the above formulas.

By the same mechanism it is also possible to obtain the essential isomorphisms of 3-term $L_\infty$-algebra structures:
\begin{example}  Given two 3-term $L_\infty$-algebra structures $(a_1,a_2,a_3,a_4)$ and $(a'_\bullet)$ over the same graded vector space $(V_{-1},V_{-2},V_{-3})$, an essential $L_\infty$-morphism between them is a collection of maps $\phi_2\:S^2V_\bullet\to V_\bullet$, $\phi_3\:S^3V_\bullet\to V_\bullet$ of degree 0 such that
\begin{gather*}  a'_1 = a_1,  \\
  \begin{split} &\forall v_i\in V_\bullet: \\
	  &a'_2(v_1,v_2) -a_2(v_1,v_2) = (-1)^{|v_1|}\phi_2(v_1,a_1(v_2)) +\phi_2(a_1(v_1),v_2) -a'_1\circ\phi_2(v_1,v_2),
  \end{split} \\
  \begin{split} a'_3(v,y,z) -a_3(v,y,z) =\;
   & (-1)^{|v|}\phi_2(v,a_2(y,z)) 
	 +\phi_2(a_2(v,y),z) -(-1)^{|v|}\phi_2(y,a_2(v,z)) \\
   &-a'_2(\phi_2(v,y),z) +a'_2(\phi_2(v,z),y) -a'_2(v,\phi_2(y,z)) \\
   &+\phi_3(a_1(v),y,z) +(-1)^{|v|}\phi_3(v,a_1(y),z) \\
   &-(-1)^{|v|}\phi_3(v,y,a_1(z)) -a'_1\circ\phi_3(v,y,z)
  \end{split} \\
  \begin{split}  a'_4(w,x,y,z) -a_4(w,x,y,z) =\; 
	  &-\phi_2(w,a_3(x,y,z)) -a'_3(\phi_2(w,x),y,z) \\
    &+\phi_3(w,x,a_2(y,z)) -a'_2(\phi_3(w,x,y),z) \\ 	 
	&-a'_2(\phi_2(w,x),\phi_2(y,z)) +\alt(w,x,y,z) \\
	\forall w,x,y,z\in V_{-1}
  \end{split}
\end{gather*}
\end{example}


\subsection{$L_\infty$-algebroids}\label{s:Linf2}
In Section~\ref{s:examQ} we see however that our $L_\infty$-algebra is infinite-dimensional, because the vector spaces are in fact spaces of sections of certain vector bundles.  Beside continuity of the operations, we also see that the bracket $[.,.]\:\Gamma(E)\wedge\Gamma(E)\to\Gamma(E)$ (which is part of $a_2$) fulfills a Leibniz rule, whereas $H\:\wedge^3E\to V$ (which is $a_3$) is $\smooth(M)$-linear.  We make this observation the foundation of the definition of $L_\infty[1]$-algebr\emph{oids} as follows.

\begin{vdef}\begin{enumerate}\item  An \emph{$L_\infty[1]$-algebroid} is a $\Z$-graded vector bundle $\VV$ over a smooth manifold $M$ together with a family of maps $a_k\:S^k V\to V$ and $\rho\:\VV\to TM$ of degree +1, where $V_\bullet=\Gamma(\VV_\bullet)$.  The brackets and $\rho$ have to be $\smooth(M)$-linear, except in degree 2 where the Leibniz rule holds
\[  a_2(X,f\cdot Y) = \rho(X)[f]\cdot Y +f\cdot a_2(X,Y),
\] for $X,Y\in V_\bullet$, $f\in V_0\cong\smooth(M)$ and $\rho(X)=0$ for $|X|\ne-1$.  The $a_k$ are required to fulfill the Jacobi identities \eqref{LinftyJacobi}.

\item  An $\N$-graded $L_\infty[1]$-algebroid is an $\N_-$-graded vector bundle $\VV$ over a smooth manifold $M$ together with a family of maps $a_k\:S^k V\to V$ of degree +1, where $V_\bullet=\Gamma(\VV'_\bullet)$ and $\VV'=\VV$ except in degree 0 where we add $M\times\R$.  The brackets have to be $\smooth(M)$-linear, except in degree 2 where the Leibniz rule holds
\[\label{NLinfLeibn}  a_2(X,f\cdot Y) = a_2(X,f)\cdot Y +f\cdot a_2(X,Y)
\] for $X,Y\in V_\bullet$, $f\in V_0\cong\smooth(M)$
.  The $a_k$ are required to fulfill the Jacobi identities \eqref{LinftyJacobi}.

We call $(\VV_\bullet,a_\bullet)$ an \emph{$n$-term $L_\infty[1]$-algebroid} if $V_\bullet$ is concentrated only in degrees $-n,\dots,0$.
\end{enumerate}
\end{vdef}
\begin{example}
Given a Lie algebroid $(A\to M,\rho,[.,.])$, then this is a 1-term $L_\infty[1]$-algebroid with $a_1=0$, $a_2(X,Y)=[X,Y]$,  and 
%
$a_2(X,f)=\rho(X)[f]$ for $X\in V_1=\Gamma(A)$ and $f\in V_0=\smooth(M)$.  The second Jacobi identity says that $\rho$ is a morphism of brackets, i.e.\ from $a_2$ on $V_1$ to the commutator bracket on $TM$ (which is always fulfilled for Lie algebroids).
\end{example}
In order to show the correspondence with N-manifolds, we will now restrict to $\N$-graded $L_\infty$-algebroids.

\begin{example}  Given a $Qp$-manifold, then this corresponds to an $\N$-graded $L_\infty[1]$-algebroid in the following way.  The $\N$-graded manifold corresponds to a tower of affine fibrations
\[  \M_p\susu \M_{p-1}\susu\dots\susu\M_1=E[1]\susu M  \label{tower}
\] where only the last step is a vector bundle.  The fibers are 
\[  \VV_k[k]\embed \M_k\susu \M_{k-1}  \label{sequence2}
\] and the affine structure is
\[\xymatrix{  \pi^*_{k-1}\VV_k[k]\ar[dr]&\laction& \M_k\ar[dl]  \\
  &\M_{k-1}&
}
\] where $\pi_k\:\M_k\to M$ is the projection to the smooth body.  A splitting of \eqref{sequence2} is equivalent to identifying the affine bundle $\M_k\susu\M_{k-1}$ with a graded vector bundle.  Thus after a choice of splittings on every level, we have identified $\M=\M_p\cong\VV$.  In the following we will assume the choice of such a splitting.  Now we write $Q$ in the form
\[\label{Qp}  Q=\rho^i_a\xi^a\partial_{x^i} +\sum_{k\ge1} C_{\alpha_1,\dots,\alpha_k}^\beta q^{\alpha_1}\dotsm q^{\alpha_k}\partial_{q^\beta}
\] where the coefficients $C$ are smooth functions on the base manifold $M$.  
Within this section we use the convention that the $\alpha_i$ and $\beta$ have degree at least 1 and the term with $\rho$ is the only one that has a partial derivative by $x^i$.  From the transformation behavior of the $C$ under coordinate changes with fixed splitting of Equation~\eqref{sequence2} we see that $\rho^i_a$ as well as $C_{\alpha_1,\dots,\alpha_k}^\beta$ except $C_{a\alpha_2}^\beta$ with $|\alpha_2|=|\beta|$ transform as tensors while the latter transform like structure functions of a bracket $|\alpha_2|=|\beta|=1$ or connections $|\alpha_2|=|\beta|\ge2$, respectively.  The definition of the $a_k$ is now as follows
\[\label{multibr}
  a_k(q^*_{\alpha_1},\dots,q^*_{\alpha_k}) = C_{\alpha_1,\dots,\alpha_k}^\beta q^*_\beta,  \\
\] together with $\smooth(M)$-linearity of $a_k$ for $k\ne 2$ and the definition
\[\label{a2function}  a_2(X,f) = \rho(X)[f]
\] where we interpret $\rho(X)=0$ for $|X|\ne-1$.
\end{example}

\begin{prop}\label{p:Lpalgd}  A $p$-term $L_\infty[1]$-algebroid is isomorphic to a $Qp$-manifold with splitting of the Sequences~\eqref{sequence2}.  Given the $Qp$-manifold, then the $\N_-$-graded vector bundle arises from the splitting and the multi-brackets are defined as in Equations~\eqref{multibr} and \eqref{a2function}.  Conversely, given a $p$-term $L_\infty[1]$-algebroid $(\VV_\bullet,a_\bullet)$, then we can define the $Qp$-manifold $\M=\VV$ and the $Q$-structure from Equations~\eqref{multibr} and \eqref{Qp}.
\end{prop}
\begin{proof}  Given the above formula for the $Q$-structure, we can compute $Q^2$ and obtain analogous to $Qp$-vector spaces
\begin{gather*}
 \begin{split} Q^2 =\,& 2\rho^i_{[a}\rho^j_{b],i}\xi^a\xi^b\partial_{x^i} +\tfrac12C^a_{bc}\rho^i_a\xi^b\xi^c\partial_{x^i} +C^a_D\rho^i_ab^D\partial_{x^i}  \\
  &+\sum_{1+|\alpha_2|+\dots+|\alpha_n|=|\beta|+1} \frac1{(n-1)!}\rho^i_{(a}\partial_{x^i}C^\beta_{\alpha_2\dots\alpha_n)}q^aq^{\alpha_2}\dotsm q^{\alpha_n}\partial_{q^\beta}  \\
  &+\sum_{n\ge k\ge1}\sum_{\substack{|\alpha_1|+\dots+|\alpha_k|=|\beta|+1\\ |\beta|+|\alpha_{k+1}|+\dots+|\alpha_n|=|\gamma|+1}} \frac1{k!(n-k)!}C_{\alpha_1,\dots,\alpha_k}^\beta C_{\beta,\alpha_{k+1},\dots,\alpha_n}^\gamma q^{\alpha_1}\dotsm q^{\alpha_n} \partial_{q^\gamma}  \\
  =\,& 2\rho^i_{[a}\rho^j_{b],i}\xi^a\xi^b\partial_{x^i} +\tfrac12C^a_{bc}\rho^i_a\xi^b\xi^c\partial_{x^i} +C^a_D\rho^i_ab^D\partial_{x^i}  \\
  &+\sum_{k\ge2}\sum_{\substack{\alpha_1\le\alpha_2\le\dots\le\alpha_k\\ |\beta|=|\alpha_1|+\dots+|\alpha_k|-1}} \rho^i_{(a}C^\beta_{\alpha_1,\dots,\alpha_k),i}\xi^a q^{\alpha_1}\dotsm q^{\alpha_k}\partial_{q^\beta}  \\
  &+\sum_{n\ge k\ge 1}\sum_{\substack{\alpha_1\le\alpha_2\le\dots\le\alpha_n\\ |\beta|=|\alpha_1|+\dots+|\alpha_n|-1}}\sum_{\sigma\in\Unsh(n-k+1,k-1)} \Koszul(\sigma)C^\beta_{\alpha_{\sigma1},\dots\alpha_{\sigma k}} C_{\beta\alpha_{\sigma k+1},\dots,\alpha_{\sigma n}}^\gamma q^{\alpha_1}\dotsm q^{\alpha_n} \partial_{q^\gamma}  \end{split}
\intertext{on the other hand the Jacobi identities read for every $n\ge1$}
 \begin{split} &\sum_{k+l=n+1}\sum_{\sigma\in\Unsh(l,k-1)} a_k(a_l(q^*_{\alpha_{\sigma1}},\dots,q^*_{\alpha_{\sigma l}}),q^*_{\alpha_{\sigma l+1}},\dots,q^*_{\alpha_{\sigma n}})  \\
  =\,&
   \sum_{k+l=n+1}\sum_{\sigma\in\Unsh(l,k+1)} \Koszul(\sigma)C_{\alpha_{\sigma1},\dots,\alpha_{\sigma l}}^\beta C_{\beta,\alpha_{\sigma l+1},\dots,\alpha_{\sigma n}}^\gamma q^*_\gamma  \end{split}
\intertext{Given a $Q$-structure, it is again sufficient to prove the Jacobi identities for $\alpha_1\le\dots\le\alpha_n$ in any fixed order of the indices together with the Leibniz rule for $Q^2$ for every $n\ge1$}
 \begin{split}  &\sum_{k=1}^n\sum_{\sigma\in\Unsh(k,n-k)_1} \Koszul(\sigma)a_{n+1-k}(a_k(f\cdot v_1,v_{\sigma2},\dots,v_{\sigma k}),v_{\sigma k+1},\dots,v_{\sigma n})  \\
 &+\sum_{k=1}^n\sum_{\sigma\in\Unsh(k,n-k)^1} \Koszul(\sigma)a_{n+1-k}(a_k(v_{\sigma 1},\dots,v_{\sigma k}),f\cdot v_1,v_{\sigma k+2},\dots,v_{\sigma n})  \end{split}
\intertext{where we denote $\Unsh(k,l)_1$ the unshuffles with $\sigma1=1$ and $\sigma\in\Unsh(k,l)^1$ the unshuffles with $\sigma k+1=1$.  For $n\ge4$ we obtain}
 \begin{split} =\,& f\cdot\text{Jacobi}+\sum_{\sigma\in\Unsh(2,n-2)_1} \Koszul(\sigma)a_2(f,v_{\sigma2})a_{n-1}(v_1,v_{\sigma3},\dots,v_{\sigma n})  \\
 &+\sum_{\sigma\in\Unsh(n-1,1)_1} \Koszul(\sigma) a_2(f,v_{\sigma n}) a_{n-1}(v_1,v_{\sigma2},\dots,v_{\sigma n-1})  \\
 &+\sum_{\sigma=(2,3,\dots,n,1)} \Koszul(\sigma) a_2(a_{n-1}(v_2,\dots,v_n),f) \cdot v_1
 \end{split}
\intertext{For $n=3$ we obtain}
 \begin{split} &a_2(a_2(f\cdot v_1,v_2),v_3) -a_2(a_2(f\cdot v_1,v_3),v_2) -a_2(a_2(v_2,v_3),f\cdot v_1)  \\
 =\,& f\bigg(a_2(a_2(v_1,v_2),v_3) -a_2(a_2(v_1,v_3),v_2) -a_2(a_2(v_2,v_3),v_1)\bigg) \\
     &+\bigg(a_2(a_2(f,v_3),v_2) -a_2(a_2(f,v_2),v_3) -a_2(a_2(v_2,v_3),f) \bigg)\cdot v_1
 \end{split}
\end{gather*}  Similar computations for $n=2$ and $n=1$ lead to the remaining terms in $Q^2$.  We see therefore that $Q^2=0$ implies a consistent definition of $a_k$ as well as the Jacobi identities and Leibniz rule.  Conversely given an $L_\infty[1]$-algebroid the definition of $Q$ is coordinate independent and it follows that $Q^2=0$.
\end{proof}

\subsection{The case $p=2$}
\begin{rem}  Given a $Q2$-manifold, we can identify its components according to Theorem~\ref{theo1} as $[.,.]\:\wedge^2\Gamma(E)\to\Gamma(E)$, $\rho\:E\to TM$, $t\:V\to E$, $H\:\wedge^3 E\to V$, $\Ena\:\Gamma(E)\otimes\Gamma(V)\to\Gamma(V)$.  According to Proposition~\ref{p:Lpalgd} we have the following 2-term $L_\infty$-algebroid: $V_0=\smooth(M)$, $V_{-1}=\Gamma(E)$, $V_{-2}=\Gamma(V)$,
\begin{align*}
  a_1\:\quad V_{-2}\to V_{-1}&: v\mapsto t(v), 0 \text{ otherwise},\\
  a_2\:\quad V_{-1}\otimes V_k\mapsto V_k &: \begin{cases}  
    a_2(\phi,\psi)=[\phi,\psi],\\
    a_2(\psi,f)=\rho(\psi)[f],\\
	a_2(\psi,v)=\Econn_\psi v,\\
	0 \text{ otherwise},\end{cases}  \\
  a_3\:\quad \wedge^3 V_{-1}\to V_{-2} &: a_3|_{\wedge^3 V_{-1}}=H, 
    0 \text{ otherwise},
\end{align*}
where $f\in V_0$, $\phi,\psi,\psi_k\in V_{-1}$, and $v\in V_{-2}$.  Now the defining equations correspond to the following Jacobi identities. $n=1$: from the definition of $a_1$; $n=2$: the two equations in the text of the definition of Lie 2-algebroid; $n=3$: Equations~\eqref{Jac}, \eqref{repr1}, and \eqref{rhomor}; $n=4$: Equation~\eqref{strange}; and the higher Jacobi identities vanish trivially.
\end{rem}

\begin{prop}  Given a $Q2$-manifold with splitting $\sigma$ of the Sequence~\eqref{sequence}, then a change of splitting $\sigma\mapsto \sigma+L^\#$ corresponds to an essential isomorphism of 2-term $L_\infty$-algebras $\phi_2=L$.
\end{prop}

\newpage
\addcontentsline{toc}{section}{References}
\newcommand\eprint[2][]{\href{http://arxiv.org/abs/#2}{arXiv: #2}}
\newcommand\doi[1]{\href{http://dx.doi.org/#1}{DOI: #1}}

\begin{thebibliography}{BCD{\etalchar{+}}09}
\providecommand{\urlprefix}{URL }
\expandafter\ifx\csname urlstyle\endcsname\relax
  \providecommand{\doi}[1]{doi:\discretionary{}{}{}#1}\else
  \providecommand{\doi}{doi:\discretionary{}{}{}\begingroup
  \urlstyle{rm}\Url}\fi
\providecommand{\bibAnnoteFile}[1]{%
}
\providecommand{\bibAnnote}[2]{%
}

\bibitem[ACJ05]{ACJ05}
\textsc{P.~Aschieri}, \textsc{L.~Cantini} and \textsc{B.~Jurco}:
  \emph{Nonabelian bundle gerbes, their differential geometry and gauge
  theory}, Commun. Math. Phys., vol. 254, pp. 367--400 (\textbf{2005}),
  \eprint{hep-th/0312154}.
\bibAnnoteFile{ACJ05}

\bibitem[AKSZ97]{AKSZ}
\textsc{M.~Alexandrov}, \textsc{M.~Kontsevich}, \textsc{A.~Schwarz} and
  \textsc{O.~Zaboronsky}: \emph{The geometry of the master equation and
  topological quantum field theory}, Int. J. Mod. Phys., vol. A12, pp.
  1405--1430 (\textbf{1997}), \eprint{hep-th/9502010}.
\bibAnnoteFile{AKSZ}

\bibitem[Bae02]{Baez02jn}
\textsc{J.~C. Baez}: \emph{Higher {Y}ang-{M}ills theory} (\textbf{2002}),
  \eprint{hep-th/0206130}.
\bibAnnoteFile{Baez02jn}

\bibitem[BBH00]{BBH00}
\textsc{G.~Barnich}, \textsc{F.~Brandt} and \textsc{M.~Henneaux}: \emph{Local
  {BRST} cohomology in gauge theories}, Phys. Rep., vol. 338(5), pp. 439--569
  (\textbf{2000}), ISSN 0370-1573, \doi{10.1016/S0370-1573(00)00049-1}.
\bibAnnoteFile{BBH00}

\bibitem[BC03]{Baez03vi}
\textsc{J.~C. Baez} and \textsc{A.~Crans}: \emph{Higher-dimensional algebra
  {VI}: {L}ie 2-algebras}, p.~47 (\textbf{2003}), \eprint{math.QA/0307263}.
\bibAnnoteFile{Baez03vi}

\bibitem[BCD{\etalchar{+}}09]{Bizdadea09}
\textsc{C.~Bizdadea}, \textsc{E.~M. Cioroianu}, \textsc{A.~Danehkar},
  \textsc{M.~Iordache}, \textsc{S.~O. Saliu} and \textsc{S.~C. Sararu}:
  \emph{Consistent interactions of dual linearized gravity in {$D=5$}:
  couplings with a topological {BF} model}, The European Physical Journal, vol.
  C63, pp. 491--519 (\textbf{2009}), \eprint{0908.2169}.
\bibAnnoteFile{Bizdadea09}

\bibitem[BCN{\etalchar{+}}09]{Bizdadea09b}
\textsc{C.~Bizdadea}, \textsc{E.~M. Cioroianu}, \textsc{I.~Negru},
  \textsc{S.~O. Saliu}, \textsc{S.~C. Sararu} and \textsc{O.~Balus}:
  \emph{Reducible second-class constraints of order {$L$}: {A}n irreducible
  approach}, Nuclear Physics B, vol. 812, pp. 12--45 (\textbf{2009}),
  \eprint{0904.1677}.
\bibAnnoteFile{Bizdadea09b}

\bibitem[BCSB10]{Bizdadea10}
\textsc{C.~Bizdadea}, \textsc{E.~M. Cioroianu}, \textsc{S.~O. Saliu} and
  \textsc{E.~M. Babalic}: \emph{Dual linearized gravity in {$D=6$} coupled to a
  purely spin-two field of mixed symmetry {$(2,2)$}}, Fortschritte der Physik,
  vol.~58, pp. 341--363 (\textbf{2010}), \eprint{1103.0623}.
\bibAnnoteFile{Bizdadea10}

\bibitem[BFLS98]{SHLA2}
\textsc{G.~Barnich}, \textsc{R.~Fulp}, \textsc{T.~Lada} and
  \textsc{J.~Stasheff}: \emph{The sh {L}ie structure of {P}oisson brackets in
  field theory}, Comm. Math. Phys., vol. 191(3), pp. 585--601 (\textbf{1998}),
  ISSN 0010-3616, \doi{10.1007/s002200050278}.
\bibAnnoteFile{SHLA2}

\bibitem[BKS05]{BKS}
\textsc{M.~Bojowald}, \textsc{A.~Kotov} and \textsc{T.~Strobl}: \emph{{L}ie
  algebroid morphisms, {P}oisson sigma models, and off-shell closed gauge
  symmetries}, J.Geom.Phys., vol.~54, pp. 400--426 (\textbf{2005}),
  \eprint{math.DG/0406445}.
\bibAnnoteFile{BKS}

\bibitem[BM05]{BM05}
\textsc{L.~Breen} and \textsc{W.~Messing}: \emph{Differential geometry of
  {G}erbes}, Adv.\ Math., vol. 198(732), p.~76 (\textbf{2005}),
  \eprint{math/0106083}.
\bibAnnoteFile{BM05}

\bibitem[BS07]{BS05}
\textsc{J.~C. Baez} and \textsc{U.~Schreiber}: \emph{Higher gauge theory}, in
  \emph{Categories in Algebra, Geometry and Mathematical Physics}, edited by
  \textsc{A.~Davydov} \textsc{et~al.}, \emph{Contemp. Math.}, vol. 431, pp.
  7--30, AMS (\textbf{2007}), \eprint{math/0511710}.
\bibAnnoteFile{BS05}

\bibitem[CF01a]{Catt01}
\textsc{A.~S. Cattaneo} and \textsc{G.~Felder}: \emph{On the {AKSZ} formulation
  of the {P}oisson sigma model}, Lett. Math. Phys., vol.~56(2), pp. 163--179
  (\textbf{2001}), ISSN 0377-9017, \doi{10.1023/A:1010963926853},
  euroConf{\'e}rence Mosh{\'e} Flato 2000, Part {II} (Dijon),
  \eprint{math.QA/0102108}.
\bibAnnoteFile{Catt01}

\bibitem[CF01b]{CrF01}
\textsc{M.~Crainic} and \textsc{R.~L. Fernandes}: \emph{Integrability of {L}ie
  brackets}, Ana.of Math., vol. 157, pp. 575--620 (\textbf{2001}),
  \eprint{math.DG/0105033}.
\bibAnnoteFile{CrF01}

\bibitem[CLS10]{CLS08}
\textsc{Z.~Chen}, \textsc{Z.-J. Liu} and \textsc{Y.~Sheng}:
  \emph{{$E$}-{C}ourant algebroids}, Int. Math. Res. Not. IMRN, (22), pp.
  4334--4376 (\textbf{2010}), ISSN 1073-7928, \doi{10.1093/imrn/rnq053},
  \eprint{0805.4093}.
\bibAnnoteFile{CLS08}

\bibitem[Cou90]{Cour90}
\textsc{T.~Courant}: \emph{{D}irac manifolds}, Trans.A.M.S., vol. 319, pp.
  631--661 (\textbf{1990}), \doi{10.2307/2001258}.
\bibAnnoteFile{Cour90}

\bibitem[CW99]{CaWe99}
\textsc{A.~Canna{s da S}ilva} and \textsc{A.~Weinstein}: \emph{Geometric Models
  for non commutative Algebras}, \emph{Berkeley Mathematics Lecture Notes},
  vol.~10, American Mathematical Society, Providence, RI (\textbf{1999}), ISBN
  0-8218-0952-0, 198 pp., \urlprefix\url{http://math.berkeley.edu/~alanw/}.
\bibAnnoteFile{CaWe99}

\bibitem[Dor87]{Dor87}
\textsc{I.~Y. Dorfman}: \emph{{D}irac structures of integrable evolution
  equations}, Physics Letters A, vol. 125, pp. 240--246 (\textbf{1987}).
\bibAnnoteFile{Dor87}

\bibitem[Dor93]{Dor93}
---{}---{}--- \emph{{D}irac stuctures and integrablility of nonliear evolution
  equations}, Wiley (\textbf{1993}).
\bibAnnoteFile{Dor93}

\bibitem[FRS13]{FRS13}
\textsc{D.~Fiorenza}, \textsc{C.~L. Rogers} and \textsc{U.~Schreiber}: \emph{A
  higher {C}hern--{W}eil derivation of {AKSZ} $\sigma$-models},
  Int.J.Geom.Meth.Mod.Phys., vol.~10, p. 1250\,078 (\textbf{2013}),
  \doi{10.1142/S0219887812500788}, \eprint{1108.4378}.
\bibAnnoteFile{FRS13}

\bibitem[Gr{\"u}09]{Gru09}
\textsc{M.~Gr{\"u}tzmann}: \emph{{C}ourant algebroids: Cohomology and Matched
  Pairs}, Ph.D. thesis, Pennsylvania State University (\textbf{2009}),
  \eprint{1004.1487}.
\bibAnnoteFile{Gru09}

\bibitem[HM90]{MaH90}
\textsc{P.~J. Higgins} and \textsc{K.~Mackenzie}: \emph{Algebraic constructions
  in the category of {L}ie algebroids}, J. Algebra, vol. 129(1), pp. 194--230
  (\textbf{1990}), ISSN 0021-8693, \doi{10.1016/0021-8693(90)90246-K}.
\bibAnnoteFile{MaH90}

\bibitem[HM93]{MaH93}
\textsc{P.~J. Higgins} and \textsc{K.~C.~H. Mackenzie}: \emph{Duality for
  base-changing morphisms of vector bundles, modules, {L}ie algebroids and
  {P}oisson structures}, Math. Proc. Cambridge Philos. Soc., vol. 114(3), pp.
  471--488 (\textbf{1993}), ISSN 0305-0041, \doi{10.1017/S0305004100071760}.
\bibAnnoteFile{MaH93}

\bibitem[HS09]{HS08}
\textsc{M.~Hansen} and \textsc{T.~Strobl}: \emph{First class constrained
  systems and twisting of {C}ourant algebroids by a closed 4-form}
  (\textbf{2009}), \eprint{0904.0711}.
\bibAnnoteFile{HS08}

\bibitem[HT92]{HT}
\textsc{M.~Henneaux} and \textsc{C.~Teitelboim}: \emph{Quantization of Gauge
  Systems}, Princeton University Press (\textbf{1992}).
\bibAnnoteFile{HT}

\bibitem[Ike94]{Ikeda94}
\textsc{N.~Ikeda}: \emph{{Two-dimensional gravity and nonlinear gauge theory}},
  Annals Phys., vol. 235, pp. 435--464 (\textbf{1994}),
  \doi{10.1006/aphy.1994.1104}, \eprint{hep-th/9312059}.
\bibAnnoteFile{Ikeda94}

\bibitem[Izq99]{Izq99}
\textsc{J.~M. Izquierdo}: \emph{{F}ree differential algebras and generic {2D}
  dilatonic (super)gravities}, Phys. Rev. D, vol.~59, p. 084\,017
  (\textbf{1999}), \eprint{hep-th/9807007}.
\bibAnnoteFile{Izq99}

\bibitem[Kon03]{Kon97}
\textsc{M.~Kontsevich}: \emph{Deformation quantization of {P}oisson manifolds},
  Lett.Math.Phys., vol.~66, pp. 157--216 (\textbf{2003}),
  \eprint{q-alg/9709040}.
\bibAnnoteFile{Kon97}

\bibitem[Kos96]{YKS96}
\textsc{Y.~Kosman{n-S}chwarzbach}: \emph{From {P}oisson algebras to
  {G}erstenhaber algebras}, Ann.Inst.Fourier, vol. 46(5), pp. 1243--1274
  (\textbf{1996}),
  \urlprefix\url{http://www.math.polytechnique.fr/cmat/kosmann/fourier96.pdf}.
\bibAnnoteFile{YKS96}

\bibitem[Kos97]{YKS97}
---{}---{}--- \emph{Derived brackets and the gauge {L}ie algebra of closed
  string theory}, pp. 53--61, Heron Press (\textbf{1997}),
  \urlprefix\url{http://www.math.polytechnique.fr/cmat/kosmann/gp21.pdf}.
\bibAnnoteFile{YKS97}

\bibitem[Kos04]{YKS03en}
---{}---{}--- \emph{Derived brackets}, Lett. Math. Phys., vol.~69, pp. 61--87
  (\textbf{2004}), \eprint{math.DG/0312524}.
\bibAnnoteFile{YKS03en}

\bibitem[KS07]{KS07}
\textsc{A.~Kotov} and \textsc{T.~Strobl}: \emph{Characteristic classes
  associated to {Q}-bundles} (\textbf{2007}), \eprint{0711.4106}.
\bibAnnoteFile{KS07}

\bibitem[KS10]{KS08}
---{}---{}--- \emph{Generalizing geometry---algebroids and sigma models}, in
  \emph{Handbook of pseudo-{R}iemannian geometry and supersymmetry}, \emph{IRMA
  Lect. Math. Theor. Phys.}, vol.~16, pp. 209--262, Eur. Math. Soc., Z\"urich
  (\textbf{2010}), \doi{10.4171/079-1/7}, \eprint{1004.0632}.
\bibAnnoteFile{KS08}

\bibitem[KSS05]{DSM}
\textsc{A.~Kotov}, \textsc{P.~Schaller} and \textsc{T.~Strobl}: \emph{{D}irac
  sigma models}, Comm. Math. Phys., vol. 260(2), pp. 455--480 (\textbf{2005}),
  \eprint{hep-th/0411112}.
\bibAnnoteFile{DSM}

\bibitem[KSS14]{SKS14}
\textsc{A.~Kotov}, \textsc{H.~Samtleben} and \textsc{T.~Strobl}:
  (\textbf{2014}), final preparations.
\bibAnnoteFile{SKS14}

\bibitem[L{i-}11]{LiB11}
\textsc{D.~L{i-B}land}: \emph{{AV-C}ourant algebroids and generalized {CR}
  structures}, Canad. J. Math., vol.~63(4), pp. 938--960 (\textbf{2011}), ISSN
  0008-414X, \doi{10.4153/CJM-2011-009-1}, \eprint{0811.4470}.
\bibAnnoteFile{LiB11}

\bibitem[LS93]{SHLA}
\textsc{T.~Lada} and \textsc{J.~Stasheff}: \emph{Introduction to sh {L}ie
  algebras for physicists}, Journal of Theoretical Physics, vol.~32(7), pp.
  1087--1103 (\textbf{1993}), \eprint{hep-th/9209099}.
\bibAnnoteFile{SHLA}

\bibitem[LSS]{SLS14}
\textsc{S.~Lavau}, \textsc{H.~Samtleben} and \textsc{T.~Strobl}: \emph{Hidden
  {Q}-structure and {L}ie-3-algebra for non-abelian superconformal models in
  six dimensions.}, final preparations.
\bibAnnoteFile{SLS14}

\bibitem[LWX97]{Xu97}
\textsc{Z.-J. Liu}, \textsc{A.~Weinstein} and \textsc{P.~Xu}: \emph{{M}anin
  triples for {L}ie bialgebroids}, J. Diff.\ Geom, vol. 45/3, pp. 547--574
  (\textbf{1997}), \eprint{math.DG/9508013}.
\bibAnnoteFile{Xu97}

\bibitem[Mac87]{Mack87}
\textsc{K.~C.~H. Mackenzie}: \emph{{L}ie groupoids and {L}ie algebroids in
  differential geometry}, \emph{Lect. Notes Series}, vol. 124, London
  Mathematical Society, Cambridge University Press (\textbf{1987}).
\bibAnnoteFile{Mack87}

\bibitem[MS09]{Str09}
\textsc{C.~Mayer} and \textsc{T.~Strobl}: \emph{{L}ie algebroid {Y}ang--{M}ills
  with matter fields}, J.Geom.Phys., vol.~59, pp. 1613--1623 (\textbf{2009}),
  \doi{10.1016/j.geomphys.2009.07.018}, \eprint{0908.3161}.
\bibAnnoteFile{Str09}

\bibitem[MZ12]{MZ12}
\textsc{R.~Mehta} and \textsc{M.~Zambon}: \emph{{$L_\infty$}-algebra actions},
  Differential Geometry and its Applications, vol.~30, pp. 576--587
  (\textbf{2012}), \eprint{1202.2607}.
\bibAnnoteFile{MZ12}

\bibitem[Roy99]{Royt99}
\textsc{D.~Roytenberg}: \emph{{C}ourant Algebroids, derived brackets, and even
  symplectic supermanifolds}, Ph.D. thesis, University of California, Berkley
  (\textbf{1999}), \eprint{math.DG/9910078}.
\bibAnnoteFile{Royt99}

\bibitem[Roy01]{Royt02}
---{}---{}--- \emph{On the structure of graded symplectic supermanifolds and
  {C}ourant algebroids}, in \emph{Workshop on Quantization, Deformations, and
  New Homological and Categorical Methods in Mathematical Physics}, edited by
  \textsc{T.~Voronov}, Contemp. Math., pp. 169--185, AMS (\textbf{2001}),
  \eprint{math.SG/0203110}.
\bibAnnoteFile{Royt02}

\bibitem[Roy07]{Royt06}
---{}---{}--- \emph{{AKSZ-BV} formalism and {C}ourant algebroid induced
  topological field theories}, Lett. Math. Phys., pp. 143--159 (\textbf{2007}),
  \eprint{hep-th/0608150}.
\bibAnnoteFile{Royt06}

\bibitem[RW98]{Royt98}
\textsc{D.~Roytenberg} and \textsc{A.~Weinstein}: \emph{{C}ourant algebroids
  and strongly homotopy {L}ie algebras}, Lett. Math. Phys., vol.~46(1), pp.
  81--93 (\textbf{1998}), ISSN 0377-9017, \doi{10.1023/A:1007452512084},
  \eprint{math.QA/9802118}.
\bibAnnoteFile{Royt98}

\bibitem[San07]{San07}
\textsc{J.~A. Sanders}: \emph{An introduction to {L}eibniz algebra cohomology}
  (\textbf{last modified 2007}),  \urlprefix\href{http://www.scholarpedia.org/article/An_introduction_to_Leibniz_algebra_cohomology}{www.scholarpedia.org/article/An\_introduction\_to\_Leibniz\_algebra\_}
  cohomology. (reference no longer available)
\bibAnnoteFile{San07}

\bibitem[Sch93]{Sch93}
\textsc{A.~S. Schwarz}: \emph{Geometry of {B}atalin-{V}ilkovisky quantization},
  Commun.Math.Phys., vol. 155, pp. 249--260 (\textbf{1993}),
  \doi{10.1007/BF02097392}, \eprint{hep-th/9205088}.
\bibAnnoteFile{Sch93}

\bibitem[SS]{GSS14}
\textsc{V.~Salnikov} and \textsc{T.~Strobl}: \emph{New methods for gauging},
  final preparations.
\bibAnnoteFile{GSS14}

\bibitem[SS94]{Str94}
\textsc{P.~Schaller} and \textsc{T.~Strobl}: \emph{{P}oisson sigma-models: A
  generalization of 2-d gravity {Y}ang-{M}ills systems} (\textbf{1994}),
  \eprint{hep-th/9411163}.
\bibAnnoteFile{Str94}

\bibitem[SS13]{SaS13}
\textsc{V.~Salnikov} and \textsc{T.~Strobl}: \emph{Dirac sigma models from
  gauging}, JHEP, vol. 1311, p. 110 (\textbf{2013}),
  \doi{10.1007/JHEP11(2013)110}.
\bibAnnoteFile{SaS13}

\bibitem[SSW11]{SSW11}
\textsc{H.~Samtleben}, \textsc{E.~Sezgin} and \textsc{R.~Wimmer}:
  \emph{{$(1,0)$} superconformal models in six dimensions}, JHEP, vol. 1112, p.
  062 (\textbf{2011}), \doi{10.1007/JHEP12(2011)062}, \eprint{1108.4060}.
\bibAnnoteFile{SSW11}

\bibitem[Str99]{Strobl:1999zz}
\textsc{T.~Strobl}: \emph{{Target superspace in 2-D dilatonic supergravity}},
  Phys.Lett., vol. B460, pp. 87--93 (\textbf{1999}),
  \doi{10.1016/S0370-2693(99)00649-8}, \eprint{hep-th/9906230}.
\bibAnnoteFile{Strobl:1999zz}

\bibitem[Str04]{Str04b}
---{}---{}--- \emph{Algebroid {Y}ang-{M}ills theories}, Phys. Rev. Lett.,
  vol.~93, p. 211\,601 (\textbf{2004}), \eprint{hep-th/0406215}.
\bibAnnoteFile{Str04b}

\bibitem[Sul77]{Sull77b}
\textsc{D.~Sullivan}: \emph{Infinitesimal computations in topology}, Publ.Math.
  de l'IHES, vol.~47, pp. 269--331 (\textbf{1977}).
\bibAnnoteFile{Sull77b}

\bibitem[SW05]{SaWei:2005bp}
\textsc{H.~Samtleben} and \textsc{M.~Weidner}: \emph{The maximal {$D=7$}
  supergravities}, Nucl.Phys., vol. B725, pp. 383--419 (\textbf{2005}),
  \doi{10.1016/j.nuclphysb.2005.07.028}, \eprint{hep-th/0506237}.
\bibAnnoteFile{SaWei:2005bp}

\bibitem[Uch02]{Uchi02}
\textsc{K.~Uchino}: \emph{Remarks on the definition of a {C}ourant algebroid}
  (\textbf{2002}), \eprint{math.DG/0204010}.
\bibAnnoteFile{Uchi02}

\bibitem[Uri]{Uri13}
\textsc{B.~Uribe}: \emph{Group actions on dg-manifolds and exact {C}ourant
  algebroids}, Comm. Math.Phys., \eprint{1010.5413}.
\bibAnnoteFile{Uri13}

\bibitem[Va{\u\i}97]{Vai97}
\textsc{A.~Y. Va{\u\i}ntrob}: \emph{{L}ie algebroids and homological vector
  fields}, Uspekhi Mat. Nauk, vol.~52(2(314)), pp. 161--162 (\textbf{1997}),
  ISSN 0042-1316, \doi{10.1070/RM1997v052n02ABEH001802}.
\bibAnnoteFile{Vai97}

\bibitem[Vin90]{Vino90}
\textsc{A.~M. Vinogradov}: \emph{The union of the {S}chouten and {N}ijenhuis
  brackets, cohomology, and superdifferential operators}, Mat. Zametki,
  vol.~47(6), pp. 138--140 (\textbf{1990}), ISSN 0025-567X.
\bibAnnoteFile{Vino90}

\bibitem[Vor05]{Vor05}
\textsc{T.~Voronov}: \emph{Higher derived brackets and homotopy algebras}, J.
  Pure Appl. Algebra, vol. 202(1-3), pp. 133--153 (\textbf{2005}),
  \eprint{math.QA/0304038}.
\bibAnnoteFile{Vor05}

\bibitem[\v{S}98]{SevLett}
\textsc{P.~\v{S}evera}: \emph{concerning {C}ourant algebroids} (\textbf{1998}),
  letters from \v{S}evera to {W}einstein,
  \urlprefix\url{http://sophia.dtp.fmph.uniba.sk/~severa/letters/}.
\bibAnnoteFile{SevLett}

\bibitem[WS05]{SaWit:2005hv}
\textsc{B.~de~Wit} and \textsc{H.~Samtleben}: \emph{Gauged maximal
  supergravities and hierarchies of non-abelian vector-tensor systems},
  Fortsch.Phys., vol.~53, pp. 442--449 (\textbf{2005}),
  \doi{10.1002/prop.200510202}, \eprint{hep-th/0501243}.
\bibAnnoteFile{SaWit:2005hv}

\end{thebibliography}
\newcommand{\etalchar}[1]{$^{#1}$}

\end{document}